\documentclass{amsart}[12pt]
\usepackage{amsmath}
\usepackage{amssymb}
\usepackage{mathrsfs}
\usepackage{amsfonts}
\usepackage{graphicx}
\usepackage{tikz}
\usepackage{texdraw}
\usepackage{graphpap}
\usepackage{enumitem}
\usepackage{chngcntr}
\usepackage{wasysym}
\usepackage{color}
\usepackage{tikz}

\usetikzlibrary{decorations.markings}
\usetikzlibrary{spy}

\newcommand{\boundellipse}[3]
{(#1) ellipse (#2 and #3)
}

\theoremstyle{plain}
\newtheorem{thm}{Theorem}[section]

\newtheorem{cor}[thm]{Corollary}
\newtheorem{lem}[thm]{Lemma}

\newtheorem{prop}[thm]{Proposition}

\theoremstyle{remark}

\newtheorem*{rem}{Remark}

\newtheorem*{ex}{Example}

\numberwithin{equation}{section}

\newcommand{\I}{\mathrm{I}}
\newcommand{\II}{\mathrm{II}}
\newcommand{\III}{\mathrm{III}}
\newcommand{\normal}{{\tt{n}}}

\newcommand{\SH}{{\mathrm{SH}}}
\newcommand{\Ext}{\operatorname{Ext}}

\newcommand{\Cl}{\operatorname{cl}}

\newcommand{\T}{{\mathbb T}}

\newcommand{\mun}{\mu}

\newcommand{\calH}{{\mathscr H}}

\newcommand{\Sz}{E_{\mathrm{sz}}}
\newcommand{\Expect}{{\mathbb{E}}}
\newcommand{\Prob}{{\mathbb{P}}}

\newcommand{\calS}{{\mathscr S}}
\newcommand{\calQ}{{\mathscr Q}}

\newcommand{\calW}{{\mathscr W}}
\newcommand{\vt}{\vartheta}

\newcommand{\Int}{\operatorname{Int}}

\newcommand{\power}{N}
\newcommand{\gsz}{\gamma_{\mathrm{sz}}}

\newcommand{\R}{{\mathbb R}}
\renewcommand{\T}{{\mathbb T}}

\newcommand{\D}{{\mathbb D}}
\newcommand{\spann}{\operatorname{span}}

\newcommand{\C}{{\mathbb C}}

\newcommand{\vro}{{\varrho}}

\newcommand{\eps}{{\varepsilon}}

\newcommand{\re}{\operatorname{Re}}
\newcommand{\im}{\operatorname{Im}}

\newcommand{\degree}{\operatorname{degree}}

\renewcommand{\d}{{\partial}}
\newcommand{\dbar}{\bar{\partial}}
\newcommand{\1}{\mathbf{1}}

\newcommand{\dist}{\operatorname{dist}}
\newcommand{\supp}{\operatorname{supp}}

\makeatletter
\newcommand*\bigcdot{\mathpalette\bigcdot@{.5}}
\newcommand*\bigcdot@[2]{\mathbin{\vcenter{\hbox{\scalebox{#2}{$\m@th#1\bullet$}}}}}
\makeatother

\begin{document}

\title[Szeg\H{o} type asymptotics]{Szeg\H{o} type asymptotics for the reproducing kernel in spaces of full-plane weighted polynomials}

\author{Yacin Ameur}
\address{Yacin Ameur\\
Department of Mathematics\\
Lund University\\
22100 Lund, Sweden}
\email{ Yacin.Ameur@math.lu.se}

\author{Joakim Cronvall}
\address{Joakim Cronvall\\
Department of Mathematics\\
Lund University\\
22100 Lund, Sweden}
\email{Joakim.Cronvall@math.lu.se}



\begin{abstract} Consider the subspace $\calW_n$ of $L^2(\C,dA)$ consisting of all weighted polynomials
$W(z)=P(z)\cdot e^{-\frac 12nQ(z)},$ where $P(z)$ is a holomorphic polynomial of degree at most $n-1$, $Q(z)=Q(z,\bar{z})$ is a fixed, real-valued function called the ``external potential'', and $dA=\tfrac 1 {2\pi i}\, d\bar{z}\wedge dz$ is normalized Lebesgue measure in the complex plane $\C$.

We study large $n$ asymptotics for the reproducing kernel $K_n(z,w)$ of $\calW_n$;
this depends crucially on the position of the points $z$ and $w$ relative to the droplet $S$, i.e., the support of Frostman's equilibrium measure in external potential $Q$.
We mainly focus on the case when both $z$ and $w$ are in or near the component $U$ of $\hat{\C}\setminus S$ containing $\infty$, leaving aside such cases which are at this point well-understood.

 For the Ginibre kernel, corresponding to $Q=|z|^2$, we find an asymptotic formula after examination of
classical work due to G.~
Szeg\H{o}. Properly interpreted, the formula
turns out to
generalize to a large class of potentials $Q(z)$; this is what we call ``Szeg\H{o} type asymptotics''.
 Our derivation in the general case uses the theory of approximate full-plane orthogonal polynomials instigated by Hedenmalm and Wennman, but with nontrivial additions, notably a technique involving ``tail-kernel approximation'' and summing by parts. 
 
 In the off-diagonal case $z\ne w$
 when both $z$ and $w$ are on the boundary $\d U$, we obtain that up to unimportant factors (cocycles) the correlations obey the asymptotic
$$K_n(z,w)\sim\sqrt{2\pi n}\,\Delta Q(z)^{\frac 14}\,\Delta Q(w)^{\frac 1 4}\,S(z,w)$$ where
$S(z,w)$ is the Szeg\H{o} kernel, i.e., the reproducing kernel for the Hardy space $H^2_0(U)$ of analytic functions on $U$ vanishing at infinity, equipped with the norm
of $L^2(\d U,|dz|)$. 

Among other things, this gives a rigorous description of the slow decay of correlations at the boundary, which was predicted by Forrester and Jancovici in 1996, in the context of elliptic Ginibre ensembles.
 

\end{abstract}

\maketitle

\section{Introduction}

\subsection{The Ginibre ensemble} Recall that the standard (complex) Ginibre ensemble \cite{Fo,G,HKPV,M,RV} is the determinantal point-process $\{z_j\}_1^n$ in the complex plane $\C$ with kernel
\begin{equation}\label{gink}K_n(z,w)=n\sum_{j=0}^{n-1}\frac {(nz\bar{w})^j}{j!}e^{-\frac 12n|z|^2-\frac 12n|w|^2}.\end{equation}
To arrive at this kernel, we are prompted to equip $\C$ with the background measure
$$dA=\frac 1 {2\pi i}\, d\bar{z}\wedge dz=\frac 1 \pi\, dxdy,\qquad (z=x+iy).$$

The law of $\{z_j\}_1^n$ is the Gibbs measure
\begin{equation}\label{gibb0}d\Prob_n(z_1,\ldots,z_n)=\frac 1 {n!} \det(K_n(z_i,z_j))_{i,j=1}^n\, dA_n(z_1,\ldots,z_n),\end{equation}
where $dA_n=(dA)^{\otimes n}$ is the normalized Lebesgue measure on $\C^n$. (The combinatorial factor $1/n!$ accounts for the
fact that elements $(z_j)_1^n\in\C^n$ are ordered sequences, while configurations
$\{z_j\}_1^n$ are unordered.)

The expected number of particles which fall in a given Borel set $E$ is
$$\Expect_n(\#(\{z_j\}_1^n\cap E))=\int_E K_n(z,z)\, dA(z),$$
and if $f(z_1,\ldots,z_k)$ is a compactly supported Borel function on $\C^k$ where $k\le n$, then
$$\Expect_n(f(z_1,\ldots,z_k))=\frac {(n-k)!}{n!}\int_{\C^k}fR_{n,k}\, dA_k,$$
where the $k$-point function $R_{n,k}(w_1,\ldots,w_k)=\det(K_n(w_i,w_j))_{i,j=1}^k$. We reserve the notation $$R_n(z)=R_{n,1}(z)=K_n(z,z)$$ for the $1$-point function.

The circular law (e.g. \cite{B,G})
states that $\tfrac 1 nR_n(z)$ converges as $n\to\infty$ to the characteristic function $\1_S(z)$, where $S$ (the droplet) is the closed unit disc
$\{|z|\le 1\}$.
More refined asymptotic estimates 
may be found in \cite{A2,AKM,BM,BG,ES,FH,HW,RV}, for example.


\smallskip

We shall here study the case when $|z\bar{w}-1|\ge \eta$ for some $\eta>0$ and deduce asymptotics for $K_n(z,w)$ using techniques which hark back to Szeg\H{o}'s work \cite{Sz}
on the distribution of zeros of partial sums of the Taylor series of the exponential function. With a suitable interpretation, the asymptotic turns out generalize to to a large class of random normal matrix ensembles. In addition we shall find that the so-called \textit{Szeg\H{o} kernel} emerges in the off-diagonal boundary asymptotics. 
For those reasons we shall refer to a group of asymptotic results below as ``Szeg\H{o} type''.

\smallskip

The complete asymptotic picture of \eqref{gink} is intimately connected with
the \textit{Szeg\H{o} curve}
\begin{equation}\label{Jordan}\gsz=\{z\in \C\,;\, |z|\le 1,\,|ze^{1-z}|=1\}.\end{equation}
We define the \emph{exterior Szeg\H{o} domain} $\Sz$ to be the
unbounded component of $\C\setminus \gsz$, i.e.,
$$\Sz=\Ext\gsz.$$ (See Figure \ref{FigSz1}.)
\begin{figure}[ht]
\begin{center}
\includegraphics[width=0.5\textwidth]{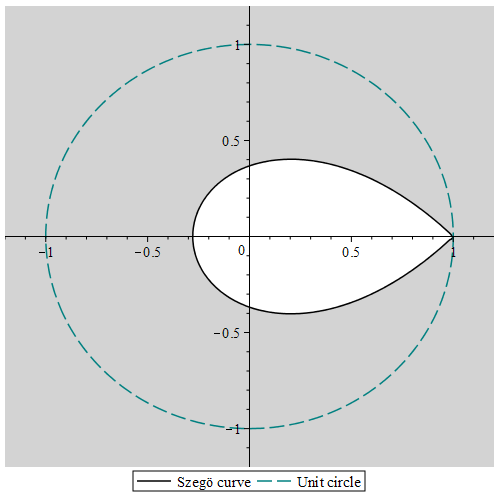}
\end{center}
\caption{The exterior Szeg\H{o} domain $\Sz$ in grey.}
\label{FigSz1}
\end{figure}

\subsubsection{Szeg\H{o} type asymptotics for the Ginibre kernel}

Three principal cases emerge, depending on the location of the product $z\bar{w}$.
\begin{enumerate}[label=(\roman*)]
\item \label{bula} If $z\bar{w}\in \C\setminus (\Sz\cup\{1\})$ we have \textit{bulk type asymptotic} in the sense that $$K_n(z,w)=ne^{nz\bar{w}-\frac n 2 |z|^2-\frac n 2 |w|^2}\cdot (1+O(n^{-\frac 1 2})).$$ (Cf. Subsection \ref{bulag} for more about this.)
\item \label{erf} If $z\bar{w}$ is in a microscopic neighbourhood of $z\bar{w}=1$, then \eqref{gink} has a well-understood \textit{error-function asymptotic} given in \cite[Subsection 2.2]{AKM}. Further results in this direction can be found in \cite{BM,HW,TV}, for example.
\item \label{extern} If $z\bar{w}\in\Sz$, it turns out that \eqref{gink} has a third kind of asymptotic, which we term \textit{exterior type}. 
This is our main concern in what follows, and we immediately turn our focus on it.
\end{enumerate}

\begin{thm} \label{tue}
Suppose that $z\bar{w}\in \Sz$ and let $K_n(z,w)$ be the Ginibre kernel \eqref{gink}. Then as $n\to\infty$
\begin{equation}\label{fex}
\begin{split}K_n(z,w)=\sqrt{\frac n {2\pi}}&\frac 1 {z\bar{w}-1}(z\bar{w})^ne^{n-\frac 12n|z|^2-\frac 
 12n|w|^2}\\
&\times (1+\frac 1 n \rho_1(z\bar{w})+\frac 1 {n^2} \rho_2(z\bar{w})+\cdots+\frac 1 {n^{k}}\rho_k(z\bar{w})+O(n^{-k-1})).\\
\end{split}
\end{equation}

The $O$-constant
is uniform provided that $\zeta=z\bar{w}$ remains in a compact subset of $\Sz$; the correction term $\rho_j(\zeta)$
is a rational function having a pole of order $2j$ at $\zeta=1$ and no other poles in the extended complex plane $\hat{\C}=\C\cup\{\infty\}$;
the first one is given by
$$\rho_1(\zeta)=-\frac 1 {12}-\frac {\zeta}{(\zeta-1)^2},$$
and the higher $\rho_j(\zeta)$ can be computed by a recursive procedure based on \eqref{Tri}, \eqref{recursion} below.
\end{thm}

In the case when $\zeta=z\bar{w}$ belongs to the sector $|\arg(\zeta-1)|<\tfrac {3\pi}4$, the result can alternatively be deduced by writing the kernel as a product involving an incomplete gamma-function and appealing to an asymptotic result
 due to Tricomi \cite{Tr}. Our present approach (found independently) is quite different and 
 has the advantage of leading to the precise domain $\Sz$ where the same asymptotic formula applies.
See Subsection \ref{rawo} for further details. 


\smallskip

For $k=0$, Theorem \ref{tue} implies that
\begin{equation}\label{sd}K_n(z,w)=\frac {\sqrt{n}}{\sqrt{2\pi}}\frac {1}{z\bar{w}-1}\cdot (z\bar{w})^n e^{n-\frac 12n|z|^2-\frac 12n|w|^2}\cdot  (1+o(1)),\qquad (z\bar{w}\in\Sz).\end{equation}

Now assume that both $z$ and $w$ are on the unit circle
$\T=\d S=\{|z|=1\}.$ In this case the function $c_n(z,w)=z^n\bar{w}^n=\frac {z^n}{w^n}$ is a \textit{cocycle}, which may be canceled from the kernel \eqref{sd} without changing the
value of the determinant \eqref{gibb0}. Therefore, using the symbol ``$\sim$'' to mean ``up to cocycles'', \eqref{sd} implies
\begin{equation}\label{offd}K_n(z,w)\sim \sqrt{2\pi n} \cdot S(z,w)\cdot (1+o(1)),\qquad (z,w\in\T),\end{equation}
where $S(z,w)$ is the (exterior)
Szeg\H{o} kernel
\begin{equation}\label{cauch}S(z,w)=\frac 1 {2\pi}\frac 1 {z\bar{w}-1}.\end{equation}

Let $\D_e=\{|z|>1\}\cup\{\infty\}$ be the exterior disc and $d\theta=|dz|$ the arclength measure on $\T$.
Consider the Hardy space $H^2_0(\D_e)$ of analytic functions $f:\D_e\to \C$ which vanish at infinity, equipped with the norm of $L^2(\T,d\theta)$. The kernel $S(z,w)$
is the reproducing kernel of $H^2_0(\D_e)$.

\smallskip

Let us now consider the \textit{Berezin kernel} rooted at a point $z\in\C$,
\begin{equation}\label{berk}B_n(z,w)=\frac {|K_n(z,w)|^2} {K_n(z,z)}.\end{equation}

It is a household fact that if $z\in\T$, then $K_n(z,z)=\frac 1 2 \, n\cdot (1+o(1)).$

(Proof:
$K_n(z,z)=n\cdot \Prob(\{X_n\le n\})$
where $X_n$ is a Poisson random variable with intensity $n$. Since $(X_n-n)/\sqrt{n}$ converges in distribution to a standard normal, $\Prob(\{X_n\le n\})\to\frac 12$
as $n\to\infty$.)

It follows that if $z,w\in\T$ and $z\ne w$, then
\begin{equation}\label{decay}B_n(z,w)=\frac 1 {\pi}\frac 1 {|z-w|^2}\cdot (1+o(1)).\end{equation}

It is interesting to compare \eqref{decay}
 with the case when $z\in\Int S$; then $B_n(z,w)$ decays exponentially in $n$
 by the heat-kernel estimate
 in Subsection \ref{bulag}. (Alternatively, by results in \cite{AHM1}.)

\smallskip

The moral is that, in the off-diagonal case $z\ne w$, the magnitude of $K_n(z,w)$ is exceptionally large when both $z$ and $w$ are on the boundary $\T$, compared with any other kind of configuration. (Some heuristic explanations for this kind of behaviour are sketched below in Subsection \ref{sec13}.)

\subsubsection{Gaussian convergence of Berezin measures}
It is natural to
regard the Berezin kernel \eqref{berk} 
as the probability density of the
\textit{Berezin measure} $\mu_{n,z}$ rooted at $z$,
\begin{equation}\label{berm}d\mu_{n,z}(w)=B_n(z,w)\, dA(w).\end{equation}

\begin{figure}[t]
\begin{center}
\includegraphics[width=0.5\textwidth]{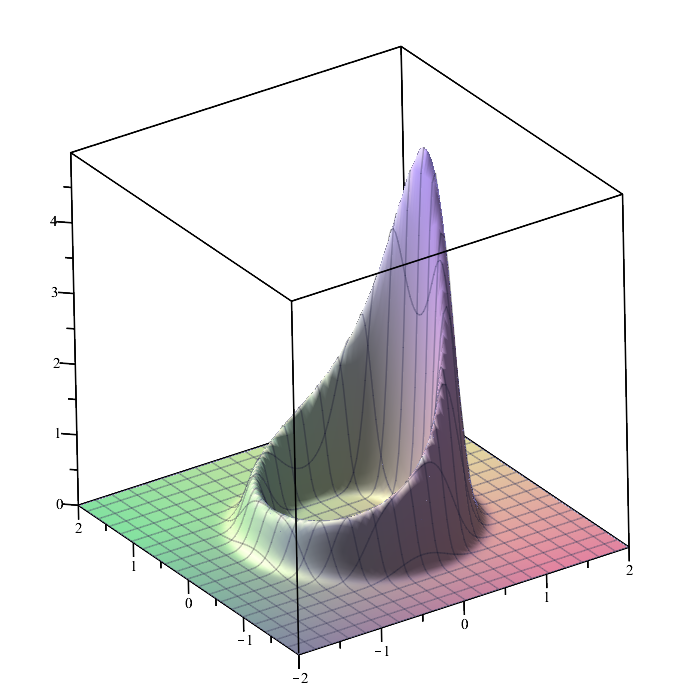}
\end{center}
\caption{Plot of the Berezin kernel for the Ginibre ensemble, $w\mapsto B_n(z,w)$ for $n=20$ and $z=2$.}
\label{fig2}
\end{figure}


\smallskip

It is shown in \cite[Section 9]{AHM1} that if $z\in\D_e$ then the measures $\mu_{n,z}$ converge weakly to the harmonic measure relative to $\D_e$ evaluated at $z$,
$d\omega_z(\theta)=P_z(\theta)\, d\theta$
where $P_z(\theta)$ is the (exterior) Poisson kernel
\begin{equation}\label{poisson}P_z(\theta)=\frac 1 {2\pi}\frac {|z|^2-1}{|z-e^{i\theta}|^2}.\end{equation}

We will denote by $d\gamma_n$ the following Gaussian probability measure on $\R$ 
$$d\gamma_n(\ell)=\frac {2\sqrt{n}} {\sqrt{2\pi}}e^{-2n \ell^2}\, d\ell,\qquad (\ell\in\R).$$

(Here and throughout, ``$d\ell$'' is Lebesgue measure on $\R$.)

It is also convenient to represent points $w$ close to  $\T$ in ``polar coordinates''
\begin{equation}\label{frun}w=e^{i\theta}\cdot (1+\ell),\qquad (\theta\in [0,2\pi),\,\ell\in\R).\end{equation}

As a consequence of the kernel asymptotic in Theorem \ref{tue}, we obtain the following result.

\begin{cor} \label{gharm} Fix a point $z\in\D_e$ and an arbitrary sequence $(c_n)_1^\infty$ of positive numbers with 
$$nc_n^2\to \infty,\qquad  \text{and}\qquad nc_n^3\to 0,\qquad  \text{as}\qquad  n\to\infty.$$
Then for $w$ in the form \eqref{frun},
we have the Gaussian approximation
\begin{equation}\label{bobb}d\mu_{n,z}(w)=(1+o(1))\cdot P_z(\theta)\cdot \gamma_n(\ell)\, d\theta d\ell,\end{equation}
where $o(1)\to 0$ as $n\to\infty$ uniformly for $w$ in the belt $N(\T,c_n):=\{|\ell|\le c_n\}.$
\end{cor}

Here and henceforth, a sequence of functions $f_n:E_n\to\C$ is said to \emph{converge uniformly to $0$} if there is a sequence $\epsilon_n\to 0$
such that $|f_n|<\epsilon_n$ on $E_n$ for each $n$.

\begin{rem} The approximating measures
$d\tilde{\mu}_{n,z}(\theta,\ell)=P_z(\theta)\cdot \gamma_n(\ell)\, d\theta d\ell$ are probability measures on $\T\times \R$ which
assign a mass of $O(e^{-nc_n^{\,2}})$ to the complement of $N(\T,c_n)$. The condition that $nc_n^2\to\infty$ insures that
$\mu_{n,z}-\tilde{\mu}_{n,z}\to 0$ in the sense of measures on $\C$.
This justifies the Gaussian approximation picture, as exemplified in Figure \ref{fig2}.
\end{rem}

\subsection{Notation and potential theoretic setup} \label{backgr} In order to generalize beyond the Ginibre ensemble, we require some notions
from potential theory; cf. \cite{ST}.

\smallskip

We are about to write down rather a dry list of definitions and generally useful facts; the reader may skim it to his advantage. 

\smallskip

We begin by fixing a lower semicontinuous function $Q:\C\to \R\cup\{+\infty\}$ which we call the \textit{external potential}. (The Ginibre ensemble corresponds to
the special choice $Q(z)=|z|^2$.)

We assume that $Q$ is finite on some set of positive capacity and that
\begin{equation}\label{g1}\liminf_{z\to\infty}\frac {Q(z)}{\log|z|^2}> 1.\end{equation}
Further conditions are given below.

Given a compactly supported Borel probability measure $\mu$, we define its $Q$-energy by
\begin{equation}\label{wen}I_Q[\mu]=\iint_{\C^2}\log\frac 1 {|z-w|}\, d\mu(z)\, d\mu(w)+\mu(Q),\end{equation}
where $\mu(Q)$ is short for $\int Q\, d\mu$.

It is well-known \cite{ST} that there exists a unique equilibrium measure $\sigma=\sigma_Q$ of unit mass which minimizes $I_Q[\mu]$
over all compactly supported Borel probability measures on $\C$. The support of $\sigma$ is denoted by
$S=S[Q]=\supp\sigma,$
and is called the droplet.
Perhaps even more central to this work is the \textit{exterior component containing $\infty$},
$$U=U[Q]:=\textrm{``component of $\hat{\C}\setminus S$ which contains $\infty$''}.$$ The boundary of $U$ is called the
outer boundary of $S$ and is
written
$\Gamma=\d U.$

We now introduce four standing assumptions \ref{1st}-\ref{4th}.

\begin{enumerate}[label=(\arabic*)]
\item \label{1st} $S$ is connected and $Q$ is $C^2$ smooth in a neighbourhood of $S$ and real-analytic in a neighbourhood of $\Gamma$.
\end{enumerate}

This assumption has the consequence that the equilibrium measure $\sigma$ is absolutely continuous and has the structure
$d\sigma=\1_S\cdot\Delta Q\, dA,$
where $\Delta=\d\dbar=\tfrac 1 4 (\tfrac {\d^2}{\d x^2}+\tfrac {\d^2}{\d y^2})$
is the normalized Laplacian. 

We are guaranteed that $\Delta Q\ge 0$ on $S$; we will require a bit more:

\begin{enumerate}[label=(\arabic*)]
\setcounter{enumi}{1}
\item  $\Delta Q(z)>0$ for all $z\in\Gamma$.
\end{enumerate}

Let $\SH_1(Q)$ denote the class of all subharmonic functions $s(z)$ on $\C$ which
satisfy $s\le Q$ on $\C$ and $s(z)\le \log|z|^2+O(1)$ as $z\to\infty$. We define the \textit{obstacle function} $\check{Q}(z)$ to be the envelope
\begin{equation}\label{obs1}\check{Q}(z)=\sup\{s(z)\,;\,s\in\SH_1(Q)\}.\end{equation}

Clearly $\check{Q}(z)$ is subharmonic and grows as $\log|z|^2+O(1)$ as $z\to\infty$. Furthermore, $\check{Q}(z)$ is $C^{1,1}$-smooth on $\C$, i.e., its gradient
is Lipschitz continuous.

Denote by $S^*=\{z\, ;\, Q(z)=\check{Q}(z)\}$ the coincidence set for the obstacle problem. In general we have the inclusion $S\subset S^*$
and if $p$ is a point of $S^*\setminus S$ then there is a neighbourhood $N$ of $p$ such that $\sigma(N)=0$.
We impose:

\begin{enumerate}[label=(\arabic*)]
\setcounter{enumi}{2}
\item $U\cap S^*$ is empty.
\end{enumerate}

Write
$\chi:\D_e\to U$ for the unique conformal mapping normalized by the conditions
$\chi(\infty)=\infty$ and $\chi'(\infty)>0$.
A fundamental theorem due to Sakai \cite{Sa} implies that
 $\chi$ extends analytically across $\Gamma$ to some neighbourhood of the closure $\Cl \D_e$. (Details about this application of Sakai's theory
 are found in Subsection \ref{saksek} below.)
Thus $\Gamma$ is a Jordan curve consisting of analytic arcs and possibly finitely many singular points where the arcs meet.
We shall assume:

\begin{enumerate}[label=(\arabic*)]
\setcounter{enumi}{3}
\item \label{4th} $\Gamma$ is non-singular, i.e., $\chi$ extends across $\T$ to a conformal mapping
 from a neighbourhood of $\Cl \D_e$ to a neighbourhood of $\Cl U$.
\end{enumerate}

In the following we denote by
$\phi=\chi^{-1}$
the inverse map, taking a neighbourhood of $\Cl U$ conformally onto a neighbourhood of
$\Cl \D_e$, and obeying $\phi(\infty)=\infty$ and $\phi'(\infty)>0$.
We denote by $\sqrt{\phi'}$ the branch of the square-root which is positive at infinity.

\subsubsection*{Class of admissible potentials} Except when otherwise is explicitly stated, all external potentials $Q$ used below are lower semicontinuous functions $\C\to\R\cup\{+\infty\}$, finite on some
set of positive capacity, satisfying the growth condition \eqref{g1} and the four conditions \ref{1st}-\ref{4th}.

\subsubsection*{Auxiliary functions} For a given admissible potential $Q$, we consider the holomorphic functions $\calQ(z)$ and $\calH(z)$ on a neighbourhood of $\Cl U$ which obey
\begin{equation}\label{aux}\re \calQ(z)=Q(z),\qquad \re \calH(z)=\log\sqrt{\Delta Q(z)},\qquad \text{when}\qquad z\in\Gamma,\end{equation}
and which satisfy $\im\calQ(\infty)=\im\calH(\infty)=0$.

We shall also frequently use the function $V$ given by
\begin{equation}V=\text{``the harmonic continuation of the restriction }\check{Q}\Big|_U\text{ across the analytic curve }\Gamma.\text{''}\end{equation}

It is useful to note the identity
\begin{equation}\label{19}V=\re\calQ+\log|\phi|^2\qquad \text{on}\qquad \C\setminus K,\end{equation}
where $K$ is a fixed compact subset $K$ of the bounded component $\Int \Gamma$ of $\C\setminus \Gamma$.

To realize \eqref{19} it suffices to note that the harmonic functions on the left and right hand sides
agree on $\Gamma$ and grow like $\log|z|^2+O(1)$ near infinity, so \eqref{19} follows by the strong version of the maximum principle (e.g. \cite{GM}).

\subsubsection*{The Szeg\H{o} kernel} Let $H_0^2(U)$ be the Hardy space of holomorphic functions $f:U\to \C$ which vanish at infinity and are square-integrable with respect to arclength: $\int_\Gamma|f(z)|^2\,|dz|<\infty$.
We equip $H^2_0(U)$ with the inner product of $L^2(\Gamma,|dz|)$ and observe that the functions $\psi_j(z)=\frac 1 {\sqrt{2\pi}}\frac {\sqrt{\phi'(z)}}{\phi(z)^j}$ ($j\ge 1$) form an orthonormal basis for
$H^2_0(U)$. The reproducing kernel for $H^2_0(U)$ is thus
\begin{equation}\label{szek}S(z,w)=\sum_{j=1}^{\infty}\psi_j(z)\overline{\psi_j(w)}=\frac 1 {2\pi}\frac {\sqrt{\phi'(z)}\overline{\sqrt{\phi'(w)}}}{\phi(z)\overline{\phi(w)}-1}.\end{equation} 
 We shall refer to $S(z,w)$ as the \textit{Szeg\H{o} kernel} associated with $\Gamma$ (or $U$).

Many interesting properties of the Szeg\H{o} kernel can be found in Garabedian's thesis work \cite{Gar} and in the book \cite{Bell}. 
A different natural way to define $H^p$-spaces over general domains is discussed in e.g.~ \cite[Section 10]{D}.


\subsubsection*{The reproducing kernel}
Let $Q$ be an admissible potential and consider the space $\calW_n=\calW_n(Q)$ consisting of all weighted polynomials $W$ of the form
$$W(z)=P(z)\cdot e^{-\frac 12nQ(z)},$$
where $P$ is a holomorphic polynomial of degree at most $n-1$. We equip $\calW_n$ with the usual norm in $L^2(\C,dA)$ and denote by $K_n(z,w)$
the corresponding reproducing kernel.

We follow standard conventions concerning reproducing kernels \cite{Ar};
we write $K_{n,z}(w)=K_n(w,z)$ and note that the element $K_{n,z}\in\calW_n$ is characterized by the reproducing property:
 $$W(z)=\int_\C W\bar{K}_{n,z}\, dA$$
for all $W\in\calW_n$ and all $z\in\C$.

We shall frequently use the formula
$$K_n(z,w)=\sum_{j=0}^{n-1}W_{j,n}(z)\overline{W_{j,n}(w)},$$ where
$\{W_{j,n}\}_{j=0}^{n-1}$ is any orthonormal basis for $\calW_n$. We fix such a basis uniquely by requiring that
$W_{j,n}=P_{j,n}\cdot e^{-\frac 1 2nQ}$ where $P_{j,n}$ is of exact degree $j$ and has positive leading coefficient.

\subsubsection*{Auxiliary regions} In the sequel we write
\begin{equation}\label{deltan}\delta_n=M\sqrt{\frac {\log\log n} n},\end{equation}
where $M$ is a fixed positive constant (depending only on $Q$). The $\delta_n$-neighbourhood of a set $E$ will be denoted
$$N(E,\delta_n)=E+D(0,\delta_n),$$
where $D(a,r)=\{z\, ;\, |z-a|<r\}$ is the euclidean disc with center $a$ and radius $r$.

\subsection{Asymptotic results for admissible potentials} \label{sec13} In the following, $Q$ denotes an admissible potential in the sense of Subsection \ref{backgr}.

\subsubsection{Szeg\H{o} type asymptotics for the reproducing kernel} We have the following result; the definitions of the various ingredients are given in the preceding subsection. (In particular $N(U,\delta_n)$ denotes the $\delta_n$-neighbourhood of the exterior set $U$, cf.~\eqref{deltan}.)



\begin{thm}\label{kern1} Fix constants $\eta$ and $\beta$ with $\eta>0$ and $0<\beta<\frac 14$. Assuming that
\begin{equation}z,w\in N(U,\delta_n),\qquad \text{and}\qquad |\phi(z)\overline{\phi(w)}-1|\ge \eta,\end{equation} 
we have the asymptotic formula
\begin{equation}\label{ny}\begin{split}K_n(z,w)&=\sqrt{2\pi n}\cdot e^{\frac n 2(\calQ(z)+\overline{\calQ(w)})-\frac n 2 (Q(z)+Q(w))+\frac 1 2(\calH(z)+\overline{\calH(w)})}(\phi(z)\overline{\phi(w)})^n
\cr
&\qquad\quad \times S(z,w)\cdot (1+O(n^{-\beta})),\quad (n\to\infty).\cr
\end{split}\end{equation}

The
$O$-constant is uniform for the given set of $z$ and $w$ (depending only on the parameters $\eta,M$ and the potential $Q$).
\end{thm}

\begin{ex} When $Q=|z|^2$ we have $\calQ=1$ and $\calH=0$ while $\phi(z)=z$. We thus recover the asymptotic
formula in \eqref{sd}.
\end{ex}


In the off-diagonal case when $z,w$ are exactly on the boundary, we recognize several exact cocycles which may be cancelled from the expression \eqref{ny} without changing the statistical properties of the corresponding determinantal process. Recall that a cocycle is just a function of the form $c_n(z,w)=g_n(z)/g_n(w)$ where $g_n$ is a continuous and nonvanishing function.


\begin{cor}\label{kern2} Suppose that $z,w\in \Gamma$ and $z\ne w$.
Then
$$c_n(z,w):=(\phi(z)\overline{\phi(w)})^ne^{i\frac n 2 \im(\calQ(z)-\calQ(w))}e^{i\frac 12\im(\calH(z)-\calH(w))}$$
is a cocycle and
\begin{equation}\label{FoJa}K_n(z,w)=\sqrt{2\pi n}\cdot \Delta Q(z)^{\frac 14}\Delta Q(w)^{\frac 14}S(z,w)\cdot c_n(z,w)\cdot (1+o(n^{-\beta})).\end{equation}
\end{cor}

\begin{figure}[ht]
\begin{center}
\includegraphics[width=0.5\textwidth]{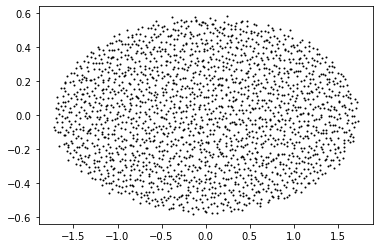}
\end{center}
\caption{A sample from an elliptic Ginibre ensemble with $a=\tfrac{2}{3}$, $b=2$ and $n=2000$. (Notation according to Subsection \ref{6.6}.)}
\label{figell}
\end{figure}

The formula \eqref{FoJa} is related to a question 
studied by Forrester and Jancovici in the paper \cite{FJ} on Coulomb gas ensembles at the edge of the droplet, in the special case of the elliptic Ginibre ensemble.  
The physical picture is that the screening cloud about a charge at the edge has
a non-zero dipole moment, which gives rise to a slow decay of the correlation function. In \cite{FJ} an argument on the physical level of rigor, based on Jancovici's linear response theory, is given, and a formula for $|K_n(z,w)|^2$ is predicted in the case when $z,w$ are on the boundary ellipse and $z\ne w$. This formula is consistent with \eqref{FoJa} in the special case of the elliptic Ginibre ensemble. 

In the recent work \cite{ADM}, the elliptic Ginibre ensemble is studied by using properties of the particular (Hermite) orthogonal polynomials which enter in that case. As a result, some more refined asymptotic results can be obtained in this case. A comparison is
found in \cite[Remark I.4]{ADM} as well as in Subsection \ref{6.6} below.


\begin{rem} It is interesting to view the slow decay of charge-charge correlations in light of the fact that fluctuations near the boundary converge to a separate Gaussian field, which is independent from the one emerging in the bulk, see \cite{AM,RV} for the case of random normal matrices; details can be found in \cite[Subsection 7.3]{AHM2}. The emergence of a separate boundary field makes it credible that a charge at the edge should correlate much stronger with other charges at the edge than with charges is the bulk, and our present results demonstrate that this expected behaviour is, in a broad sense, valid.
(One should not read too much into the above analogy; after all, fluctuations converge in a weak, distributional sense, while our present results provide different, uniform estimates, for example for the connected 2-point function $-|K_n(z,w)|^2$.) 

We refer to Forrester's recent survey article \cite{F} as a source for many other kinds of fluctuation theorems. We may recall in particular that in settings of planar $\beta$-ensembles, the two papers \cite{BBNY2,LSe} appeared almost simultaneously, suggesting two very different approaches to the question of proving Gaussian field convergence. (The case under study corresponds to $\beta=2$ and was settled in \cite{AM,RV}.)
\end{rem}



\subsubsection{Gaussian convergence of Berezin measures} Let $K_n(z,w)$ be the reproducing kernel with respect to an arbitrary admissible potential $Q$.
\begin{figure}[ht]
\begin{center}
\includegraphics[width=0.5\textwidth]{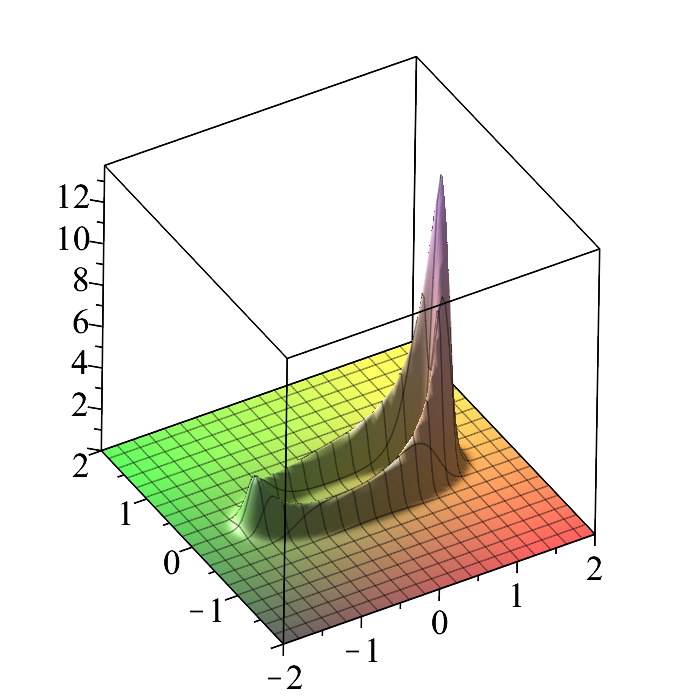}
\end{center}
\caption{The Berezin kernel $w\mapsto B_n(z,w)$ where $z=2$ and $n=20$. Here $Q$ is the elliptic Ginibre potential $Q(w)=u^2+3v^2$ where $w=u+iv$. The droplet $S$ is the elliptic disc $\tfrac 1 2 u^2+6v^2\le 1$, so $z$ belongs to the exterior component $U$
and the emergent Gaussian approximation of harmonic measure is clearly visible.}
\label{FigSzB}
\end{figure}

Naturally, we define Berezin kernels and Berezin measures by
$$B_n(z,w)=\frac {|K_n(z,w)|^2}{K_n(z,z)},\qquad
d\mu_{n,z}(w)=B_n(z,w)\, dA(w).$$
It is convenient to recall a few facts concerning these measures.

\begin{enumerate}
\item \label{p0} If $z$ is a non-degenerate bulk point (in the sense that $z\in\Int S$ and $\Delta Q(z)>0$),
then $\mu_{n,z}$ converges to the Dirac point mass $\delta_z$, whereas if $z\in U$, then
$\mu_{n,z}$ converges to the harmonic measure $\omega_z$ evaluated at $z$;
the convergence holds in the weak sense of measures on $\C$. (See \cite[Theorem 7.7.2]{AHM2}.)
\item \label{p1} If $z$ is a non-degenerate bulk-point, then
the convergence $\mu_{n,z}\to\delta_z$ is Gaussian in the sense of heat-kernel asymptotic:
$B_n(z,w)=n\Delta Q(z)\cdot e^{-n\Delta Q(z)\,|w-z|^2}\cdot (1+o(1)),$
where $o(1)\to 0$ uniformly for (say) $w\in D(z,\delta_n)$. (See e.g. \cite{AHM1}.)
\item \label{p2} If $z\in U$, then the weak convergence $\mu_{n,z}\to\omega_z$ may be combined
with an asymptotic result for the so-called root-function in \cite[Theorem 1.4.1]{HW2},
 indicating that the convergence  must in a sense be ``Gaussian''.
\end{enumerate}

We shall now state a result giving a quantitative Gaussian approximation to $\mu_{n,z}$ from which the
 convergence to harmonic measure will be directly manifest.
For this purpose we express points $w$ in some neighbourhood of $\Gamma$ as
\begin{equation}\label{coords}w=p+\ell\cdot\normal_1(p)\end{equation}
where $p=p(w)$ is a point on $\Gamma$, $\normal_1(p)$ is the unit normal to $\Gamma$ pointing outwards from $S$, and $\ell$ is a real parameter. (So $|\ell|=\dist(w,\Gamma)$ if $\ell$ is close to $0$.)

Given a point $p\in \Gamma$ we also define a Gaussian probability measure $\gamma_{p,n}$ on the real line by
\begin{equation}d\gamma_{p,n}(\ell)=\frac {\sqrt{4n\Delta Q(p)}}{\sqrt{2\pi}}e^{-2n\Delta Q(p)\ell^2}\, d\ell.\end{equation}

For a given point $z\in U$, we denote by $\omega_z$ the harmonic measure of $U$ evaluated at $z$ and consider the measure $\tilde{\mu}_{n,z}$
given in the coordinate system \eqref{coords} by
\begin{equation}\label{prodd}d\tilde{\mu}_{n,z}=d\omega_z(p)\, d\gamma_{n,p}(\ell).\end{equation}

To be more explicit, we define the \textit{Poisson kernel} $P_z(p)$ as the density of $\omega_z$ with respect to arclength $|dp|$ on $\Gamma$, i.e.,
$$d\omega_z(p)=P_z(p)\, |dp|,\qquad (p\in\Gamma).$$

Then
\begin{equation}\label{n0}d\tilde{\mu}_{n,z}(p+\ell\cdot\normal_1(p))=P_{z}(p)\frac {\sqrt{4n\Delta Q(p)}}{\sqrt{2\pi}}e^{-2n\Delta Q(p)\ell^2}\, |dp|\,d\ell.\end{equation}

For this definition to be consistent, we
fix a small neighbourhood of $\Gamma$ and define $\tilde{\mu}_{n,z}$ by \eqref{n0} in this neighbourhood and extend it by zero outside the neighbourhood. Then
$\tilde{\mu}_{n,z}$ is a sub-probability measure whose total mass quickly increases to $1$ as $n\to\infty$.


\begin{thm}\label{berth1} Suppose that $z$ is in the exterior component $U$. Then
\begin{equation}\label{berth2}\frac 1 \pi B_n(z,w)=P_{z}(p)\frac{\sqrt{4n\Delta Q(p)}}{\sqrt{2\pi}}e^{-2n\Delta Q(p)\ell^2}\cdot (1+o(1))\end{equation}
where $o(1)\to 0$ as $n\to\infty$ with uniform convergence when $w$ is in the belt $N(\Gamma,\delta_n)=\{|\ell|\le \delta_n\}.$

In other words, $\mu_{n,z}=(1+o(1))\tilde{\mu}_{n,z}$ where $o(1)\to 0$
as $n\to\infty$ in the sense of measures on $\C$ as well as in the uniform sense
of densities on $N(\Gamma,\delta_n)$.
\end{thm}

\begin{rem}
The last statement in Theorem \ref{berth1} is automatic once the uniform convergence in \eqref{berth2} is shown.
Indeed, let $\epsilon>0$ be given. It is clear from the definition \eqref{n0} that we can find $M$ and $n_0$ such that $\tilde{\mu}_{n,z}(N(\Gamma,\delta_n))>1-\epsilon$ when $n\ge n_0$.
 Then $\tilde{\mu}_{n,z}(\C\setminus N(\Gamma,\delta_n))<\epsilon$ when $n\ge n_0$. Thus $\tilde{\mu}_{n,z}|_{\C\setminus N(\Gamma,\delta_n)}\to 0$ as measures when $n\to\infty$. By the uniform convergence in \eqref{berth2} we now see that $\mu_{n,z}|_{\C\setminus N(\Gamma,\delta_n)}\to 0$, since $\mu_{n,z}$ has unit total mass.
Thus it suffices to prove the uniform convergence in \eqref{berth2}; this is done in Section \ref{Sec_Gen}.
\end{rem}


\subsubsection{Main strategy: tail-kernel approximation} An underpinning idea
is that for points $z$ and $w$ in or close to the exterior set $U$, a good knowledge of the \textit{tail kernel}
\begin{equation}\label{tail}\tilde{K}_n(z,w)=\sum_{j=n\theta_n}^{n-1}W_{j,n}(z)\overline{W_{j,n}(w)},\qquad (\theta_n:=1-\frac {\log n}{\sqrt{n}}),\end{equation}
should suffice for deciding the leading-order asymptotics of the full kernel $K_n(z,w)$.

Note that $\tilde{K}_n(z,w)$ is just the reproducing kernel for the orthogonal complement $\tilde{\calW}_n=\calW_n\ominus \calW_{k,n}$ where $\calW_{k,n}\subset\calW_n$ is the subspace consisting
of all $W=P\cdot e^{-\frac 12 nQ}$ where $P$ has degree at most $k=$ ``largest integer which is strictly less than $n\theta_n$''.

\smallskip

We shall deduce asymptotics for $\tilde{K}_n(z,w)$ using a technique based on summing by parts with the help of an approximation formula for $W_{j,n}$ found in the paper
\cite{HW}.

When this is done, some fairly straightforward estimates
for the lower degree terms (with $j\le n\theta_n$) are sufficient to show that the full kernel $K_n(z,w)$ has similar asymptotic properties as does $\tilde{K}_n(z,w)$. 

The practical execution of this strategy forms the bulk of this paper, cf.~ sections \ref{LAGR}, \ref{Sec_Gen}, and 
\ref{details}.

\subsection{The elliptic Ginibre ensemble} \label{6.6}
We now temporarily specialize to the elliptic Ginibre potential 
\begin{equation}\label{ellpot}Q(z)=ax^2 +by^2,\qquad z=x+iy\in\C,\end{equation}
where $a,b>0$. It is convenient to assume that $a<b$.

\begin{rem} In the literature on the topic it is common to restrict to potentials depending on one single ``non-Hermiticity parameter'' $\tau$ with $-1<\tau<1$
and set the parameters in \eqref{ellpot} to $a=\frac{1}{1+\tau}$ and  $b= \frac{1}{1-\tau}$, giving
$$Q(z)= \tfrac{1}{1-\tau^2} (|z|^2 -\tau \text{Re} (z^2)).$$ However, other conventions are sometimes used, e.g.~ \cite{LR} takes $a=1-\tau$ and $b=1+\tau$ while \cite{AB} fixes $a=\frac 12$ and uses $b$ as a (large) parameter. 
\end{rem}

It is easy to construct random samples with respect to the potential \eqref{ellpot}: start with two independent $n\times n$ GUE matrices $J_1$ and $J_2$ and look at the random matrix
$$X_n=\tfrac{1}{\sqrt{2a}}J_1+i\tfrac{1}{\sqrt{2b}} J_2.$$ The eigenvalues $\{z_j\}_1^n$ of $X_n$ then correspond precisely to a random sample from the determinantal $n$-point process in potential $Q$; this is what was used to produce Figure \ref{figell}.

We now recast some well-known facts about the elliptic Ginibre point-process; proofs and further details can be found
in \cite{ACV,ADM,AB} and the references there.

In terms
of the Hermite polynomials $H_j(z)=(-1)^je^{z^2}\frac {d^j}{dz^j}e^{-z^2}$, 
the correlation kernel $K_n(z,w)$ is given by
\begin{equation}\label{hkern}
K_n(z,w)=  n\sqrt{ab}\sum\limits_{j=0}^{n-1}  \frac{1}{j!} (\tfrac{1}{2}\tfrac{b-a}{b+a})^j\, H_j( \sqrt{\tfrac{nab}{b-a}} z )\, H_j( \sqrt{\tfrac{nab}{b-a}} \bar{w} )\, e^{-\tfrac{n}{2}Q(z)-\tfrac{n}{2}Q(w)}.
\end{equation}
(This formula was used to plot Figure \ref{FigSzB}.) Moreover, the droplet is the elliptic disc 
$$
S= \{ z=x+iy \,;\,  \tfrac{a^2 + ab}{2b} x^2 + \tfrac{ab+b^2}{2a}y^2 \leq 1  \},
$$
which has its major semi-axis along the real line.
The normalized conformal map $\phi$ taking $U=\hat{\C}\setminus S$ to $\D_e$ is the inverse Joukowsky map (well-known from the theory of conformal mapping \cite{Ne})
$$
\phi(z) = \tfrac{z}{2\alpha} (  1+ \sqrt{1-\tfrac{4\alpha \beta}{z^2}}   ),
$$
where $\alpha = \tfrac{1}{2} ( \sqrt{\tfrac{2b}{a^2+ab}} + \sqrt{\tfrac{2a}{b^2+ab}}    )$ and $\beta = \tfrac{1}{2} ( \sqrt{\tfrac{2b}{a^2+ab}} -\sqrt{\tfrac{2a}{b^2+ab}}    )$. (Here we use the principal branch of the square-root, so $\phi(z)\sim z/\alpha$ as $z\to\infty$ and $\phi'(\infty)=1/\alpha$.) 

Since the Laplacian $\Delta Q$ is the constant
$\frac 1 2(a+b)$, we have $\calH\equiv \frac 1 2\log(a+b)$.
Inserting these data, our Theorem \ref{kern1} (and using $\re\calQ=Q$ on $\d S$) we obtain an effective approximation formula, which is consistent with the earlier predictions due to Forrester and Jancovici \cite{FJ} as well as with more recent work due to Akemann, Duits and Molag \cite{ADM}. We now comment on these works.

In the setting of Forrester and Jancovici, the key object is $|K_n(z,w)|^2$ rather than the reproducing kernel $K_n(z,w)$ itself. Forrester and Jancovici use linear response theory and asymptotics of Hermite polynomials to predict an asymptotic formula for $|K_n(z,w)|^2$ in the off-diagonal case, when $z,w$ belong to the boundary ellipse. With some effort, their formula can be shown to be consistent with Theorem \ref{kern1} (and Corollary \ref{kern2}). Details can be found in the recent paper \cite{ADM}, see especially Remark I.4 for a comparison with our present work.

In the paper \cite{ADM}, the authors use different methods, relying on a contour integral representation of the kernel \eqref{hkern} and a saddle point analysis. Several refined results are derived there, notably \cite[Theorem I.1]{ADM}, which among other things implies 
that the exterior type asymptotics for $K_n(z,w)$ (from Theorem \ref{kern1}) persists in some fixed, $n$-independent neighbourhood of the boundary of the droplet (and away from the diagonal $z=w$). (When specialized to the Ginibre ensemble, this fact can of course be seen from Theorem \ref{tue} as well.) By contrast, Theorem \ref{kern1} only guarantees asymptotics for $K_n(z,w)$ when $z,w$ belong to the shrinking neighbourhood $N(U,\delta_n)$, of distance $\delta_n$ from the boundary. Interestingly, the asymptotic formula \cite[Theorem I.1]{ADM} extends to the case when the points $z,w$ stay away from the ``motherbody'',  i.e. the line-segment between the foci of the ellipse, and such that $|\phi(z)\overline{\phi(w)}-1|\ge \eta$ for some $\eta>0$, again see \cite[Remark I.4]{ADM}. In particular, this provides information about the transition from exterior to bulk-type asymptotics, in the elliptic Ginibre case. 


\subsection{Further results and related work} \label{rawo} A good motivation for studying the reproducing kernel $K_n(z,w)$ comes from random matrix theory, where it corresponds precisely to the ``canonical correlation kernel'', e.g. \cite{ABD,AKM,Fo,M,ST}.

If the external potential $Q$ satisfies $Q=+\infty$ on $\C\setminus \R$ we obtain Hermitian random matrix theory and Coulomb gas processes on $\R$, while if $Q$ is admissible in our present sense, we obtain normal random matrix theory and planar Coulomb gas processes.
Asymptotics for correlation kernels of normal random matrix ensembles
has been the subject of many investigations, see for example \cite{A2,AKM,HW,HW2,LR} and the references there.





\smallskip

It is noteworthy that Forrester and Honner in the paper \cite{FH} study a different problem on edge-correlations, between zeros of random polynomials $p_n(z)=\sum_{j=0}^{n-1} (j!)^{-\frac 12}
a_jz^j$ where the $a_j$ are i.i.d.~standard complex Gaussians. (The ``edge'' here is the circle $|z|=\sqrt{n}$.)


\smallskip

Szeg\H{o}'s paper \cite{Sz} concerns zeros of partial sums $S_n(z)=1+z+\cdots+\frac {z^n}{n!}$ of the Taylor series for $e^z$. 

It is not surprising that Szeg\H{o}'s results
 should have a bearing for the Ginibre ensemble, since a factor $S_{n-1}(nz\bar{w})$ enters naturally in the formula \eqref{gink}. This has been used, for instance, in the papers \cite{AHM1,HH}. The Szeg\H{o} curve \eqref{Jordan} also enters in connection with the asymptotic analysis of various orthogonal polynomials, notably such which are associated with lemniscate ensembles, see \cite{BGM,BMe,BERG,LY}, cf. also Subsection \ref{arch} below.
Szeg\H{o}'s work can also be seen as a starting point for the theory of sections of power series
of entire functions, cf.~ for instance \cite{ESV,V}.

The sum $S_{n-1}(nz\bar{w})$ also has a close relationship to the upper incomplete gamma function $\Gamma(a,z)=\int_z^\infty t^{a-1}e^{-t}\, dt$, via the identity (see \cite[(Eq. 8.4.19)]{OLBC})
\begin{equation}\label{incom}S_{n-1}(n\zeta)=e^{n\zeta}\frac {\Gamma(n,n\zeta)}{(n-1)!}.\end{equation}

Thus we have the identity
$$K_n(z,w)=ne^{n\zeta}\frac {\Gamma(n,n\zeta)}{(n-1)!}\cdot e^{-\frac n 2 (|z|^2+|w|^2)},\qquad (\zeta=z\bar{w}).$$

Asymptotics for $\Gamma(n,n\zeta)$ as $n\to\infty$ in the case when $|\arg(\zeta-1)|< \tfrac {3\pi} 4$ can be deduced from Tricomi's relation in \cite[(Eq. 11)]{Tr}, see the NIST handbook
\cite[(Eq. 8.11.9)]{OLBC} as well as \cite[(Eq. 2.2)]{ND} and the paper \cite{GOC}. The formula is reproduced in \eqref{triform} below.

Using the form in \cite{OLBC} we obtain readily that if $|\arg(\zeta-1)|< \tfrac {3\pi} 4$ then 
\begin{equation}\label{Tri}S_{n-1}(n\zeta)\sim\frac {n^{n-1}}{(n-1)!}\frac {\zeta^n} {\zeta-1}\sum_{j=0}^\infty \frac 1 {n^j} \frac {(-1)^jb_j(\zeta)}{(\zeta-1)^{2j}},\qquad (n\to\infty),\end{equation}
where $b_0(\zeta)=1$ and
\begin{equation}\label{recursion}b_j(\zeta)=\zeta(1-\zeta)\cdot b_{j-1}'(\zeta)+(2j-1)\zeta \cdot b_{j-1}(\zeta).\end{equation}

Via Stirling's formula (see Lemma \ref{bernoulli} below) we can now conclude
Theorem \ref{tue} in the case $|\arg(z\bar{w}-1)|<\tfrac {3\pi} 4$. Conversely, we can use Theorem \ref{tue} to conclude the following generalized version of Tricomi's expansion.

\begin{cor} \label{tricor} The asymptotic expansion 
\begin{equation}\label{triform}\Gamma(n,n\zeta)\sim n^{n-1}e^{-n\zeta}\frac {\zeta^n} {\zeta-1}\sum_{j=0}^\infty \frac 1 {n^j} \frac {(-1)^jb_j(\zeta)}{(\zeta-1)^{2j}},\qquad (n\to\infty),\end{equation}
holds for all $\zeta$ in the exterior Szeg\H{o} domain $\Sz$. The domain $\Sz$ is moreover the largest possible domain in which the expansion \eqref{triform} holds.
\end{cor}

\begin{rem}
The complete large $n$ asymptotics of $\Gamma(n,n\zeta)$ for $\zeta$ in the complex plane may be deduced by using bulk asymptotics in Theorem \ref{bulkapp} when $\zeta$ is inside or on the Szeg\H{o} curve, or error-function asymptotics when $\zeta$ is very close to the critical point $1$. We remark that a different kind of global asymptotics for the incomplete gamma function is given 
\cite{ND,T}. 
In a way, our above results show that the asymptotics discussed in those sources can be simplified further, and in different ways, depending on whether $\zeta$ is inside or outside of the Szeg\H{o} curve. 
\end{rem}


\smallskip

As already indicated, we will make use of (and develop) the method of approximate full-plane orthogonal polynomials from the paper \cite{HW}. Such orthogonal polynomials are sometimes called Carleman polynomials \cite{HW3}. In addition,
we want to point to the paper \cite{HW2}, which studies the ``root function'', essentially the Bergman space counterpart to the function
$$k_n(z,w)=\frac {K_n(w,z)}{\sqrt{K_n(z,z)}}.$$

This is just the weighted polynomial
square-root of the Berezin kernel: $B_n(z,w)=|k_n(z,w)|^2$. For $z$ and $w$ in appropriate regimes, an asymptotic expansion for $k_n(z,w)$ can be deduced from \cite[Theorem 1.4.1]{HW2}. In Subsection \ref{genem} we shall use this expansion to deduce qualitative information concerning the structure of Berezin kernels.

\smallskip

A different (and very successful) approach in the theory of full-plane orthogonal polynomials is found in the paper \cite{BBLM}, where strong asymptotics with respect to certain special types
of potentials is deduced using Riemann-Hilbert techniques. In recent years, a number of other particular ensembles of intrinsic interest have turned out to be tractable by this method, see for instance the discussion in Subsection \ref{arch} below. 
In \cite{IT} it is noted that planar orthogonal polynomials can be characterized as the unique solution to a certain matrix-valued $\dbar$-problem. In the recent papers \cite{He,HW0}, related ideas are used to study fine asymptotics for orthogonal polynomials, leading to some additional insights besides the original approach in \cite{HW} (which uses foliation flows, as we do below). 
\smallskip

In Section \ref{soop}, our main results are viewed in relation to the loop equation. Some further results and a comparison with other relevant work is found there.




\subsection{Plan of this paper} In Section \ref{Szesec} we consider the Ginibre ensemble and prove Theorem \ref{tue} and Corollary \ref{gharm}.

In Section \ref{LAGR} we provide some necessary background for dealing with more general random normal matrix ensembles.

In Section \ref{Sec_Gen}, we state an approximation formula for
$\tilde{K}_n(z,w)$ in \eqref{tail}
valid when $z$ and $w$ belong to $N(U,\delta_n)$.
This formula 
expresses $\tilde{K}_n(z,w)$ as a sum of certain weighted ``quasi-polynomials'',
which have the advantage of being analytically more tractable than the actual orthogonal polynomials. Summing by parts in this formula
we deduce
 Theorem \ref{kern1} and Theorem \ref{berth1}.

In Section \ref{details}, we provide a self-contained proof of the main approximation lemma used in Section \ref{Sec_Gen}. Our exposition is based on the method in \cite{HW}, but
is easier since (for example) we only require leading order asymptotics.

In Section \ref{soop} we view our main results in the context of the loop equation (or Ward's identity).
This leads to a hierarchy of identities
relating the Berezin measures with various nontrivial (geometrically significant) objects.


\subsection{Basic notation and terminology} \,

\smallskip
\noindent
Discs: $D(a,r)=\{z\in\C\,;\,|z-a|<r\}$; $\D_e(r)=\{|z|>r\}\cup\{\infty\}$; $\D_e=\D_e(1)$;

\smallskip
\noindent
Neighbourhood of a set $E$: $N(E,r)=E+D(0,r)$.

\smallskip
\noindent
Differential operators: $\d=\tfrac 1 2(\d_x-i \d_y)$, $\dbar=\tfrac 1 2(\d_x+i \d_y)$, $\Delta=\d\dbar$.

\smallskip
\noindent
Area measure: $dA=\tfrac 1 \pi\, dxdy$.

\smallskip
\noindent
$L^2$-scalar product and norm: $(f,g)=\int_\C f\bar{g}\, dA$; $\|f\|=\sqrt{(f,f)}$.

\smallskip
\noindent
Asymptotic relations: Given two sequences $a_n$ and $b_n$ of positive numbers we write: $a_n\sim b_n$ if $\lim_{n\to\infty}a_n/b_n= 1$; $a_n\lesssim b_n$ if $a_n/b_n\le C$ ($C$ some constant);
$a_n\asymp b_n$ if $a_n\lesssim b_n$ and $b_n\lesssim a_n$.

\subsubsection*{Acknowledgement} We want to thank P.J.~Forrester for helpful communication.

\section{Szeg\H{o}'s asymptotics and the Ginibre kernel} \label{Szesec}

In this Section we prove Theorem \ref{tue} and Corollary \ref{gharm} on asymptotics for the Ginibre kernel $K_n(z,w)$ in the case
when $z\bar{w}$ belongs to the exterior Szeg\H{o} domain $\Sz$. In addition, we shall state and prove Theorem \ref{bulkapp}
on bulk type asymptotics.

\begin{figure}[ht]
\begin{center}
\includegraphics[width=0.4\textwidth]{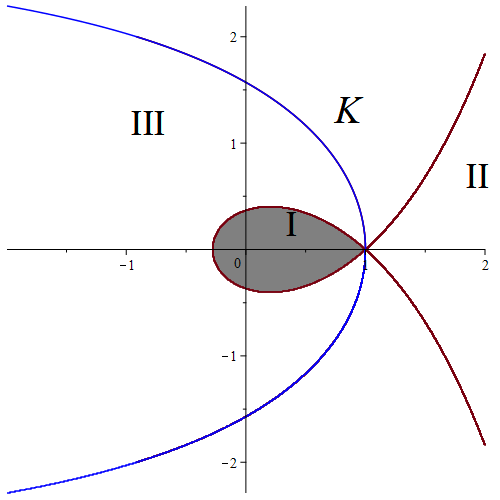}
\end{center}
\caption{Regions and curves used in the proof of Theorem \ref{tue}.}
\label{FigSz2}
\end{figure}

\subsection{Proof of Theorem \ref{tue}} \label{ptue} We start by writing the Ginibre kernel \eqref{gink} in the form
\begin{equation}\label{sd2}K_n(z,w)=nE_n(z\bar{w})e^{n z\bar{w}-\frac 12n|z|^2-\frac 12n|w|^2},\end{equation}
where
$$E_n(\zeta)=s_{n-1}(n\zeta),\qquad s_{n-1}(\zeta)=\sum_{k=0}^{n-1}\frac {\zeta^k}{k!} e^{-\zeta},\qquad \zeta=z\bar{w}.$$

A differentiation shows that
\begin{equation}\label{det0}E_n'(\zeta)=
-\frac {n^ne^{-n}}{(n-1)!}u(\zeta)^n\frac 1 \zeta,\qquad u(\zeta):=\zeta\, e^{\,1-\zeta}.\end{equation}

Following Szeg\H{o} \cite{Sz} we shall integrate in \eqref{det0} along certain judiciously chosen paths.

The proof of the following lemma is straightforward from the usual Stirling series for $\log n!$ (e.g.~\cite{Ahl}).

\begin{lem} \label{bernoulli} There are numbers $b_k$ starting with $b_0=1$ and $b_1=-\frac 1 {12}$ such that, for each $k\ge 0$,
\begin{align*}\frac {n^n e^{-n}}{(n-1)!}&=\sqrt{\frac n {2\pi}}\cdot (b_0+\frac {b_1} {n}+\cdots+\frac
{b_k}{n^k}+O(n^{-k-1})),\qquad (n\to\infty).\end{align*}
\end{lem}

We now define a curve $K$ and three regions $\I,\II,\III$ using the function $u(\zeta)=\zeta\, e^{\,1-\zeta}$, depicted
in Figure \ref{FigSz2}. The regions $\I$ (bounded) and $\II$ (unbounded) are defined to be the connected components of the set $\{|u(\zeta)|<1\}$. (Note that $\I=\Int\gsz$ is the
domain interior to the Szeg\H{o} curve \eqref{Jordan}.) We also define $\III:=\{|u(\zeta)|>1\}$.


\smallskip

Note that $u(\zeta)$ has a critical point at $\zeta=1$. We define the curve $K$ to be the portion of the level curve $\im u(\zeta)=0$ which intersects the real axis at right angles at $\zeta=1$.
We assume that $\zeta\ne 1$ and divide in two cases according to which $\zeta$ is to the left or to the right of the curve $K$. (The case when $\zeta$ is exactly on $K\setminus\{1\}$ will be handled easily afterwards.)

\smallskip


First assume that $\zeta$ is strictly to the right of $K$. (So $\zeta$ is either in region $\II$ or in region $\III$ or on the common boundary
of those regions.)

We integrate in \eqref{det0} over the curve connecting $\zeta$ to $\infty$ in a way so that the argument of $u(t)$ remains constant when $t$ traces the path of integration. The path is chosen so that $\re t\to +\infty$ as $t\to\infty$ along the curve; Figure \ref{FigSz3} illustrates the point. We find
\begin{equation}\label{det1}E_n(\zeta)=\frac {n^ne^{-n}}{(n-1)!}\int_\zeta^{+\infty}u(t)^{n}\frac {dt} t.
\end{equation}
\begin{figure}[ht]
\begin{center}
\includegraphics[width=0.4\textwidth]{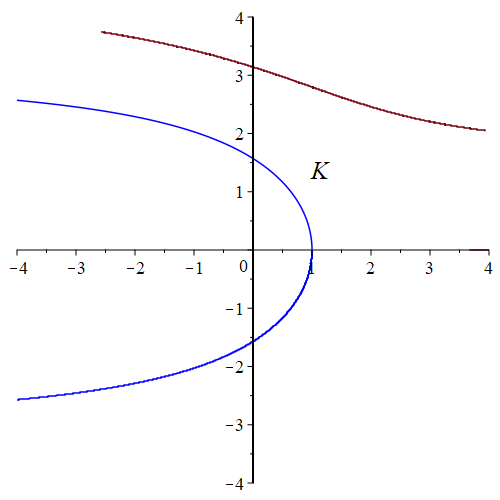}\qquad
\includegraphics[width=0.4\textwidth]{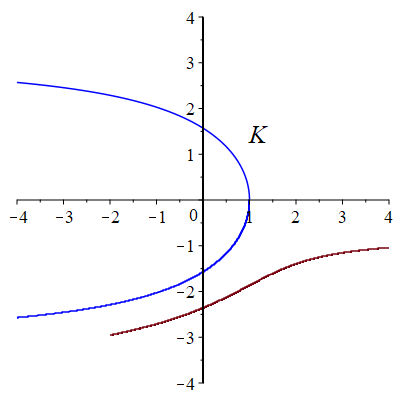}
\end{center}
\caption{Curves of constant argument connecting $\zeta$ with $+\infty$ when $\zeta$ is to the right of $K$. The first picture shows a curve where $u(t)$ is real; the second
picture has the argument of $u(t)$ equal to $\frac \pi 4+2\pi k$.}
\label{FigSz3}
\end{figure}

A curve on which $\arg u$ is constant is a steepest decent curve for $\log |u|$ by the Cauchy-Riemann equations. This gives that $|u(t)|$ strictly decreases from $|u(\zeta)|$ to zero as $t$ traces the curve from left to right. Thus we may unambiguously define
an inverse function $t(u)$ along the curve and obtain
$$\int_{\zeta}^{+\infty}u(t)^n\,\frac {dt} t=-\int_0^{u(\zeta)}\frac {u^n}{t(u)}\frac {dt}{du}\, du.$$

From $te^{1-t}=u$ we obtain
$$\frac u t\frac {dt}{du}=\frac 1 {1-t}$$
so the last integral reduces to
$$\int_0^{u(\zeta)}u^{n-1}\frac 1 {1-t(u)}\, du.$$

Now write $f(u)=(t(u)-1)^{-1}$ and consider the point
$$A=u(\zeta)=\zeta e^{1-\zeta}.$$

Then $f(A)=(\zeta-1)^{-1}$ and a (formal) repeated integration by parts gives
\begin{equation}\label{intbp}\begin{split}\int_0^A f(u)u^{n-1}\, du&=\frac {A^n}{n}f(A)-\int_0^Af'(u)\frac {u^n}n\, du=\cdots\\
&=\frac {A^n} n f(A)-\frac {A^{n+1}}{n(n+1)}f'(A)+\cdots+ (-1)^k\frac {A^{n+k}}
{n(n+1)\cdots (n+k)}f^{(k)}(A)+\cdots,\\
\end{split}\end{equation}
where as before $u$ has constant argument along the path of integration, say $u=e^{i\theta}x$ where $0\le x\le |A|$.

\smallskip

Setting $\tilde{f}(x)=f(e^{i\theta}x)$ and $M=\max\limits_{0\leq x\leq |A|}\{|\tilde{f}^{(k+1)}(x)|\}$ we obtain the estimate
\begin{equation}
\label{det2}\begin{split}\Big|\int_0^A f(u)u^{n-1}\, du& - \frac{A^n}{n} \left(f(A)-\frac A {n+1}f'(A)+\cdots +(-1)^k\frac {A^kf^{(k)}(A)}{(n+1)\cdots (n+k)}\right)\Big|\\
&\le M\int_0^{|A|}\frac {x^{n+k}}
{n(n+1)\cdots(n+k)}\, dx=M\frac {|A|^{n+k+1}}{n(n+1)\cdots(n+k+1)}.\\
\end{split}
\end{equation}

Now $f(u)=(t(u)-1)^{-1}$ gives
$$f'(u)=(-1)^j(t(u)-1)^{-2}t'(u)=-\frac 1 {u'(t)}\frac 1 {(t(u)-1)^{2}}.$$

Inserting here $u=u(\zeta)=A$ and using that $t(u(\zeta))=\zeta$ and
$$u'(\zeta)=e^{1-\zeta}(1-\zeta)=A\frac {1-\zeta}\zeta,$$ we obtain
\begin{equation}\label{det3}f'(A)=\frac 1 A\frac {\zeta} {(\zeta-1)^{3}}.\end{equation}

By induction, one shows easily that the higher derivatives have the structure
\begin{equation}\label{det4}f^{(j)}(A)=\frac {r_j(\zeta)} {A^j}\end{equation}
where $r_j(\zeta)$ is a rational function having a pole of order $2j+1$ at $\zeta=1$ and no other poles in $\hat{\C}$.

\smallskip

On account of \eqref{det2}, \eqref{det3}, \eqref{det4} and since $f(A)=(\zeta-1)^{-1}$ we have shown that
\begin{equation}\label{subl}\begin{split}\int_0^Af(u)u^{n-1}\, du
=\frac {A^n}{n}\frac 1 {\zeta-1}\left(1-\frac 1n\frac \zeta {(\zeta-1)^2}+\frac 1 {n^2}\tilde{r}_2(\zeta)+
\cdots +\frac 1 {n^k}\tilde{r}_k(\zeta)+O(n^{-k-1})\right),
\end{split}
\end{equation}
where $\tilde{r}_j(z)$ is a new rational function with pole of order $2j$ at $\zeta=1$.

Recalling that $\zeta=z\bar{w}$ and using \eqref{det1} and Lemma \ref{bernoulli},
\begin{align*}K_n(z,w)&=nE_n(\zeta)e^{n\zeta-\frac 1 2 n|z|^2-\frac 12 n|w|^2}\\
&=\frac {\sqrt{n}}
{\sqrt{2\pi}}\left(1-\frac 1 {12n}+\cdots\right) (\zeta e^{1-\zeta})^n\frac 1 {\zeta-1}\left(1-\frac 1 n \frac \zeta {(\zeta-1)^2}+\cdots\right)e^{n\zeta-\frac 1 2 n|z|^2-\frac 12 n|w|^2}\\
&=\frac {\sqrt{n}}
{\sqrt{2\pi}} \zeta^n e^{n-\frac 1 2 n|z|^2-\frac 12 n|w|^2} \frac 1 {\zeta-1}\left[1-\frac 1 n\left(\frac 1 {12}+\frac \zeta {(\zeta-1)^2}\right)+\cdots\right],
\end{align*}
where the expression in brackets is short for
\begin{equation}\label{inbra}1-\frac 1 n\left(\frac 1 {12}+\frac \zeta {(\zeta-1)^2}\right)+\sum_2^k \frac {\rho_j(\zeta)}{n^{j}}+O(n^{-k-1}),\end{equation} and each $\rho_j(\zeta)$ is a rational function with a pole of order $2j$ at $\zeta=1$
and no other poles.

We have arrived at the expansion formula \eqref{fex} in the case when $\zeta=z\bar{w}$ is strictly to the right of the curve $K$.

\smallskip

Next we suppose that $\zeta=z\bar{w}$ is strictly to the left of the curve $K$. (Thus $\zeta$ is either in $\I$ or in $\II$ or on the common boundary
of these domains.)

This time we can find a curve of constant argument of $u(t)=te^{1-t}$ connecting $0$ with $z$, along which $|u(t)|$ is strictly increasing. See Figure \ref{FigSz4}.

\begin{figure}[ht]
	\begin{center}
		\includegraphics[width=0.4\textwidth]{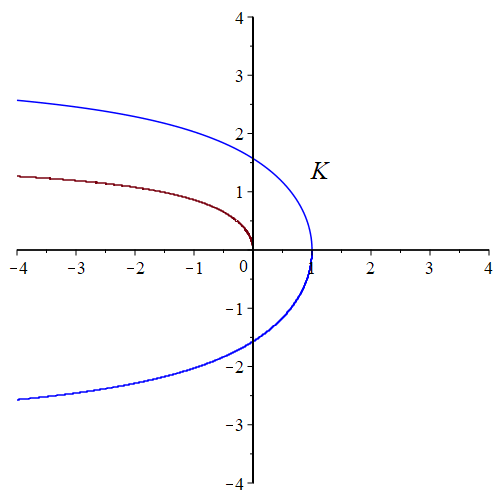}\qquad
		\includegraphics[width=0.4\textwidth]{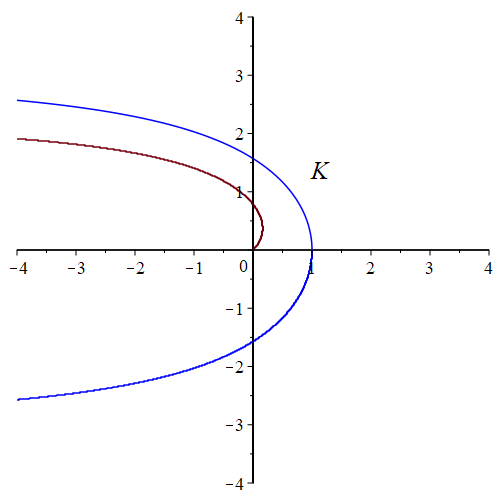}
	\end{center}
	\caption{Curves of constant argument connecting $0$ with $\zeta$ when $\zeta$ is to the left of $K$. The first picture shows a curve where $u(t)$ is real; the second
		picture has the argument of $u(t)$ equal to $\frac \pi 4+2\pi k$.}
	\label{FigSz4}
\end{figure}

We now integrate in \eqref{det0} (using the fact that $E_n(0)=1$) to write
\begin{equation}\label{1o}E_n(\zeta)= 1-g_n(\zeta)\end{equation}
where
\begin{align*}g_n(\zeta)= \frac {n^ne^{-n}}{(n-1)!}
\int_0^\zeta u(t)^n\,\frac {dt} t.\end{align*}
The path of integration is the curve of constant argument of $u(t)$ indicated above.

As before, letting $t(u)$ be the inverse function we find
\begin{align*}g_n(\zeta)=\frac {n^ne^{-n}}{(n-1)!}\int_0^{u(\zeta)}u^{n-1} \frac 1 {1-t}\, du.\end{align*}

By Stirling's approximation (Lemma \ref{bernoulli}) and the asymptotic expansion \eqref{subl},
\begin{equation}\label{2o}g_n(\zeta)= -\sqrt{\frac  1{2\pi n}}u(\zeta)^n\cdot \frac 1 {\zeta-1}\left[1-\frac 1 n\left(\frac 1 {12}+\frac \zeta {(\zeta-1)^2}\right)+\cdots\right],\end{equation}
where the expression in brackets is precisely the same as in \eqref{inbra}.

Recalling that $\zeta=z\bar{w}$ we obtain,
as a consequence of \eqref{1o} and \eqref{2o} that
$$E_n(z\bar{w})= 1+\sqrt{\frac 1 {2\pi n}}(z\bar{w}e^{1-z\bar{w}})^n\cdot \frac 1 {z\bar{w}-1}\cdot \left(1-\frac 1 n\left(\frac 1 {12}+\frac \zeta {(\zeta-1)^2}\right)+\cdots\right),$$
and hence, by \eqref{sd2},
\begin{equation}\label{chug}\begin{split}K_n(z,w)&= n\left[1+\sqrt{\frac  1 {2\pi n}}(z\bar{w}e^{1-z\bar{w}})^n\cdot \frac 1 {z\bar{w}-1}\cdot \left(1-\frac 1 n\left(\frac 1 {12}+\frac \zeta {(\zeta-1)^2}\right)+\cdots\right)\right]\cr
&\quad \times e^{nz\bar{w}-\frac 12n|z|^2-\frac 12n|w|^2}.\cr
\end{split}\end{equation}

The asymptotic formula in \eqref{chug} has proven for all $\zeta=z\bar{w}$ to the left of the curve $K$.

We next note that if $\zeta=z\bar{w}$ is in the region $\III$, i.e., if $|z\bar{w}e^{1-z\bar{w}}|>1$, then the first term ``$1$'' inside the bracket in \eqref{chug} is negligible, so in this case
$$K_n(z,w)= \sqrt{\frac n {2\pi}}(z\bar{w})^ne^{n-\frac 12n|z|^2-\frac 12n|w|^2}\cdot \left(1-\frac 1 n\left(\frac 1 {12}+\frac \zeta {(\zeta-1)^2}\right)+\cdots\right),\qquad (z\bar{w}\in \III),$$
as desired.

\smallskip

There remains to treat the case when $\zeta=z\bar{w}$ happens to be precisely on the curve $K$ and $\zeta \ne 1$. In this case, we consider nearby points $\zeta'$ which are either to the left or to the right of $K$ and use a limiting procedure, as $\zeta'\to\zeta$ to deduce that the asymptotic formula \eqref{fex} is true in this case as well. (Intuitively, one can picture that for $\zeta\in  K$ we connect $\zeta$ either to $0$ or to $+\infty$ by first following the curve $K$ until we reach
$t=1$, and then continue along the real axis until we reach either $0$ or $+\infty$. This picture is however not entirely rigorous, since $\frac {dt}{du}$ has a pole at at $u=t=1$.) 

Our proof of Theorem \ref{tue} is complete. q.e.d.

\subsection{Bulk asymptotics for the Ginibre kernel} \label{bulag}
As a corollary of our above proof, we also obtain the following bulk type asymptotic expansion. (A related statement is found in \cite[Proposition 2]{BG}.)

\begin{thm} \label{bulkapp} For $z\bar{w}\in\C\setminus(\Sz\cup\{1\})$ we write $\rho=|z\bar{w}e^{1-z\bar{w}}|$. Then $\rho\le 1$ and
we have the bulk-asymptotic formula
\begin{equation}\label{standard}
K_n(z,w) = n e^{nz\bar{w}-\frac 12n|z|^2-\frac 12n|w|^2}\cdot (1+O(\frac {\rho^n} {\sqrt{n}})),\qquad z\bar{w}\in \C\setminus (\Sz\cup\{1\}),
\end{equation}
where the implied $O$-constant is uniform for $z\bar{w}$ in the complement of any neighbourhood of $1$.
\end{thm}

\begin{proof} The asymptotic formula in \eqref{chug} applies since $\zeta=z\bar{w}$ is on the left of the curve $K$ under the assumptions in Theorem \ref{bulkapp}. Moreover
the second term inside the bracket in \eqref{chug} is $O(n^{-\frac 1 2}|\zeta e^{1-\zeta}|^n|\zeta-1|^{-1})$.
\end{proof}

\begin{rem} It follows that if $z\bar{w}\in \C\setminus (\Sz\cup\{1\})$, then the Berezin kernel $B_n(z,w)$ satisfies the heat-kernel asymptotic
$B_n(z,w)=ne^{-n|z-w|^2}\cdot (1+o(1))$. This has been well-known
when the points $z$ and $w$ are close
enough to the diagonal $z=w$ and in the interior of the droplet, cf. \cite{A2,AHM1,AKM}. The main point in Theorem \ref{bulkapp}
is that we obtain the precise domain of ``bulk asymptoticity''.
\end{rem}

\subsection{Proof of Corollary \ref{gharm}}
Fix a (finite) point $z$ in the exterior disc $\D_e$, and consider the
Berezin measure
$d\mu_{n,z}(w)=B_n(z,w)\, dA(w).$
We aim to prove that $\mu_{n,z}$ converges to the harmonic measure $\omega_z$ in a Gaussian way.

For this purpose, we fix a sequence $(c_n)$ of positive numbers with $nc_n^2\to\infty$ and $nc_n^3\to 0$ as $n\to\infty$; we can without loss of generality
assume that $c_n<1$ for all $n$. We then consider points $w$ in the belt $N(\T,c_n)$, represented in the form
\begin{equation}\label{rel}w=f_n(\theta,t)=e^{i\theta}(1+\frac t {2\sqrt{n}}),\qquad (|t|\le 2\sqrt{n}\,c_n).\end{equation}

A computation shows that
$$dA(w)=\frac 1 {2\pi\sqrt{n}}(1+\frac t {2\sqrt{n}})\,d\theta dt.$$

Let $\hat{\mu}_{n,z}(\theta,t)=\mu_{n,z}\circ f_n(\theta,t)$
be the pull-back of $\mu_{n,z}$ by $f_n$. Also fix $\theta\in\T$ and consider the radial cross-section
$$\vro_n(t)=\vro_{n,\theta}(t)=\frac 1 {2\pi\sqrt{n}}(1+\frac t {2\sqrt{n}})B_n(z,e^{i\theta}(1+\frac t {2\sqrt{n}})),\qquad (t\in\R).$$

Since $d\hat{\mu}_{n,z}(\theta,t)=\vro_{n,\theta}(t)\, d\theta dt$, it suffices to study asymptotics of the function $\vro_{n,\theta}(t)$.

Fixing $t$ we now define an $n$-dependent point $w=f_n(\theta,t)$, as in \eqref{rel}. Then by Theorem \ref{tue}
\begin{align*}\vro_n(t)\sim \frac 1 {(2\pi)^{\frac 32}}\frac {|zw|^{2n}e^{2n-2n\re(z\bar{w})} e^{2n\re(z\bar{w})-n|z|^2-n|w|^2}\frac {1}{|1-z\bar{w}|^2}}
{|z|^{2n}e^{n-n|z|^2}\frac {1}{|z|^2-1}}\end{align*}

After some simplification using that $f_n(\theta,t)=e^{i\theta}(1+o(1))$ and $|f_n(\theta,t)|=1+o(1)$ as $n\to\infty$, we find that
$$\vro_n(t)\sim \frac 1 {\sqrt{2\pi}}|w|^{2n}e^{n-n|w|^2}P_z(\theta)$$
where
$P_z(\theta)=\frac 1 {2\pi}\frac {|z|^2-1}{|1-z e^{-i\theta}|^2}$ is the Poisson kernel.

We next observe that
$$\log|w|^{2n}=2n\log(1+\frac t {2\sqrt{n}})=t\sqrt{n}-\frac {t^2}4+O(nc_n^{\,3}),$$
so (since $nc_n^3\to 0$) we obtain, with $\gamma(t)=\frac 1 {\sqrt{2\pi}}e^{-\frac 1 2 t^2}$,
\begin{align*}\vro_n(t)&= \frac 1 {\sqrt{2\pi}}e^{t\sqrt{n}-\frac {t^2}4+n-n(1+\frac t{2\sqrt{n}})^2}P_z(\theta)\cdot (1+o(1))\\
&=\gamma(t)P_z(\theta)\cdot (1+o(1)),
\end{align*}
where the last equality follows by straightforward simplification.

Finally setting $\ell=t/(2\sqrt{n})$ it follows that the Berezin measure $\mu_{n,z}$ has the uniform asymptotic
\begin{align*}d\mu_{n,z}(w)&\sim \frac  {2\pi\sqrt{n}}{\pi} \gamma(t)P_z(\theta)\, d\theta dt= \gamma_n(\ell)P_z(\theta)\, d\theta d\ell,\qquad (|\ell|\le c_n),\end{align*}
where $\gamma_n(\ell)=2\sqrt{n}\frac 1 {\sqrt{2\pi}}e^{-2n\ell^2}\, d\ell$. q.e.d.

\section{Potential theoretic preliminaries} \label{LAGR}

This section
begins by recalling how boundary regularity
follows from
Sakai's main result in \cite{Sa}. After that we recast some useful facts pertaining to Laplacian growth and obstacle problems. Finally 
we will state and prove a number of estimates for weighted polynomials, which will come in handy when approximating the reproducing kernel $K_n(z,w)$ by its tail in the next section.

\subsection{Sakai's theorem on boundary regularity} \label{saksek} Let $Q$ be an admissible potential.

As always we denote by $S$ the droplet and $U$ the component of $\hat{\C}\setminus S$ containing infinity.

We also write $\Gamma=\d U$ and
$\chi:\D_e\to U$ for the conformal mapping that satisfies $\chi(\infty)=\infty$ and $\chi'(\infty)>0$.

\begin{lem} \label{schw} Let $p$ be an arbitrary point on $\Gamma$. There exists a neighbourhood $N$ of $p$ and a ``local Schwarz function'', i.e., a holomorphic function $\calS(z)$ on $N\setminus S$, continuous up to
$N\cap \Gamma$ and satisfying $\calS(z)=\bar{z}$ there.
\end{lem}

\begin{proof} Without loss of generality set $p=0$. 

Choosing the neighbourhood $N$ sufficiently small we can write
$Q(z)=\sum_{j,k=0}^\infty a_{j,k} z^j\bar{z}^k$ with convergence for all $z$ in $N$.
By polarization we define
$H(z,w)=\sum_{j,k=0}^\infty a_{j,k} z^jw^k.$

Next define a Lipschitzian function $G(z,w)$ in $N\times N$ by
$$G(z,w)=\d_z H(z,w)-\d\check{Q}(z).$$

Here $\check{Q}$ is the obstacle function, defined in Subsection \ref{backgr}.

In $N\setminus S$, the function $\d\check{Q}$ is holomorphic and we further have that $\d\check{Q}=\d Q$ on $N\cap(\d S)$.
Hence $G(z,w)$ is a holomorphic function of $z$ for $z\in N\setminus S$, $G(0,0)=0$ and $\d_w G(z,w)|_{(0,0)}=\Delta Q(0)>0$.
Moreover, the identity $Q(z)=H(z,\bar{z})$ shows that
$$G(z,\bar{z})=\d(Q-\check{Q})(z),\qquad (z\in N)$$
so $G(z,\bar{z})=0$ for all $z\in N\cap S$, and, in particular, for all $z\in N\cap (\d S)$.

By the implicit function theorem (in its version for Lipschitz functions \cite{DR}) we may, by diminishing $N$ if necessary, find a unique Lipschitzian solution
$\calS(z)$ to the equation $G(z,\calS(z))=0$, $z\in N$. 

Then for $z\in N\setminus S$ we obtain $0=\dbar_z G(z,\calS(z))=\d_2 G(z,\calS(z))\cdot \dbar \calS(z)$,
proving that $\calS(z)$ is holomorphic in $N\setminus S$. 

We also find that $G(z,\bar{z})=G(z,\calS(z))=0$ for $z\in N\cap(\d S)$, so $\calS(z)=\bar{z}$
at such points. 
\end{proof}

\begin{thm} \label{sak} The conformal map $\chi:\D_e\to U$ extends analytically across $\T$ to an analytic function on a neighbourhood
of the closure of $\D_e$. As a consequence
$\Gamma=\chi(\T)$ is a finite union of real analytic arcs and possibly finitely many singular points, which are either cusps
(corresponding to points $p=\chi(z)$ with $z\in\T$ and $\chi'(z)=0$) or double points ($p=\chi(z_1)=\chi(z_2)$ where $z_1,z_2\in\T$ and $z_1\ne z_2$).
\end{thm}

The proof is immediate from Sakai's regularity theorem in \cite{Sa}, since $\Gamma$ is a continuum and since a local Schwarz function for $U$ exists near each point
of $\Gamma$ by Lemma \ref{schw}.

\subsection{Laplacian growth and Riemann maps} \label{lgrm}

For fixed $\tau\in (0,1]$ we let $\check{Q}_\tau$ be the obstacle function which grows like
$$\check{Q}_\tau(z)=2\tau\log|z|+O(1),\qquad z\to \infty.$$

By this we mean that $\check{Q}_\tau(z)$ is the supremum of $s(z)$ where $s$ runs through the class $\SH_\tau(Q)$ of subharmonic functions $s$ on $\C$ which satisfy
$s\le Q$ on $\C$ and $s(w)\le 2\tau\log|w|+O(1)$ as $w\to\infty$.

\smallskip

Similar as for the case $\tau=1$,
the function $\check{Q}_\tau$ is $C^{1,1}$-smooth on $\C$ and harmonic on $\C\setminus S_\tau$ where $S_\tau=S[Q/\tau]$ is the droplet in potential $Q/\tau$,
while $Q=\check{Q}_\tau$ on $S_\tau$. (See \cite{LM,ST}.)

Clearly the droplets $S_\tau$ increase with $\tau$; the evolution is known as Laplacian growth, cf. \cite{GTV,HW,LM,Z}.

\smallskip

Recall that
$d\sigma=\Delta Q\cdot \1_S\, dA$ denotes the equilibrium measure in external potential $Q$.
It is easy to see that
$\sigma(S_\tau)=\tau,$
and that the restricted measure $\sigma_\tau$ defined by
\begin{equation}\label{eqtau}\sigma_\tau=\Delta Q\cdot \1_{S_\tau}\, dA,\end{equation}
minimizes the weighted energy $I_Q[\mu]$ in \eqref{wen}
among all compactly supported Borel measures $\mu$ of total mass $\mu(\C)=\tau$. We refer to $\sigma_\tau$ as the \textit{equilibrium measure of mass $\tau$}.

\smallskip

 Write $U_\tau$
for the component of $\hat{\C}\setminus S_\tau$ containing $\infty$ and
$\Gamma_\tau=\d U_\tau$
for the outer boundary of $S_\tau$.

By hypothesis, $\Gamma=\Gamma_1$ is everywhere regular (real-analytic). From this and
basic facts about Laplacian growth \cite{GTV,HW} we conclude that there are numbers $\tau_0<1$ and $\epsilon>0$ such that
$\Gamma_\tau$ is everywhere regular whenever $\tau_0-\epsilon\le\tau\le 1$. Indeed $\tau_0$ and $\epsilon$ can be chosen so that each potential $Q/\tau$
with $\tau_0-\epsilon\le\tau\le 1$ is admissible in the sense of Subsection \ref{backgr}.

\smallskip

We denote by
$\phi_\tau:U_\tau\to\D_e$
the conformal mapping normalized by $\phi_\tau(\infty)=\infty$ and $\phi_\tau'(\infty)>0$.
We also write
$\normal_\tau:\Gamma_\tau\to\T$ for
the unit normal on $\Gamma_\tau$ pointing out of $S_\tau$.

\begin{lem} \label{lem13} (\emph{``Rate of propagation of $\Gamma_\tau$.''}) The boundary $\Gamma_\tau$ moves in the direction of $\normal_\tau$ with local speed $|\phi_\tau'|/2\Delta Q$
in the following precise sense.

Pick two numbers $\tau',\tau$ in the interval $[\tau_0-\epsilon,1]$. Fix $z\in \Gamma_{\tau'}$ and let $p$ be the point in $\Gamma_{\tau}$
which is closest to $z$. Then
$$z=p+(\tau'-\tau)\frac {|\phi_{\tau}'(p)|}{2\Delta Q(p)}\normal_{\tau}(p)+O((\tau'-\tau)^2),\qquad (\tau\to \tau')$$
and
$$\normal_{\tau}(p)=\normal_{\tau'}(z)+O(\tau-\tau'),\qquad (\tau'\to \tau),$$
where the $O$-constants are uniform in $z$.
In particular there are constants $0<c_1\le c_2$ such that
\begin{equation}\label{l02}c_1|\tau-\tau'|\le \dist(\Gamma_{\tau},\Gamma_{\tau'})\le c_2|\tau-\tau'|.\end{equation}
\end{lem}

For a proof we refer to \cite[Lemma 2.3.1]{HW}. (Cf. \cite[Lemma 5.2]{AKS}.)

\smallskip

For given $\tau$ with $\tau_0\le \tau\le 1$ we denote
$$V_\tau=\text{``harmonic continuation of }\check{Q}_\tau\Big|_{U_\tau}\text{ across }\Gamma_\tau\text{.''}$$
Modifying $\tau_0<1$  and $\epsilon>0$ if necessary, we may
assume that $V_\tau$ is well-defined and harmonic on $\C\setminus K$ where $K$ is a compact subset of $\Int \Gamma_{\tau_0-\epsilon}$.
The set $K$ can be chosen depending only on $\tau_0$ and $\epsilon$ and not
on the particular $\tau$ with $\tau_0\le \tau\le 1$.

\smallskip

We shall frequently use the following identity:
\begin{equation}\label{vb}V_\tau(z)=\re\calQ_\tau(z)+\tau\log|\phi_\tau(z)|^2,\qquad z\in\C\setminus K,\end{equation}
where $\calQ_\tau$ is the unique holomorphic function on $U_\tau$ with $\re\calQ_\tau=Q$ on $\Gamma_\tau$ and $\im\calQ_\tau(\infty)=0$. (This follows since the left and right sides
agree on $\Gamma_\tau$ and have the same order of growth at infinity.)

\smallskip

We turn to a few basic estimates for the function $Q-V_\tau$, which we may call ``$\tau$-ridge''.

\begin{lem} \label{short} Suppose that $\tau_0-\epsilon \le\tau\le 1$ and let $p$ be a point on $\Gamma_\tau$. Then for $\ell\in\R$,
$$(Q-V_\tau)(p+\ell\cdot \normal_\tau(p))=2\Delta Q(p)\cdot \ell^2+O(\ell^3),\qquad (\ell\to 0),$$
where the $O$-constant can be chosen independent of the point $p\in\Gamma_\tau$.
\end{lem}

\begin{proof} Using that $Q=\check{Q}_\tau$ on $S_\tau$ and that $\check{Q}_\tau$ is $C^{1,1}$-smooth, we find that $\tfrac {\d^2}{\d n^2}(Q-V_\tau)(p)=
(\tfrac {\d^2}{\d n^2}+\tfrac {\d^2}{d s^2})(Q-V_\tau)(p)=4\Delta Q(p)$, where $\tfrac \d {\d n}$ and $\tfrac \d {\d s}$ denote differentiation in the normal
and tangential directions, respectively. The result now follows from Taylor's formula.
\end{proof}


 Our next lemma is immediate from Lemma \ref{short} when $z$ is close to $\Gamma_\tau$ and follows easily from our standing assumptions on $Q$ when $z$ is further away (cf. the proof of \cite[Lemma 2.1]{A}, for example).

\begin{lem} \label{long} Suppose that $\tau_0-\epsilon\le \tau\le 1$ 
and that $z$ is in the complement $\C\setminus K$, where $K\subset \Int \Gamma_{\tau_0-\epsilon}$ is defined above.
Then with $\delta_\tau(z)=\dist(z,\Gamma_\tau)$ there is a number $c>0$ such that
$$(Q-V_\tau)(z)\ge c\min\{\delta_\tau(z)^2,1\}.$$
\end{lem}

Combining Lemma \ref{lem13} with Lemma \ref{short}, we now obtain the following useful result.

\begin{lem} \label{slong} Suppose that $\tau,\tau'$ are in the interval $[\tau_0-\eps,1]$. For a given point $z\in\Gamma_{\tau'}$ let $p\in \Gamma_\tau$ be the point closest to $z$. Then
\begin{equation}(Q-V_\tau)(z)=\frac {|\phi_\tau'(p)|^2}{2\Delta Q(p)}(\tau'-\tau)^2+O((\tau'-\tau)^3),\qquad (\tau'\to\tau).\end{equation}
In particular, if $\tau_0<1$ and $\epsilon>0$ are chosen close enough to $1$ and $0$ respectively, then there are constants $c_1$ and $c_2$ independent of $\tau$, $\tau'$, $z$ such that
\begin{equation}c_1(\tau-\tau')^2\le (Q-V_\tau)(z)\le c_2(\tau-\tau')^2,\qquad (z\in\Gamma_{\tau'}).\end{equation}
\end{lem}

Before closing this section, it is convenient to prove a few facts about weighted polynomials.

\subsection{Pointwise estimates for weighted orthogonal polynomials} We now collect a number of estimates whose main purpose is to ensure a desired tail-kernel approximation in the next section.
To this end, the main fact to be applied is Lemma \ref{jugo}.

\smallskip

We start by proving the following pointwise-$L^2$ estimate,  following a slight variation on a technique which is well-known in the literature.


\begin{lem} \label{ho0} Let $W=P\cdot e^{-\frac 1 2 nQ}$ be a weighted polynomial where $j=\deg P\le n$. Put $\tau(j)=j/n$ and suppose $\tau(j)\le\tau$ where $\tau$ satisfies $0<\tau\le 1$. There is then a constant
$C$ depending only on $Q$ such that for all $z\in\C$,
$$|W(z)|\le C\sqrt{n}\|W\|e^{-\frac 1 2 n(Q-\check{Q}_\tau)(z)}.$$
\end{lem}

\begin{proof} Let $M_\tau$ be the maximum of $W$ over $S_\tau$. We shall first prove that
\begin{equation}\label{fp}|W(z)|\le M_\tau \cdot e^{-\frac 1 2 n(Q-\check{Q}_\tau)(z)}.\end{equation}

To this end we may assume that $M_\tau= 1$.
Consider the function
$$s(w)=\frac 1 n\log |P(w)|^2=\frac 1 n\log|W(w)|^2+Q(w),$$
which is subharmonic on $\C$ and satisfies $s\le Q$ on $\Gamma_\tau$. Moreover, $s(w)\le 2\tau\log|w|+O(1)$ as $w\to\infty$.
Hence by the strong maximum principle we have $s\le \check{Q}_\tau$ on $\C$, proving \eqref{fp}.

We next observe that there is a constant $C$ independent of $\tau$ such that
\begin{equation}\label{dp}M_\tau\le C\sqrt{n}\|W\|.\end{equation}
Indeed, \eqref{dp} follows from a standard pointwise-$L^2$ estimate, see for example \cite[Lemma 2.4]{A}.

Combining \eqref{fp} and \eqref{dp} we finish the proof of the lemma.
\end{proof}

We shall need to compare obstacle functions $\check{Q}_\tau(z)$ for different choices of parameter $\tau$. It is convenient to note the following two lemmas.

\begin{lem}\label{ho1} Suppose that $\tau_0\le\tau\le \tau'\le 1$. Then there is a constant $c>0$ depending only on $\tau_0$ and $Q$ such that
$$(\check{Q}_{\tau'}-\check{Q}_\tau)(z)\ge c(\tau'-\tau)^2
,\qquad (z\in\Cl U_{\tau'}).$$
\end{lem}

\begin{proof} Write
$$H(z)=(\check{Q}_{\tau'}-\check{Q}_\tau)(z)-(\tau'-\tau)\log|\phi_{\tau'}(z)|^2.$$
Then $H$ is harmonic on $U$ (including infinity) and has boundary values $H(z)=(Q-V_\tau)(z)$ for $z\in \Gamma_{\tau'}$.

For a given $z\in\Gamma_{\tau'}$ we let $p\in\Gamma_\tau$ be the closest point and write $z=p+\ell\cdot\normal_\tau(p)$. Then $|\ell|\asymp \tau'-\tau$ by Lemma \ref{lem13}, and by Lemma \ref{short}
$H(z)=2\Delta Q(p)\cdot\ell^2+O(\ell^3).$

Increasing $\tau_0<1$ a little if necessary, we obtain $H\ge c(\tau'-\tau)^2$ everywhere on $\Gamma_{\tau'}$ where $c>0$ is a constant depending on $\tau_0$ and $Q$.
By the maximum principle, the inequality $H\ge c(\tau'-\tau)^2$ persists on $U_{\tau'}$.
\end{proof}

\begin{lem}\label{ch0} Let $W=P\cdot e^{-\frac 1 2 nQ}$ be a weighted polynomial where $j=\deg P\le n$ with $\|W\|=1$. Suppose also that $\tau(j)\le \tau$ where $\tau_0\le \tau\le\tau'\le 1$. Then there are constants $C$ and $c>0$ such that
$$|W(z)|\le C\sqrt{n}e^{-cn(\tau'-\tau)^2}e^{-\frac 1 2 n(Q-\check{Q}_{\tau'})(z)},\qquad z\in\Cl U_{\tau'}.$$
\end{lem}

\begin{proof} Combining Lemma \ref{ho0} with Lemma \ref{ho1} we find that for all $z\in\Cl U_{\tau'}$
\begin{align*}|W(z)|&\le C\sqrt{n}e^{-\frac 1 2 n(\check{Q}_{\tau'}-\check{Q}_\tau)(z)}e^{-\frac 1 2 n(Q-\check{Q}_{\tau'})(z)}\\
&\le C\sqrt{n}e^{-\frac 1 2 cn(\tau'-\tau)^2}e^{-\frac 1 2 n(Q-\check{Q}_{\tau'})(z)}.
\end{align*}
\end{proof}

Finally, we arrive at following estimate, which will be used to discard lower order terms in the tail-kernel approximation in the succeeding section.

\begin{lem}\label{jugo} Let
$$\theta_n=1-\frac {\log n}{\sqrt{n}},\qquad \delta_n=M\sqrt{\frac {\log\log n} n}.$$
Suppose that $\tau(j)\le \theta_n$ and let $W_{j,n}(z)$ be the $j$:th weighted orthonormal polynomial in the subspace $\calW_n\subset L^2(\C)$. There are then constants $C$ and $c>0$ depending only on $Q$ such that
$$|W_{j,n}(z)|\le Ce^{-c\log^2 n}e^{-\frac 1 2 n(Q-\check{Q})(z)},\qquad (z\in N(U,\delta_n)).$$
\end{lem}

\begin{proof} We may assume that $\tau(j)\ge \tau_0$ where $\tau_0<1$ is as close to $1$ as we please.
We apply Lemma \ref{ch0} with
$$\tau=\theta_n, \qquad \tau'=1-CM\sqrt{\frac {\log\log n} n}$$
where $C>0$ is chosen so that $N(U,\delta_n)\subset U_{\tau'}$. It is possible to find such a $C$ by Lemma \ref{lem13}.

Applying Lemma \ref{ch0}, 
we find that (with a new $C$)
$$|W_{j,n}(z)|\le C\sqrt{n}e^{-cn(\tau'-\tau)^2}e^{-\frac 1 2 n(Q-\check{Q}_{\tau'})(z)},\qquad (z\in N(U,\delta_n)).$$
Since $\check{Q}_{\tau'}\le\check{Q}$ and
$$n(\tau'-\tau)^2\asymp \log^2 n,\qquad (n\to\infty)$$
we finish the proof by choosing $c>0$ somewhat smaller.
\end{proof}






\section{Kernel asymptotics: proofs of the main results} \label{Sec_Gen}

In this section, we prove Theorem \ref{kern1} on asymptotics for reproducing kernels, and Theorem \ref{berth1} on Gaussian convergence of Berezin measures.

Throughout the section, we fix an external potential $Q$ obeying the standing assumptions in Subsection \ref{backgr}.

\subsection{The tail kernel} \label{tailk}
Consider the tail kernel
\begin{equation}\label{hdk}\tilde{K}_n(z,w)=\sum_{j=n\theta_n}^{n-1}W_{j,n}(z)\overline{W_{j,n}(w)},\end{equation}
where $W_{j,n}=P_{j,n}\cdot e^{-\frac 12nQ}$, is the $j$:th weighted orthogonal polynomial, i.e., $P_{j,n}$ has degree $j$ and positive leading coefficient. The numbers $\theta_n$ and $\delta_n$ are defined by
\begin{equation}\label{tetn}\theta_n=1-\frac {\log n}{\sqrt{n}},\qquad \delta_n=M\sqrt{\frac {\log\log n} n},\end{equation}
where $M$ is fixed (depending only on $Q$).

\smallskip

The following approximation lemma is our main tool; we
remind once and for all that the symbol $U$ denotes the component of the complement of the droplet $S$ which contains $\infty$.


\begin{lem} \label{main0} \emph{(``Main approximation lemma'')} Suppose that
$$z,w\in N(U,\delta_n)$$
and let $\beta$ be any fixed number with $0<\beta<\tfrac 1 4$. Then with $\tau(j)=\frac j n$ we have
\begin{equation}\label{exa}\begin{split}&\tilde{K}_n(z,w)=\sqrt{\frac n {2\pi}}e^{-\frac n 2(Q(z)+Q(w))}\cdot (1+O(n^{-\beta}))\cr
&\times \sum_{j=n\theta_n}^{n-1}\sqrt{\phi_{\tau(j)}'(z)}\overline{\sqrt{\phi_{\tau(j)}'(w)}}e^{\frac n 2(\calQ_{\tau(j)}(z)+\overline{\calQ_{\tau(j)}(w)})}e^{\frac 1 2(\calH_{\tau(j)}(z)+\overline{\calH_{\tau(j)}(w)})}\phi_{\tau(j)}(z)^j\overline{\phi_{\tau(j)}(w)^j}.\cr
\end{split}
\end{equation}
\end{lem}

Throughout this section, we will accept the lemma;
a relatively short derivation, based on 
the method 
in \cite{HW}, is given in Section \ref{details}.

\medskip

We now turn to the proofs of our main results (Theorems \ref{kern1} and \ref{berth1}).

Towards this end (using notation such as $\calQ=\calQ_1$ and $\phi=\phi_1$) we rewrite \eqref{exa} as
\begin{equation}\label{exa2}\begin{split}\tilde{K}_n(z,w)=\sqrt{\frac n {2\pi}}&e^{\frac n 2(\calQ(z)+\overline{\calQ(w)})}e^{-\frac n 2(Q(z)+Q(w))}e^{\frac 1 2(\calH(z)+\overline{\calH(w)})}\cr
&\times \sqrt{\phi'(z)}\overline{\sqrt{\phi'(w)}}\cdot \tilde{S}_n(z,w)\cdot (1+O(n^{-\beta})),\cr
\end{split}
\end{equation}
where we used the notation
\begin{equation}\label{tsn}\tilde{S}_n(z,w)=\sum_{j=n\theta_n}^{n-1} \rho_j(z,w)^j (\phi(z)\overline{\phi(w)})^j\end{equation}
with
\begin{equation}\label{roj}\rho_j(z,w)=\frac{\phi_{\tau(j)}(z)\overline{\phi_{\tau(j)}(w)}}{\phi(z)\overline{\phi(w)}}e^{\frac 1 {2\tau(j)}(\calQ_{\tau(j)}-\calQ)(z)+\frac 1 {2\tau(j)}\overline{(\calQ_{\tau(j)}-\calQ)(w)}}.\end{equation}

Our main task at hand is to estimate the sum $\tilde{S}_n(z,w)$. 

\begin{rem} In going from \eqref{exa} to \eqref{exa2} we used the facts that $\calH_\tau(z)=\calH(z)+O(1-\tau)$ and $\phi'_\tau(z)=\phi'(z)+O(1-\tau)$ as $n\to\infty$ where
the $O$-constants are uniform for $z$ in $U_\tau$ and (say) $n\theta_n\le \tau\le n$. This follows by an application of the maximum principle, using that the functions
are holomorphic on $\hat{\C}\setminus K$ and that relevant estimates are clear on the boundary curve $\Gamma_\tau$.
\end{rem}


We have the following main lemma, in which we fix a small number $\eta>0$.

\begin{lem} \label{fundlem} Suppose that $z,w\in N(U,\delta_n)$ and that
$|\phi(z)\overline{\phi(w)}-1|\ge\eta$. Then there is a positive constant $\power$ such that
$$\tilde{S}_n(z,w)=\frac {(\phi(z)\overline{\phi(w)})^n} {\phi(z)\overline{\phi(w)}-1}\cdot
\left(1+O\left(\frac {(\log n)^\power} {\sqrt{n}}\right)\right).$$

The constant $N$ as well as the $O$-constant can be chosen depending only on the parameters $\eta$ and $M$, and on the potential $Q$.
\end{lem}

Taken together with \eqref{exa2}, the lemma gives a convenient approximation formula for the tail $\tilde{K}_n(z,w)$. We shall later find that the full kernel $K_n(z,w)$ obeys the same asymptotic to a negligible error, for the set of $z$ and $w$ in question. 

\smallskip

We first turn to our proof of Lemma \ref{fundlem} in the following two subsections. After that, the proof of Theorem \ref{kern1} follows in Subsection \ref{follows}. 

\subsection{Preparation for the proof of Lemma \ref{fundlem}} For $\tau$ close to $1$ we introduce the following holomorphic function on $\hat{\C}\setminus K$,
\begin{equation}\label{bas1}F_{\tau}(z)=\frac {\phi_\tau(z)}{\phi(z)}e^{\frac 1 {2\tau}(\calQ_\tau-\calQ)(z)}.\end{equation}

Notice that
 $F_{\tau}(\infty)>0$ and that \eqref{roj} can be written
$$\rho_j(z,w)=F_{\tau(j)}(z)\overline{F_{\tau(j)}(w)},\qquad (\tau(j)=\frac j n).$$

For the purpose of estimating $\tilde{S}_n(z,w)$ we write
$$a_j=a_j(z,w)=(F_{\tau(j)}(z)\overline{F_{\tau(j)}(w)})^j$$
and
$$b_j=b_j(z,w)=(\phi(z)\overline{\phi(w)})^j.$$

We also denote $$m = \lfloor n\theta_n \rfloor,$$ the integer part of $n\theta_n$. 

Applying summation by parts, we write
\begin{equation}\label{sbp0}\tilde{S}_n(z,w)=\sum_{j=m}^{n-1}a_jb_j=a_{n-1}B_{n-1}-a_mB_{m-1}-\sum_{j=m}^{n-2}(a_{j+1}-a_j)B_j,\end{equation}
where
$$B_j=B_j(z,w)=\sum_{k=0}^j b_k=\frac {1-(\phi(z)\overline{\phi(w)})^{j+1}}{1-\phi(z)\overline{\phi(w)}}.$$

The proof of the following lemma is immediate from \eqref{sbp0}.

\begin{lem} \label{l42l} For all $z,w\in \C\setminus K$,
\begin{equation}\label{bstn}
\begin{split}\tilde{S}_n(z,w)&=a_{n-1}\frac {(\phi(z)\overline{\phi(w)})^n} {\phi(z)\overline{\phi(w)}-1}\\
&\,+\frac 1 {\phi(z)\overline{\phi(w)}-1}\sum_{j=m}^{n-2}(a_{j+1}-a_j)\cdot (\phi(z)\overline{\phi(w)})^{j+1}-a_m\frac {(\phi(z)\overline{\phi(w)})^{m+1}} {\phi(z)\overline{\phi(w)}-1}.
\end{split}\end{equation}
\end{lem}

We shall find below that $a_{n-1}\to 1$ and $a_m\to 0$ quickly as $n\to\infty$. Once this is done there remains to show that the penultimate term in the right hand side is negligible in comparison with the first one. This latter point is where our main efforts will be deployed.

\subsection{Proof of Lemma \ref{fundlem}} 
Throughout this subsection it is assumed that $z$ and $w$ belong to $N(U,\delta_n)$
and that $|\phi(z)\overline{\phi(w)}-1|\ge\eta$, and we write $\tau(j)=\frac j n$.

\smallskip

We begin with the following lemma.

\begin{lem} \label{vopp} Let $h(z)$ be the unique holomorphic function in a neighbourhood of $\overline{U}$ which
satisfies the boundary condition
\begin{equation}\label{bch}\re h(z)=-\frac {|\phi'(z)|^2}{4\Delta Q(z)},\qquad (z\in\Gamma)\end{equation}
and the normalization $\im h(\infty)=0.$ 

Then for all $z,w$ in a neighbourhood of $\overline{U}$ and all $j$ such that $\tau_0\le\tau(j)\le 1$ we have as $n\to\infty$
\begin{equation}\label{desa1}a_j(z,w)=\exp\{n(h(z)+\overline{h(w)})(1-\tau(j))^2+n(b_3(z)+\overline{b_3(w)})
(1-\tau(j))^3+n\cdot O(1-\tau(j))^4\},\end{equation}
where $b_3(z)$ is a holomorphic function in a neighbourhood of $\overline{U}$.

\end{lem}

Before proving the lemma, we note
that the harmonic function $\re h(z)$ defined by the boundary condition \eqref{bch} is strictly negative in a neighbourhood of $\overline{U}$ by the maximum principle. 

Hence Lemma \ref{vopp} implies the following result.

\begin{cor} \label{cl} By slightly increasing the compact set $K \subset \C\setminus \overline{U}$ if necessary, we can ensure that for all $z,w\in \C\setminus K$,
$$a_{n-1}(z,w)=1+O(n^{-1})$$
and there is a constant $s>0$ such that (with $m=n\theta_n$)
$$|a_m(z,w)|\lesssim e^{-s\log^2 n}.$$
Moreover, $s$ and the implied constants can be chosen uniformly for the given set of $z$ and $w$.
\end{cor}

\begin{proof}[Proof of Lemma \ref{vopp}] 
For $z\in \hat{\C}\setminus K$ and real $\tau$ near $1$ we consider the function
$$P(\tau,z):=\tau\log\left[\frac {\phi_\tau(z)}{\phi(z)}e^{\frac 1 {2\tau}(\calQ_\tau-\calQ)(z)}\right],$$
where we use the principal determination of the logarithm, i.e., $\im P(\tau,\infty)=0$.
It is clear that 
$$P(1,z)=0.$$ 

We now consider the Taylor expansion in $\tau$, about $\tau=1$,
\begin{equation}\label{bop0}P(\tau,z)=(1-\tau)\cdot b_1(z)+(1-\tau)^2\cdot b_2(z)+\cdots,\end{equation}
where $$b_k(z)=\frac {(-1)^k}{k!}\frac {\d^k}{\d \tau^k} P(\tau,z)\bigm|_{\tau=1}$$ is holomorphic
in $\hat{\C}\setminus K$ and $\im b_k(\infty)=0$.

Now using that
\begin{align*}V&=\re\calQ+\log|\phi|^2\\
V_\tau&=\re\calQ_\tau+\tau\log|\phi_\tau|^2
\end{align*}
we conclude that
\begin{align*}\frac 1 {2\tau}(V_\tau-V)(z)&=\frac 1 {2\tau}\re(\calQ_\tau-\calQ)(z)+\re\log \left(\frac {\phi_\tau}{\phi}\right)(z)-\frac {1-\tau}{2\tau}\log|\phi(z)|^2\\
&=\frac 1 {\tau}\re P(\tau,z)-\frac {1-\tau}{2\tau}\log|\phi(z)|^2,\end{align*}
which we write as
\begin{align}\label{bopper}\re P(\tau,z)=\frac 1 2(V_\tau-V)(z)+(1-\tau)\log|\phi(z)|,\qquad (z\in\C\setminus K).
\end{align}

If $z\in\Gamma$, this reduces to
\begin{align}\label{bopper2}\re P(\tau,z)=\frac 1 2(V_\tau-Q)(z),\qquad (z\in\Gamma),
\end{align}
whence by the asymptotics in Lemma \ref{slong}, we have as $\tau\to 1$,
\begin{align}\label{bopper3}\re P(\tau,z)=-\frac {|\phi'(z)|^2}{4\Delta Q(z)}\cdot (1-\tau)^2+b_3(z)(1-\tau)^3+O(1-\tau)^4,\qquad (z\in\Gamma).
\end{align}

Comparing with \eqref{bop0} we infer that the holomorphic functions $b_1$ and $b_2$ on $\hat{\C}\setminus K$ satisfy
$\re b_1=0$ on $\Gamma$ and
\begin{equation}\label{diricl}\re b_2(z)=-\frac {|\phi'(z)|^2}{4\Delta Q(z)},\qquad (z\in\Gamma).\end{equation}

The normalization at infinity determines $b_1=0$ and $b_2=h$ uniquely, where $h(z)$ is the function in the statement of the lemma. 

To finish the proof, it suffices to observe that
$$a_j(z,w)=\exp\left\{n(P(\tau(j),z)+\overline{P(\tau(j),w)})\right\}$$
and refer to \eqref{bopper3}.
\end{proof}

At this point, it is convenient to switch notation and write
$$k=n-j,$$
where then $1\le k\le \sqrt{n}\log n$. We will denote 
$$\mun=n-n\theta_n=\sqrt{n}\log n$$ and assume that this is an integer. We will also write 
\begin{equation}\label{epadef}\eps_k=1-\tau(j)=\frac k n,\qquad (1\le k\le \mun).\end{equation}

The following lemma is a direct consequence of Lemma \ref{vopp}.

\begin{lem} \label{smo} 
For $n-\mun\le j\le n-1$ we
have the asymptotic (as $n\to\infty$)
\begin{align*}(a_{j+1}&-a_j)(z,w)=e^{n(h(z)+\overline{h(w)})\eps_k^2}\\
&\times\left[ 
-2\eps_k(h(z)+\overline{h(w)})+O(n^{-1})+O(\eps_k^2)+n\eps_k^3 (b_3(z)+\overline{b_3(w)})+O(n\eps_k^4)+O(n^2\eps_k^6)\right].\end{align*}
\end{lem}

\begin{proof} This is immediate on writing
$$a_{j+1}-a_j=a_j\cdot(\frac {a_{j+1}}{a_j}-1),$$
noting that $n(\eps_{k-1}^2-\eps_{k}^2)=-2\eps_k+\frac 1 n$ and
inserting the asymptotics in Lemma \ref{vopp}; details are left for the reader.
\end{proof}

We are now ready to give our proof of Lemma \ref{fundlem}.

\begin{proof}[Proof of Lemma \ref{fundlem}] Pick $z,w\in N(U,\delta_n)$ and write
$$\phi(z)\overline{\phi(w)}=re^{i\vt},$$
where $r>0$ and $|\vt|\le\pi$. 

We are assuming that $|re^{i\vt}-1|\ge \eta>0$.
By continuity of the reflection in $\T$: $re^{i\vt}\mapsto r^{-1} e^{i\vt}$, there is also a constant $\eta_0=\eta_0(\eta)>0$ such that
\begin{equation}\label{dumbas}|r^{-1}e^{i\vt}-1|\ge\eta_0.\end{equation}

By Lemma \ref{lem13} we have in addition that
\begin{equation}\label{newC}r\ge 1-C\delta_n\end{equation}
for some constant $C$ depending only on $M$ and $Q$.

We now consider the sum
$$\sigma_n=\sum_{j=m}^{n-2}(a_{j+1}-a_j)r^{j+1}e^{i(j+1)\vt}.$$

In view of Lemma \ref{l42l} and Corollary \ref{cl}, we shall be done when we can prove the bound
\begin{equation}\label{desib}|\sigma_n|\lesssim \frac {(\log n)^\power}{\sqrt{n}}\, r^n\end{equation}
with some constant $N$.




Using Lemma \ref{smo}, it is seen that
\begin{equation}\label{al}\begin{split}\sigma_n=r^{n+1}e^{i(n+1)\vt}&\sum_{k=2}^{\mun}e^{nc\eps_k^2}r^{-n\eps_k}e^{in\vt\eps_k}\\
&\times (A\eps_k+O(n^{-1})+O(\eps_k^2)+Bn\eps_k^3+O(n\eps_k^4)+O(n^2\eps_k^6)).
\end{split}\end{equation}

Here $A,B,c$ are certain complex numbers depending on $z$ and $w$; the important fact is that
$$\re c<0.$$

To analyze the right hand side in \eqref{al}, we set
$$d=-\log r+i\vt$$
and
introduce the notation
\begin{align*}\sigma_{n,1}&=\sum_{k=2}^{\mun} \eps_k e^{nc\eps_k^2}e^{nd\eps_k},\\
\sigma_{n,2}&=n\sum_{k=2}^{\mun} \eps_k^3e^{nc\eps_k^2}e^{nd\eps_k}.
\end{align*}

From \eqref{dumbas} we have the lower bound
\begin{equation}\label{dfirst}|e^d-1|\ge \eta_0.\end{equation}

Also, since $r\ge 1-C\delta_n$ we have
\begin{equation}\label{dsecond}\re d\le\log\frac 1 {1-C\delta_n}\le C'\delta_n.\end{equation}

\smallskip

We now show that $\sigma_{n,1}$ and $\sigma_{n,2}$ are negligible as $n\to\infty$. 

To treat the case of $\sigma_{n,1}$ we write 
$$a_k^{(1)}=\eps_ke^{nc\eps_k^2},\qquad b_k^{(1)}=e^{nd\eps_k}=e^{dk},$$ so that
$$\sigma_{n,1}=\sum_{k=2}^\mun a_k^{(1)}b_k^{(1)}.$$

Using \eqref{dfirst} and \eqref{dsecond} we see that the partial sums 
$$B_k^{(1)}=\sum_2^k b_l^{(1)}=e^{2d}\frac {1-e^{(k-1)d}}{1-e^d}$$
obey the estimate
$|B_k^{(1)}|\lesssim e^{C\delta_n k}.$
In particular we have that
$$B_1^{(1)}=0,\qquad |B_\mun^{(1)}|\lesssim e^{CM\log n\sqrt{\log\log n}}.$$

Let us write
$$\alpha=-\re c>0.$$ 

Since $|a_\mun^{(1)}|\lesssim e^{-\alpha\log^2 n}$, a summation by parts gives
\begin{equation}\label{bra}\sigma_{n,1}=-\sum_{k=2}^{\mun-1}(a_{k+1}^{(1)}-a_k^{(1)})B_k^{(1)}+O(e^{-s\log^2 n}),
\end{equation}
with any $s$ satisfying $0<s<\alpha$.

Next observe that
$$a_{k+1}^{(1)}-a_k^{(1)}=e^{cn\eps_k^2}(\eps_{k+1}e^{c(2\eps_k+\frac 1 n)}-\eps_k),$$
whence
$$|a_{k+1}^{(1)}-a_k^{(1)}|\lesssim \frac 1 n e^{-\alpha n\eps_k^2}(1 +n\eps_k^2).$$

Making use of a Riemann sum and the substitution $t=\sqrt{n}\eps$, we get
\begin{align*}|\sigma_{n,1}|&\lesssim \frac 1 n\sum_{k=2}^{\mun-1}e^{Cn\delta_n\eps_k-\alpha n\eps_k^2}(1+n\eps_k^2)+O(e^{-s\log^2 n})\\
&\sim\int_0^{(\log n)/\sqrt{n}}e^{Cn\delta_n\eps-\alpha n\eps^2}(1+n\eps^2)\,d\eps+O(e^{-s\log^2 n})\\
&=\frac 1 {\sqrt{n}}\int_0^{\log n}e^{CMt\sqrt{\log\log n}-\alpha t^2}(1+t^2)\, dt+O(e^{-s\log^2 n})\\
&\lesssim \frac {(\log n)^{\tilde{C}^2/\alpha^2}}{\sqrt{n}},
\end{align*}
where we put $\tilde{C}=CM/2$.

In the case when $r$ is ``large'' in the sense that $r\ge r_0>1$, we can do better. Indeed since $\re d=-\log r$, the partial sums $B_k^{(1)}$ obey the bound $|B_k^{(1)}|\lesssim r^{-1}$ where the implied constant depends on $r_0$. The method of estimation above thus gives
\begin{align*}|\sigma_{n,1}|&\lesssim r^{-1}\left(\frac 1 n\sum_{k=2}^{\mun-1}e^{-\alpha n\eps_k^2}(1+n\eps_k^2)+O(e^{-s\log^2 n})\right)\\
&\lesssim \frac 1 {r\sqrt{n}},\qquad\qquad (n\to\infty,\, r\ge r_0>1).
\end{align*}

The term $\sigma_{n,2}$ can be handled similarly: we introduce the notation
$$a_k^{(2)}=n\eps_k^3e^{nc\eps_k^2},\qquad b_2^{(2)}=b_k^{(1)}=e^{dk},\qquad \sigma_{n,2}=\sum_{k=2}^\mu
a_k^{(2)}b_k^{(2)}.$$
One deduces without difficulty that
$$|a_{k+1}^{(2)}-a_k^{(2)}|\lesssim \frac 1 n e^{Cn\delta_n\eps_k-\alpha n\eps_k^2}(1+n\eps_k^2+n^2\eps_k^4).$$
A straightforward adaptation of our above estimates for $\sigma_{n,1}$ now leads to
\begin{align*}|\sigma_{n,2}|&\lesssim\int_0^{(\log n)/\sqrt{n}}e^{Cn\delta_n\eps-\alpha n\eps^2}(1+n\eps^2+n^2\eps^4)\, d\eps+O(e^{-s\log^2 n})\\
&=\frac 1 {\sqrt{n}}\int_0^{\log n}e^{CMt\sqrt{\log\log n}-\alpha t^2}(1+t^2+t^4)\, dt+O(e^{-s\log^2 n})\\
&\lesssim\frac {(\log n)^{\tilde{C}^2/\alpha^2}}{\sqrt{n}}.
\end{align*}

Moreover, in the case when $r\ge r_0>1$ we obtain the improved estimate $|\sigma_{n,2}|\lesssim 1/(r\sqrt{n})$.



\smallskip

The remaining terms in the right hand side of  \eqref{al} will be estimated in a more straightforward manner, by taking the absolute values inside the corresponding sums. 

Keeping the notation
$\alpha=-\re c>0$ we thus consider the following four terms: 
$$\tilde{\sigma}_{n,\nu}:=n^{\nu-1}\sum_{k=2}^{\mun}\eps_k^{2\nu}e^{-\alpha n\eps_k^2}e^{-n\eps_k\log r},\qquad \nu=0,1,2,3.$$

By a Riemann sum approximation and the estimate \eqref{newC} 
we find
\begin{align*}\tilde{\sigma}_{n,\nu}\sim n^\nu\int_{0}^{(\log n)/\sqrt{n}}&\eps^{2\nu}e^{-\alpha n\eps^2}e^{-n\eps\log r}\, d\eps\\
&\le \frac {1} {\sqrt{n}}\int_0^{\log n} t^{2\nu}e^{-\alpha t^2}e^{CMt\sqrt{\log\log n}}\, dt.
\end{align*}

Since (for $0\le\nu\le 3$)
$$\int_0^{\log n} t^{2\nu}e^{-\alpha t^2}e^{Ct\sqrt{\log\log n}}\, dt=O((\log n)^{\tilde{C}^2/\alpha^2}),$$
we conclude that 
$$\tilde{\sigma}_{n,\nu}\lesssim \frac {(\log n)^{\tilde{C}^2/\alpha^2}}{\sqrt{n}}.$$

It is also easy to verify that for $r\ge r_0>1$ we have $\tilde{\sigma}_{n,\nu}\lesssim 1/(r\sqrt{n})$. (For example, one can sum by parts as above, using that the summation index $k$ starts at $2$.)

\smallskip




All in all, by virtue of the relation \eqref{al},
we conclude the estimate (with a new $C$) 
\begin{equation}\label{fkc}|\sigma_n|\lesssim \frac {(\log n)^{C^2/\alpha^2}}
{\sqrt{n}}\, r^{n+1},\end{equation}
while if $r\ge r_0>1$,
\begin{equation}\label{fkc2}|\sigma_n|\lesssim \frac 1 {\sqrt{n}}\, r^n.\end{equation}

Combining these estimates, we find in all cases that $|\sigma_n|\lesssim \frac {(\log n)^{C^2/\alpha^2}}
{\sqrt{n}}\, r^n$.
We have reached the desired bound \eqref{desib} with $N=C^2/\alpha^2$, and our proof of Lemma \ref{fundlem} is complete.
\end{proof}

\subsection{Proof of Theorem \ref{kern1}} \label{follows}

In what follows we consider two arbitrary points $z,w\in N(U,\delta_n)$ such that $|\phi(z)\overline{\phi(w)}-1|\ge \eta$. 

\smallskip

Consider the \emph{full} reproducing kernel
$$K_n(z,w)=\sum_{j=0}^{n-1}W_{j,n}(z)\overline{W_{j,n}(w)}.$$
In view of Lemma \ref{fundlem} it suffices to
prove that $K_n(z,w)$ is, in a suitable sense, ``close'' to the tail kernel $\tilde{K}_n(z,w)$.


\smallskip

To prove this we first note that Lemma \ref{fundlem} implies that the size of the tail-kernel is 
\begin{equation}\label{tequ}|\tilde{K}_n(z,w)|\asymp \sqrt{n}\,e^{\frac n 2(V-Q)(z)+\frac n 2(V-Q)(w)}.
\end{equation}

To estimate lower order terms, corresponding to $j$ with $\tau(j)\le \theta_n$, we recall
Lemma \ref{jugo} that
there is a number $c'>0$ such that for all $z\in N(U,\delta_n)$
\begin{equation}\label{vu1}|W_{j,n}(z)|\le Ce^{-c'\log^2 n} e^{\frac n 2(\check{Q}-Q)(z)},\qquad (\tau(j)\le\theta_n).\end{equation}

Using a similar estimate for $W_{j,n}(w)$ and picking any $c>0$ with $c<c'$, we conclude the estimate
\begin{equation}\label{spuck}\begin{split}\sum_{j=0}^{n\theta_n}|W_{j,n}(z)W_{j,n}(w)|&\lesssim ne^{-c' \log^2 n}e^{\frac n 2(\check{Q}-Q)(z)+\frac n2(\check{Q}-Q)(w)}\\
&\lesssim e^{-c\log^2 n}e^{\frac n 2(\check{Q}-Q)(z)+\frac n2(\check{Q}-Q)(w)}
.\end{split}
\end{equation}

Since $\check{Q}=V$ on $U$, we obtain from \eqref{tequ} and \eqref{spuck} that $K_n(z,w)= \tilde{K}_n(z,w)\cdot (1+O(e^{-\frac c2\log^2 n}))$ in the case when both $z$ and $w$ are in $\overline{U}$. However, since
$z$ and $w$ are allowed to vary in the $\delta_n$-neighbourhood, we require a slight extra argument. 

We shall use the following simple lemma, which also appears implicitly in the proof of \cite[Lemma 6.6]{AKS}.

\begin{lem} There is a constant $C$ such that for all $z\in N(U,\delta_n)$,
\begin{equation}\label{66}(\check{Q}-V)(z)\le C\delta_n^2.\end{equation}
\end{lem}

\begin{proof} Since $\check{Q}=V$ on $U$, we can assume that $z\in S$. Then $\check{Q}(z)=Q(z)$. Let $p\in \Gamma$ be the closest point and write $z=p+\ell\normal_1(p)$ where $|\ell|\lesssim \delta_n$ by Lemma \ref{lem13}. 
The Taylor expansion in Lemma \ref{short} now shows that $(Q-V)(z)=2\Delta Q(p)\ell^2+O(\ell^3)$, finishing the proof of the claim. \end{proof}

Combining \eqref{spuck} with \eqref{66} we conclude that if $z,w\in N(U,\delta_n)$ then
\begin{equation*}\sum_{j=0}^{n\theta_n}|W_{j,n}(z)W_{j,n}(w)|\lesssim e^{-c \log^2 n+c'\log\log n}e^{\frac n 2(V-Q)(z)+\frac n2(V-Q)(w)},\end{equation*}
for a suitable positive constant $c'$. Fix $c''$ with $0<c''<c$ and then pick a new $c>0$ with $c<c''$.
Comparing with
\eqref{tequ}, we obtain
\begin{align*}|\sum_{j=0}^{n\theta_n}W_{j,n}(z)\overline{W_{j,n}(w)}|&\lesssim e^{-c'' \log^2 n}e^{\frac n 2(V-Q)(z)+\frac n2(V-Q)(w)}\\
&\lesssim  e^{-c \log^2 n}|\tilde{K}_n(z,w)|.\end{align*}

We have shown that
\begin{align}\label{final}K_n(z,w)&=\tilde{K}_n(z,w)\cdot (1+O(e^{-c\log^2 n})).\end{align}

By Lemma \ref{fundlem}, we know that the tail kernel $\tilde{K}_n(z,w)$ has the desired asymptotic when the points $z,w$ belong to $N(U,\delta_n)$ and $|\phi(z)\overline{\phi(w)}-1|\ge \eta$. Thus by \eqref{final}
we find that $K_n(z,w)$ obeys the same asymptotic, finishing our proof of Theorem \ref{kern1}.
q.e.d.

\subsection{Proof of Theorem \ref{berth1}} 

Fix a point $z\in U$ and recall that
$$d\mu_{n,z}(w)=B_n(z,w)\, dA(w),\qquad B_n(z,w)=\frac {|K_n(z,w)|^2}{K_n(z,z)}.$$

We express points $w$ in $\C$ (in some fixed neighbourhood of $\Gamma$) as
$w=p+\ell\cdot\normal_1(p)$
where $p=p(w)$ is a point on $\Gamma$, $\normal_1(p)$ is the unit normal to $\Gamma$ pointing outwards from $S$ and $\ell$ is a real parameter.

Given a point $p\in \Gamma$ we also recall the Gaussian probability measure $\gamma_{p,n}$ on the real line,
\begin{equation}d\gamma_{p,n}(\ell)=\frac {\sqrt{4n\Delta Q(p)}}{\sqrt{2\pi}}e^{-2n\Delta Q(p)\ell^2}\, d\ell.\end{equation}

Denote by $\omega_z=\omega_{z,U}$ the harmonic measure of $U$ evaluated at $z$ and consider the measure $d\tilde{\mu}_{n,z}=d\omega_z(p)\, d\gamma_{n,p}(\ell)$;
writing
$d\omega_z(p)=P_z(p)\, |dp|$, we have
\begin{equation*}d\tilde{\mu}_{n,z}(p+\ell\cdot\normal_1(p))=P_{z}(p)\frac {\sqrt{4n\Delta Q(p)}}{\sqrt{2\pi}}e^{-2n\Delta Q(p)\ell^2}\, |dp|\,d\ell.\end{equation*}

\smallskip

By Theorem \ref{kern1} we have,
for fixed $z\in U$ and any $w\in N(U,\delta_n)$,
\begin{align*}B_n(z,w)&=\frac {|K_n(z,w)|^2}{K_n(z,z)}=\frac {\sqrt{n}}{\sqrt{2\pi}}e^{\re\calH(w)}\frac {|\phi(z)|^2-1}{|\phi(z)\overline{\phi(w)}-1|^2}\\
&\qquad \times |\phi'(w)||\phi(w)|^{2n}e^{n(\re\calQ(w)-Q(w))}\cdot (1+o(1)).
\end{align*}

Recalling that $|\phi(w)|^{2n}e^{n\re\calQ(w)}=e^{nV(w)}$, we obtain
\begin{align}\label{bnb}B_n(z,w)=(1+o(1))\frac {\sqrt{n}}{\sqrt{2\pi}}e^{\re\calH(w)}\frac {|\phi(z)|^2-1}{|\phi(z)\overline{\phi(w)}-1|^2}
|\phi'(w)|e^{-n(Q-V)(w)}.\end{align}

\smallskip

We next recall that (by Lemma \ref{long}),
the factor $e^{-n(Q-V)(w)}$ is negligible when $\dist(w,\Gamma)\ge \delta_n$, so we can focus on the asymptotics of \eqref{bnb} in the $\delta_n$-neighbourhood $N(\Gamma,\delta_n)$.

Near the curve $\Gamma$, Lemma \ref{short} gives
\begin{equation}\label{t1000}(Q-V)(p+\ell\cdot \normal_1(p))=2\Delta Q(p)\cdot \ell^2+O(\ell^3),\qquad (p\in\Gamma,\quad \ell\to 0),\end{equation}
where $O$-constant is independent of the point $p\in\Gamma$.

Using \eqref{t1000} and \eqref{bnb} we infer that, when $|\ell|\le \delta_n$,
\begin{equation}\label{infer}B_n(z,w)=(1+o(1))\frac {\sqrt{n}}{\sqrt{2\pi}}\sqrt{\Delta Q(w)}\frac {|\phi(z)|^2-1}{|\phi(z)\overline{\phi(w)}-1|^2}|\phi'(w)|e^{-2n\Delta Q(w)\ell^2}.\end{equation}

We now change variables from $(p,\ell)\in \Gamma\times \R$ to $(\theta,t)\in\T\times\R$ by the inverse of the mapping
\begin{equation}f_n:\T\times\R\to\Gamma\times \R\qquad ,\qquad (\theta,t)\mapsto (p,\ell):=(\phi^{-1}(e^{i\theta}),\frac t {2\sqrt{n}}).\end{equation}

In these coordinates, \eqref{infer} becomes
$$B_n(z,w)=(1+o(1))\frac {\sqrt{n}}{\sqrt{2\pi}}\sqrt{\Delta Q(p)}\frac {|\phi(z)|^2-1}{|\phi(z)-\phi(p)|^2}|\phi'(p)|e^{-\frac 1 2 \Delta Q(p) t^2}.$$

An easy computation shows that
$$dA(w)=(1+o(1))\frac 1 {2\pi\sqrt{n}}\frac 1 {|\phi'(p)|}\, d\theta\, dt,$$
whence the pull-back measure
$\mu_{n,z}\circ f_n$
satisfies
\begin{equation}\label{bob}d(\mu_{n,z}\circ f_n)(\theta,t)=(1+o(1))\cdot \frac 1 {2\pi} \frac {|\phi(z)|^2-1}{|\phi(z)-e^{i\theta}|^2}\, d\theta\times \frac {\sqrt{\Delta Q(p)}}
{\sqrt{2\pi}}e^{-\frac 1 2 \Delta Q(p)t^2}\, dt.\end{equation}
(The convergence $o(1)\to 0$ holds in the uniform sense of densities on the sets where $|t|\le 2M\sqrt{\log\log n}$.)

At this point we notice that the measure
$$d\omega_{\phi(z)}(\theta)=\frac 1 {2\pi} \frac {|\phi(z)|^2-1}{|\phi(z)-e^{i\theta}|^2}\, d\theta$$
is precisely the harmonic measure for $\D_e$ evaluated at the point $\phi(z)\in\D_e$ (cf. \cite{GM}).

Pulling back the left and right hand sides in \eqref{bob} by the inverse $f_n^{-1}(p,\ell)=(\phi(p),2\sqrt{n}\ell)$ and using conformal invariance of the harmonic measure,
we infer that the measure $\mu_{n,z}$ satisfies
$$d\mu_{n,z}(p+\ell\cdot \normal_1(p))=(1+o(1))\, d\omega_z(p)\times d\gamma_{n,p}(\ell),$$
and that the uniform convergence on the level of densities asserted in \eqref{berth2} holds. (The factor $\tfrac 1 \pi$ in the left hand side comes from our normalization of the area measure $dA$.)
 q.e.d.

\section{Proof of Lemma \ref{main0}} \label{details}

In this section, we provide a detailed proof of Lemma \ref{main0} on tail kernel approximation, based on ideas from \cite{HW}. The main point is to give a 
derivation which leads to our desired estimates with minimal fuss, and in precisely the form that we want them. Aside from this, we believe that the following exposition could be of value for other investigations where the main interest is in leading order asymptotics.

When working out the details of this section, in addition to the original paper \cite{HW}, we were inspired by  \cite{AKS}, for example.

\smallskip

To briefly recall the setup, we take $\{W_{j,n}\}_{j=0}^{n-1}$ to be the orthonormal basis for the weighted polynomial subspace $\calW_n$ of $L^2$ with
 $W_{j,n}=P_{j,n}\cdot e^{-\frac 12nQ},$
 where the polynomial $P_{j,n}$ has degree $j$ and positive leading coefficient. The tail kernel $\tilde{K}_n(z,w)$ is then given by
\begin{equation}\label{tn}\tilde{K}_n(z,w)=\sum_{j=n\theta_n}^{n-1}W_{j,n}(z)\overline{W_{j,n}(w)},\qquad (\theta_n=1-\frac {\log n}{\sqrt{n}}).\end{equation}

As always, we write
\begin{equation*}N(U,\delta_n)=U+D(0,\delta_n),\qquad \delta_n=M\sqrt{\frac {\log\log n} n},\end{equation*}
where $U$ is the component of $\hat{\C}\setminus S$ containing $\infty$.

\subsection{Reduction of the problem}

Fix numbers $\tau_0<1$, $\epsilon>0$ and a compact subset $K\subset \Int \Gamma_{\tau_0-\epsilon}$ with the properties in Subsection \ref{lgrm}.
Also fix $j$ and $n$ such that $\tau_0\le \tau(j)\le 1$, where (as always) $\tau(j)=j/n$.

Following \cite{HW} we define an approximation of $W_{j,n}(z)$ on $\C\setminus K$ by
\begin{equation}\label{qp0} W^\sharp_{j,n}(z)=F_{j,n}(z)\cdot e^{-\frac 1 2 nQ(z)},\end{equation}
where
\begin{equation}\label{fp0}F_{j,n}(z)=\left(\frac n {2\pi}\right)^{\frac 1 4}e^{\frac 1 2 \calH_{\tau(j)}(z)}\sqrt{\phi_{\tau(j)}'(z)}\, \phi_{\tau(j)}(z)^j\, e^{\frac 1 2 n\calQ_{\tau(j)}(z)}.\end{equation}

Here $\calH_\tau$ and $\calQ_\tau$ are bounded holomorphic functions on $\hat{\C}\setminus K$ with
$\re\calH_\tau=\log\sqrt{\Delta Q}$ and $\re\calQ_\tau=Q$ on $\Gamma_\tau$; $\phi_\tau$ is the univalent extension to $\hat{\C}\setminus K$ of the normalized conformal map $U_\tau\to\D_e$.

\begin{lem} \label{spock} \textrm{(``Main approximation formula'')} The number $\tau_0<1$ may be chosen so that if $\tau_0\le \tau(j)\le 1$ and if $\beta$ is any number in the range $0<\beta<\tfrac 14$ then as $n\to\infty$,
$$W_{j,n}(z)=W^\sharp_{j,n}(z)\cdot (1+O(n^{-\beta})),\qquad z\in N(U,\delta_n).$$
\end{lem}

It is clear that Lemma \ref{spock} implies Lemma \ref{main0} on asymptotics for the tail kernel $\tilde{K}_n(z,w)$.

The rest of this section is devoted to a proof of Lemma \ref{spock}.

\subsection{Foliation flow} One of the key ideas in \cite{HW} is to introduce a set of ``flow coordinates'' to facilitate computations.

In the following suppose that $\tau_0\le \tau\le 1$; it will be convenient to write
\begin{equation}\label{dnp}\eps_n=\frac {\log n}{\sqrt{n}}.\end{equation}

Fix a small $\delta>0$. For a small real parameter $t$,
we denote by $L_{\tau,t}$ the level set
$$L_{\tau,t}=\{z\in N(\Gamma_\tau,\delta)\,;\, (Q-V_\tau)(z)=t^2\}.$$
Of course $L_{\tau,0}=\Gamma_\tau$.

By Lemma \ref{short} we see that for small $t\ne 0$, $L_{\tau,t}$ is the disjoint union of two analytic Jordan curves $L_{\tau,t}=\Gamma_{\tau,t}^-\cup \Gamma_{\tau,t}^+$ where
$\Gamma_{\tau,t}^-\subset \Int\Gamma_\tau$ and $\Gamma_{\tau,t}^+\subset\Ext\Gamma_\tau$. We set $\Gamma_{\tau,t}=\Gamma_{\tau,t}^-$ if $t\le 0$ and $\Gamma_{\tau,t}=\Gamma_{\tau,t}^+$ if
$t\ge 0$. 


Let $U_{\tau,t}$ be the exterior domain of $\Gamma_{\tau,t}$ and consider the simply connected domain $\phi_\tau(U_{\tau,t})\subset\hat{\C}$. 

Also denote
by
$$\psi_t=\psi_{\tau,t}:\D_e\to\phi_\tau(U_{\tau,t})$$
the normalized conformal mapping (i.e., $\psi_t(\infty)=\infty$ and $\psi_t'(\infty)>0$). Thus $\psi_0(z)=z$ and $\psi_t$ is to be regarded as a slight perturbation of the identity.

Note that $\psi_t$ continues analytically across $\T$ and obeys the basic relation
\begin{equation}\label{barel}(Q-V_\tau)\circ\phi_\tau^{-1}\circ\psi_t\equiv t^2\qquad \text{on}\qquad \T.\end{equation}

Indeed, our definitions have been set up so that, for all large $n$,
\begin{equation}\Gamma_{\tau,t}=\phi_\tau^{-1}\circ\psi_t(\T),\qquad (-2\eps_n\le t\le 2\eps_n).\end{equation}

With $\tau(j)=j/n$, we define a neighbourhood $D_{j,n}$ of $\T$ by
\begin{equation}\label{flowd}D_{j,n}=\bigcup_{-2\eps_n\le t\le 2\eps_n}\psi_{\tau(j),t}(\T).\end{equation}

The inverse image $\phi_{\tau(j)}^{-1}(D_{j,n})$ plays the role of an ``essential support'' for $W_{j,n}$ and $W_{j,n}^\sharp$.

\subsection{Approximation scheme} In the following we fix $j$ and $n$ with $\tau_0\le \tau(j)\le 1$, where $\tau(j)=j/n$.
We shall extend $W_{j,n}^\sharp$ to a smooth function on $\C$ by a straightforward cut-off procedure.

It is convenient to modify the compact set $K\subset \Int \Gamma_{\tau_0-\epsilon}$
so that $\phi_{\tau(j)}$ maps $\hat{\C}\setminus K$ biholomorphically
onto some exterior disc $\D_e(\rho_0-\delta)$ where $\rho_0<1$ and $\delta>0$. (Then $K=K(j)$ may slightly vary with $j$,
but it will be harmless to suppress the $j$-dependence in our notation.)

Next we fix a smooth function $\chi_0$ such that $\chi_0=0$ on $K$ and $\chi_0=1$ on $\phi_{\tau(j)}^{-1}(\D_e(\rho_0))$ and define
\begin{equation}\label{qddef}W_{j,n}^\sharp=\chi_0\cdot F_{j,n}\cdot e^{-\frac 12nQ}.\end{equation}
(It is understood that $W_{j,n}^\sharp=0$ on $K$.)


The following properties of the function $W_{j,n}^\sharp$ are key for what follows:

\begin{enumerate}
\item \label{gu} $W_{j,n}^\sharp$ is asymptotically normalized:
$\|W_{j,n}^\sharp\|=1+O(\eps_n)$ as $n\to\infty$.
\item \label{gr} $W_{j,n}^\sharp$ is approximately orthogonal to lower order terms: $|(W,W_{j,n}^\sharp)|\le Cn^{-\frac 12}\|W\|$ for any $W=P\cdot e^{-\frac 12nQ}\in\calW_n$ with $\degree P<j$.
\end{enumerate}

\subsection{Positioning and the isometry property} \label{poip} Continuing in the spirit of \cite{HW}, we define the ``positioning operator'' $\Lambda_{j,n}$ by
$$\Lambda_{j,n}[f]=\phi_{\tau(j)}'\cdot\phi_{\tau(j)}^j\cdot e^{\frac 12n\calQ_{\tau(j)}}\cdot f\circ\phi_{\tau(j)}.$$
Also define a function (``$\tau(j)$-ridge'') by
$$R_{\tau(j)}=(Q-V_{\tau(j)})\circ \phi_{\tau(j)}^{-1}$$
where $V_\tau$ is the harmonic continuation of $\check{Q}_\tau\Big|_{U_\tau}$ inwards across $\Gamma_\tau$.

The map $\Lambda_{j,n}$ is then an isometric isomorphism
$$\Lambda_{j,n}:L^2_{nR_{\tau(j)}}(\D_e(\rho_0))\to L^2_{nQ}(\phi_{\tau(j)}^{-1}(\D_e(\rho_0)))$$ which preserves holomorphicity.
(Here and in what follows, the norm in the weighted $L^2$-space $L^2_\phi(\Omega)$ is, by definition, $\|f\|_\phi^2=\int_\Omega |f|^2 e^{-\phi}\, dA$.)

In particular we have the following ``isometry property'',
\begin{equation}\label{isop}\int_{\phi_{\tau(j)}^{-1}(\D_e(\rho_0))}\Lambda_{j,n}[f]\overline{\Lambda_{j,n}[g]}e^{-nQ}\, dA=\int_{\D_e(\rho_0)}f\bar{g}e^{-nR_{\tau(j)}}\, dA,
\quad (f,g\in L^2_{nR_{\tau(j)}}(\D_e(\rho_0))).\end{equation}

We now define a function $f_{j,n}$ on $\D_e(\rho_0)$ by
$$n^{\frac 1 4}\Lambda_{j,n}[f_{j,n}]=F_{j,n}.$$
This gives
\begin{equation}\label{g0}f_{j,n}=(2\pi)^{-\frac 1 4}((\phi_{\tau(j)}')^{-\frac 1 2}\cdot e^{\frac 12\calH_{\tau(j)}})\circ \phi_{\tau(j)}^{-1}.\end{equation}





\begin{lem} \label{handy} With
$\delta_\tau(z)=\dist(z,\Gamma_\tau)$, we have for all $z\in \C\setminus K$
$$|F_{j,n}(z)|^2e^{-nQ(z)}\le C
\sqrt{n}e^{-cn\min\{\delta_{\tau(j)}(z)^2,1\}}$$
where $C$ and $c$ are positive constants.
\end{lem}

\begin{proof} Observe that
$$|F_{j,n}(z)|^2e^{-nQ(z)}=\sqrt{\frac n {2\pi}}e^{\re\calH_{\tau(j)}(z)}|\phi_{\tau(j)}'(z)|e^{-n(Q-V_{\tau(j)})(z)}$$
and use Lemma \ref{long}.
\end{proof}



\subsection{Integration in flow-coordinates}
Define a domain $\tilde{D}_{n}$ in coordinates $(t,w)\in \R\times \T$ by
\begin{equation}\label{djn}\tilde{D}_{n}=\{(t,w)\, ;\,w\in \T\, ,\, -2\eps_n\le t\le 2\eps_n\}.\end{equation}
Now fix $j$ with $\tau_0\le \tau(j)\le 1$ and recall the definition of the flow domain $D_{j,n}$ in \eqref{flowd}.

Following \cite{HW} we define a flow map $\Psi:\tilde{D}_n\to D_{j,n}$ by $\Psi(t,w)=\psi_t(w)$, where $\psi_t=\psi_{\tau(j),t}$. The Jacobian of the map
$(t,w)\mapsto (\re\Psi,\im\Psi)$ is calculated as
$$J_\Psi(t,w)=\re(\d_t\psi_t)\cdot\im(\psi_t'\cdot iw)-\im(\d_t\psi_t)\cdot \re(\psi_t'\cdot iw)=\re(\overline{w\cdot \psi_t'(w)}\cdot\d_t\psi_t(w)).$$

\begin{lem} \label{bandy} As $t\to 0$ we have
\begin{equation}\label{band1}J_\Psi(t,w)=(\frac {|\phi_{\tau(j)}'|}{\sqrt{2\Delta Q}})\circ\phi_{\tau(j)}^{-1}(w)+O(t).\end{equation}
Moreover with $f_{j,n}$ given by \eqref{g0} we have
\begin{equation}\label{band2}|f_{j,n}\circ\psi_t(w)|^2J_\Psi(t,w)
=\frac 1 {2\sqrt{\pi }}\cdot (1+O(t)).\end{equation}
\end{lem}

\begin{proof} Write $\psi_t(w)=w\cdot(1+t\hat{\psi}_1(w)+O(t^2))$.
Since $R_{\tau(j)}\circ\psi_t=t^2$ we obtain by Taylor's formula (Lemma \ref{short}) that
$$t^2=R_{\tau(j)}\circ\psi_t=(2\Delta Q\cdot |\phi_{\tau(j)}'|^{-2})\circ \phi_{\tau(j)}^{-1}(w)\cdot (\re\hat{\psi}_1(w))^2t^2+O(t^3),$$
so, since $\re\hat{\psi}_1(w)>0$,
$$\re\hat{\psi}_1(w)=(\frac {|\phi_{\tau(j)}'|}{\sqrt{2\Delta Q}})\circ\phi_{\tau(j)}^{-1}(w),$$
and consequently
$$J_\Psi(0,w)=\re(\bar{w}\cdot \d_t\psi_t(w))|_{t=0}=\re\hat{\psi}_1(w)=(\frac {|\phi_{\tau(j)}'|}{\sqrt{2\Delta Q}})\circ\phi_{\tau(j)}^{-1}(w).$$

This proves \eqref{band1}; to prove \eqref{band2} we set $p=\phi_{\tau(j)}^{-1}(w)$ and compute
$$|f_{j,n}\circ\psi_t(w)|^2J_\Psi(t,w)= \frac 1 {\sqrt{2\pi}}\cdot \frac {\sqrt{\Delta Q(p)}}{|\phi_{\tau(j)}'(p)|}\cdot \frac {|\phi_{\tau(j)}'(p)|}
{\sqrt{2\Delta Q(p)}}+O(t).$$
\end{proof}

It follows that
if $f$ is an integrable function on $D_{j,n}$ then 
\begin{equation}\label{jacc}\int_{D_{j,n}}f\, dA=\frac 1 \pi\int_{\tilde{D}_n}f\circ\phi^{-1}\circ\psi_t\cdot(1+O(t))
\cdot |\d_t\psi_t(w)|
\, dt\, |dw|,\end{equation}
so by Lemma \ref{bandy}, 
\begin{equation}\label{jaccp}\int_{D_{j,n}}f\, dA=\frac 1 {2\pi\sqrt{\pi}}\int_{\tilde{D}_n}f\circ\phi^{-1}\circ \psi_t\cdot (1+O(t))\, dt\,|dw|.\end{equation}


Taking 
$f\circ\phi^{-1}=|f_{j,n}|^2e^{-nR_{\tau(j)}}$ and using 
that $R_{\tau(j)}\circ\psi_t=t^2$ on $\T$, we now see that
\begin{align*}\int_{D_{j,n}}|W_{j,n}^\sharp|^2&=\sqrt{n}\int_{D_{j,n}}|f_{j,n}|^2e^{-nR_{\tau(j)}}\\
&=\frac 1 {2\pi}\frac {\sqrt{n}} {\sqrt{\pi}}\int_{\tilde{D}_n}(1+O(\eps_n))e^{-nt^2}\, dt\,|dw|\\
&=1+O(\eps_n).
\end{align*}

Hence
\begin{align*}\int_\C |W_{j,n}^\sharp|^2&=1+O(\eps_n)+\sqrt{n}\int_{\C\setminus D_{j,n}}
(\chi_0\circ \phi_{\tau(j)}^{-1})^2|f_{j,n}|^2e^{-nR_{\tau(j)}},\end{align*}
and by Lemma \ref{handy}, the last term on the right is $O(\sqrt{n}e^{-c\log^2 n})$ for a suitable constant $c>0$.

We have shown the approximate normalization property \eqref{gu}, i.e., we have shown:

\begin{lem} \label{sl0} If $\tau(j)\in [\tau_0,1]$ then
$\|W_{j,n}^\sharp\|=1+O(\eps_n)$ as $n\to\infty.$
\end{lem}

\subsection{Approximate orthogonality} We now prove property \eqref{gr} of the quasipolynomials.

Given a positive integer $k$, it is convenient to write
$\calW_{k,n}$ for the space of weighted polynomials $W=P\cdot e^{-\frac 12nQ}$ where $P$ has degree at most $k$, equipped with the usual $L^2$-norm.

\begin{lem} \label{sl1} Suppose that $\tau_0\le\tau(j)\le 1$. Then for all $W\in\calW_{j-1,n}$ we have
$$\left|\int_\C W_{j,n}^\sharp\cdot \bar{W}\, dA\right|\le Cn^{-\frac 1 2}\|W\|.$$
\end{lem}

\begin{proof}

 Let $W=Pe^{-\frac 12nQ}$ where $P$ has degree $\ell<j$. Write $q=\Lambda_{j,n}^{-1}[P]$;
then $q$ is holomorphic on $\D_e(\rho_0)$ and satisfies $q(z)=O(z^{\ell-j})$ as $z\to\infty$.

By the Cauchy-Schwarz inequality and Lemma \ref{handy} we conclude that
\begin{align*}|\int_{\C\setminus \phi_{\tau(j)}^{-1}(D_{j,n})}W_{j,n}^\sharp\bar{W}|&\le \|W\|(\int_{\C\setminus\phi_{\tau(j)}^{-1}(D_{j,n})}\chi_0^2|F_{j,n}|^2e^{-nQ})^{1/2}\\
&\le Cn^{\frac 1 4}e^{-c\log^2 n}\|W\|.
\end{align*}

Hence it suffices to estimate the integral
\begin{align}\label{i0}I=\int_{\phi_{\tau(j)}^{-1}(D_{j,n})}W\bar{W}_{j,n}^\sharp=
\int_{\phi_{\tau(j)}^{-1}(D_{j,n})}P\bar{F}_{j,n}e^{-nQ}=n^{\frac 1 4}\int_{D_{j,n}}h\cdot |f_{j,n}|^2e^{-R_{\tau(j)}}\end{align}
where $h=q/f_{j,n}$ is holomorphic of $\D_e(\rho_0)$ and vanishes at infinity (since $f_{j,n}(\infty)>0$).

By \eqref{jaccp},
\begin{align*}I=\frac {n^{\frac 1 4}} {2\pi\sqrt{\pi}}\int_{\tilde{D}_{n}}h\circ \psi_t(w)\cdot (1+O(t))\,e^{-n t^2}\, dt\, |dw|.\end{align*}

But
$$\int_\T h\circ\psi_t(w)\, |dw|=h\circ\psi_t(\infty)=0$$
by the mean-value property of holomorphic function, so we obtain the estimate
$$|I|\lesssim n^{\frac 1 4} \int_{\tilde{D}_{n}}|h\circ \psi_t(w)||t|e^{-n t^2}\, dt\, |dw|.$$

Since $1/f_{j,n}$ is bounded on $\tilde{D}_{n}$  we see that
$$|I|\lesssim n^{\frac  1 4}  \int_{\tilde{D}_{n}}|q\circ \psi_t(w)||t|e^{-nt^2}\, dt\, |dw|.$$
Using the Cauchy-Schwarz inequality the right hand side is estimated by
\begin{align*}
 &C_1n^{\frac 1 4}(\int_{-\infty}^{+\infty}t^2 e^{-nt^2}\, dt)^{\frac 12}(\int_{D_{j,n}}|q|^2e^{-nR_{\tau(j)}}\, dA)^{\frac 12}\\
&=C_2n^{-\frac 1 2}(\int_{\phi_{\tau(j)}^{-1}(D_{j,n})}|P|^2e^{-nQ}\, dA)^{\frac 12}\le C_2n^{-\frac 1 2}\|W\|,\end{align*}
where we used the isometry property \eqref{isop} to deduce the equality.
\end{proof}

\subsection{Pointwise estimates} We wish to show that when $\tau(j)$ is close to $1$, then $W_{j,n}$ is ``pointwise close'' to
$W_{j,n}^\sharp$ near the curve $\Gamma_{\tau(j)}$.


\begin{lem} \label{l2} There are constants $C$ and $n_0$ such that for all $n\ge n_0$ and all $j$ with $\tau_0\le \tau(j)\le 1$  we have
$\|W_{j,n}-W_{j,n}^\sharp\|\le C\eps_n.$
\end{lem}

\begin{proof} Let $u_0$ be the norm-minimal solution in $L^2(e^{-nQ},dA)$ to the following $\dbar$-problem:
\begin{enumerate}[label=(\roman*)]
\item \label{fi} $\dbar u=F_{j,n}\cdot \dbar\chi_0$ on $\C$,
\item \label{fii} $u(z)=O(z^{j-1})$ as $z\to\infty$.
\end{enumerate}

A standard estimate found in \cite[Section 4.2]{H} shows that there is a constant $C$ such that
$$\|u_0\|_{L^2(e^{-nQ})}^2\le \frac C n\int_\C |(\dbar\chi_0)\cdot F_{j,n}|^2e^{-nQ}.$$

Since $\dbar\chi_0=0$ on $U_{\tau_0-\epsilon}$, Lemma \ref{handy} implies that there is a constant $c>0$ such that $|F_{j,n}|^2e^{-nQ}\le e^{-cn}$ on the support of $\dbar\chi_0$.
Thus \begin{equation}\label{es1}\|u_0\|_{L^2(e^{-nQ})}\le Ce^{-cn}\end{equation}
with (new) positive constants $C$ and $c$.

We correct $F_{j,n}\cdot \chi_0$ to a polynomial $\tilde{P}_{j,n}$ of exact degree $j$ by setting
$$\tilde{P}_{j,n}=F_{j,n}\cdot\chi_0-u_0.$$
($\tilde{P}_{j,n}$ is then an entire function of exact order of growth $O(z^j)$ as $z\to\infty$, since $|F_{j,n}(z)|\asymp |z|^j$ and $|u_0(z)|\lesssim |z|^{j-1}$ as $z\to\infty$, so indeed
$\tilde{P}_{j,n}$ is a polynomial of exact degree $j$.)

It follows from \eqref{es1} that
\begin{equation}\label{es2}\|\tilde{P}_{j,n}-F_{j,n}\cdot \chi_0\|_{L^2(e^{-nQ})}\le Ce^{-cn}.\end{equation}

Recall that $W_{j,n}^\sharp=\chi_0\cdot F_{j,n}\cdot e^{-\frac 12nQ}$, set
$\tilde{W}_{j,n}=\tilde{P}_{j,n}\cdot e^{-\frac 12nQ}$, and
note that \eqref{es2} says that
\begin{equation}\label{es3}\|\tilde{W}_{j,n}-W_{j,n}^\sharp\|\le Ce^{-cn}.\end{equation}

By Lemma \ref{sl0} we have $\|W_{j,n}^\sharp\|=1+O(\eps_n)$ and so by \eqref{es3},
\begin{equation}\label{no0}\|\tilde{W}_{j,n}\|=1+O(\eps_n).\end{equation}
Similarly,
the approximate orthogonality in Lemma \ref{sl1} implies (with the estimate \eqref{es2}) that
\begin{equation}\label{ap0}|(\tilde{W}_{j,n},W)|\le Cn^{-\frac 1 2}\|W\|,\qquad W\in\calW_{j-1,n}.\end{equation}

Now let
$\pi_{j-1,n}:L^2\to \calW_{j-1,n}$
be the orthogonal projection and put
$W_{j,n}^*=\tilde{W}_{j,n}-\pi_{j-1,n}(\tilde{W}_{j,n}).$
Then $\|\tilde{W}_{j,n}-W_{j,n}^*\|=\|\pi_{j-1,n}(\tilde{W}_{j,n})\|=O(n^{-\frac 1 2})$ by \eqref{ap0}, and so
$$\|W_{j,n}^*\|=1+O(\eps_n)\qquad \text{and} \qquad \|W_{j,n}^*-W_{j,n}^\sharp\|=O(\eps_n)$$
by \eqref{no0} and \eqref{es3}.

Moreover, since $W_{j,n}^*\in \calW_{j,n}\ominus \calW_{j-1,n}=\spann\{W_{j,n}\}$, we can write
$W_{j,n}^*=c_{j,n}W_{j,n}$
for some constant $c_{j,n}$, which we can assume is positive. Since $\|W_{j,n}\|=1$ we then have $c_{j,n}=1+O(\eps_n)$. It follows that
\begin{align*}\|W_{j,n}-W_{j,n}^\sharp\|\le |1-c_{j,n}|+\|W_{j,n}^*-W_{j,n}^\sharp\|=O(\eps_n),\end{align*}
and the proof of the lemma is complete.
\end{proof}

Following a well-known circle of ideas we shall now turn the $L^2$-estimate in Lemma \ref{l2} into a pointwise one.

\begin{lem} \label{pl2} Suppose $\tau_0\le \tau(j)\le 1$ and that $u$ is a smooth function on $\C$ which is holomorphic in $\C\setminus K$ with $|u(z)|\lesssim |z|^{j}$ as $z\to\infty$. Consider the weighted analytic function $W=u\cdot e^{-\frac 12nQ}$ on $\C\setminus K$. Then there exists a constant $C$ such that
$$|W(z)|\le C\sqrt{n}\|W\|e^{-\frac 1 2 n(Q-\check{Q}_{\tau(j)})(z)},\qquad z\in U_{\tau_0}.$$
\end{lem}

\begin{proof} We shall slightly modify our proof of Lemma \ref{ho0}.
Write $\tau=\tau(j)$. We begin by recording the basic estimate
\begin{equation}\label{reb}|W(z)|\le M_\tau \cdot e^{-\frac 1 2 n(Q-\check{Q}_\tau)(z)},\qquad z\in U_{\tau_0},\qquad (M_\tau=\sup_{U_{\tau_0}\setminus U_\tau} |W|).
\end{equation}

In order to verify \eqref{reb}, we may assume that $M_\tau\le 1$. The estimate is trivial if $z\not\in U_\tau$ so we may also assume that $z\in U_\tau$.

We then form the function
$$s(z)=\frac 1 n \log|u(z)|^2=\frac 1 n\log |W(z)|^2+Q(z),\quad (z\in \Cl U_\tau).$$

By assumption, this function is subharmonic on $\C\setminus K$ and satisfies $s\le Q$ on $\Gamma_\tau$. Moreover,
we know that $s(z)\le 2\tau\log|z|+O(1)$ as $z\to\infty$. Hence $s\le \check{Q}_\tau$ on $\C\setminus U_\tau$ by the strong version of the maximum principle.

Next pick $n_0$ such that $Q$ is smooth in a neighbourhood of $(\Cl U_{\tau_0})\setminus U_\tau$.
Also fix an arbitrary point $w\in (\Cl U_{\tau_0})\setminus U_\tau$.

By \cite[Lemma 2.4 and its proof]{A}, we have
\begin{equation}\label{cheezy}|W(w)|\le C\sqrt{n}(\int_{D(w,1/\sqrt{n})}|W|^2)^{\frac 12}\le C\sqrt{n}\|W\|,\qquad (w\in (\Cl U_{\tau_0})\setminus U_\tau),\end{equation}
where $C$ is independent of $w$. (Indeed, $C$ depends only on the maximum of the Laplacian $\Delta Q$ over a slightly enlarged set.)

The lemma is immediate on combining \eqref{reb} and \eqref{cheezy}.
\end{proof}

\subsection{Proof of Lemma \ref{spock}}  Fix a number $\beta\in (0,\tfrac 1 4)$ and
suppose that $z\in N(U_{\tau(j)},\delta_n)$, where $\tau(j)$ is in the interval $[\tau_0,1]$.

By (a straightforward generalization of) the inequality \eqref{66}
we have the estimate
\begin{equation}\label{cm2}(\check{Q}_{\tau(j)}-V_{\tau(j)})(z)\le C\delta_n^2= CM^2\frac {\log\log n} n.\end{equation}

Applying Lemma \ref{pl2} and Lemma \ref{l2} with $W=W_{j,n}-W_{j,n}^\sharp$ we obtain
\begin{align}\nonumber |W_{j,n}(z)-W_{j,n}^\sharp(z)|&\le C_1\sqrt{n}\, \|W_{j,n}-W_{j,n}^\sharp\|\, e^{-\frac  12 n(Q-\check{Q}_{\tau(j)})(z)}\\
\label{no1} &\le C_1(\log n)^{CM^2+1}\, e^{-\frac 1 2 n(Q-V_{\tau(j)})(z)}.
\end{align}

But by definition of $F_{j,n}$ it is clear that
\begin{equation}\label{no2}|W_{j,n}^\sharp(z)|=(\frac n {2\pi})^{\frac 1 4}|\sqrt{\phi_{\tau(j)}'}(z)|e^{-\frac 12n(Q-V_{\tau(j)})(z)}e^{\re\calH_{\tau(j)}(z)}.\end{equation}

Since $\beta<\frac 14$, $n^{\frac 1 4}$ outgrows
$n^{\beta}(\log n)^{CM^2+1}$ as $n\to\infty$, so it follows from \eqref{no1} and \eqref{no2} that
$$W_{j,n}(z)=W_{j,n}^\sharp(z)\cdot (1+O(n^{-\beta})).$$

Our proof of Lemma \ref{spock} is complete. \hfill $\qed$

\section{The loop equation and complete integrability} \label{soop}

In this section we view the Berezin measures as exact solutions to the loop equation  and we briefly discuss the
imposed integrable structure on the coefficients in the corresponding large $n$ expansion of the one-point function.

An advantage of the loop equation point of view is that it
continues to hold in a context of $\beta$-ensembles, thereby making it potentially useful for the study of the Hall effect, freezing problems and related issues of
interest in contemporary mathematical physics.

We shall not attempt a profound analysis here; we will merely point out how the
loop equation fits in with some of our work in the previous sections. More about the use of loop equations and large $n$-expansions can be found in
the papers \cite{A,AM,AKM,AR,BBNY2,CFTW,CSA,F,HW2,La20,ZW} and the references there.

\subsection{Gaussian approximation of harmonic measure as a solution to the loop equation}
Let $K_n(z,w)$ be the reproducing kernel with respect to an admissible potential $Q$.
We will write
$R_n(z)=K_n(z,z)$ for the $1$-point function and $R_{n,k}(w_1,\ldots,w_k)=\det(K_n(w_i,w_j))_{k\times k}$ for the $k$-point function of the
determinantal Coulomb gas process $\{z_j\}_1^n$ associated with $Q$.

By Theorem \ref{berth1} we know that if $z$ is in the exterior domain $U$, then the Berezin measure $\mu_{n,z}$ obeys the asymptotic
\begin{equation}\label{asy1}d\mu_{n,z}(p+\ell\,\normal_1(p))=(d\omega_z(p)\times d\gamma_{n,p}(\ell))\cdot (1+o(1)),\end{equation}
where $\omega_z$ and $\gamma_{n,p}$ denote certain harmonic and Gaussian measures, respectively.

As we shall see (whether or not $z$ is in the exterior) $\mu_{n,z}$ is an exact solution to the \textit{loop equation}
\begin{equation}\label{loop}\frac \d {\d \bar{z}}(\mu_{n,z}(k_z))=R_n(z)-n\Delta Q(z)-\Delta \log R_n(z),\end{equation}
where $k_z(w)$ is the Cauchy kernel
$$k_z(w)=\frac 1 {z-w}.$$
(We remind that $\mu(f)$ is short for $\int f\, d\mu$.)

The relation \eqref{loop} is not the ``usual'' form of the two-dimensional loop equation (e.g. \cite{AM}), but rather a kind of infinitesimal variant; for completeness we include a derivation
of it below.

\subsection{$\beta$-ensembles} Given a large $n$ and a configuration $\{z_j\}_1^n$
we consider the Hamiltonian
$$H_n=\sum_{j\ne k}^n\log \frac 1 {|z_j-z_k|}+n\sum_{j=1}^n Q(z_j).$$

The Boltzmann-Gibbs law in external potential $Q$ and inverse temperature $\beta$ is the following probability law on $\C^n$,
\begin{equation}\label{bg1}d\Prob_n^{\,\beta}=\frac 1 {Z_n^\beta}e^{-\beta\cdot H_n}\, dA_n.\end{equation}

Suppose that $\{z_j\}_1^n$ is picked randomly with respect to \eqref{bg1}.
For fixed $k\le n$ we denote by
$R_{n,k}^\beta$ the $k$-point function, i.e.,
the unique (continuous) function on $\C^k$ obeying
$$\Expect_n^\beta(f(z_1,\ldots,z_k))=\frac {(n-k)!}{n!}\int_{\C^k}fR_{n,k}^\beta\, dA_k$$
for each bounded Borel function $f$ on $\C^k$.

We shall also use the connected 2-point function
$R_{n,2}^{\beta,(c)}$, which is defined by
$$R_{n,2}^{\beta,(c)}(z,w)=R_{n,2}^\beta(z,w)-R_{n,1}^\beta(z)R_{n,1}^\beta(w).$$

Finally, we introduce the \textit{Berezin kernel} $B_n^\beta(z,w)$ and the \textit{Berezin measure} $\mu_{n,z}^\beta$ by
\begin{align*}B_n^\beta(z,w)=-\frac {R_{n,2}^{\beta,(c)}(z,w)}{R_n^\beta(z)},\qquad
d\mu_{n,z}^{\,\beta}(w)=B_n^\beta(z,w)\, dA(w).\end{align*}

In the following we denote by $R_n^\beta=R_{n,1}^\beta$ the 1-point function.

\subsection{Proof of the loop equation} We have the following variant of the loop equation.
(See e.g.~\cite{AKM,BBNY2,CSA,La20,ZW} and references for related identities).

\begin{prop} \label{wdid} If $Q$ is $C^2$-smooth in a neighbourhood of a point $z$, then
$$\frac \d {\d\bar{z}} (\mu_{n,z}^{\,\beta}(k_z))=R_n^\beta(z)-n\Delta Q(z)-\frac 1 \beta\Delta \log R_n^\beta(z).$$
\end{prop}

\begin{proof} Let $\Lambda\subset\C$ be an open set such that $Q$ is $C^2$-smooth in a neighbourhood of the closure $\Cl\Lambda$. Fix a point $z\in\Lambda$ and a
smooth real-valued function $\psi$ supported in $\Lambda$.

Given a random sample $\{z_j\}_1^n$, we can for each $j$ view the number $\psi(z_j)$ as a random variable with respect to \eqref{bg1}.

We use integration by parts to see that, for each fixed $j$
\begin{align*}\Expect_n^\beta[\d\psi(z_j)]&=\frac 1 {Z_n^\beta}\int_{\C^n}\d\psi(z_j)\cdot e^{-\beta\cdot H_n(z_1,\ldots,z_n)}\, dA_n(z_1,\ldots,z_n)\\
&=-\frac 1 {Z_n^\beta}\int_{\C^n}\psi(z_j)\cdot \frac {\d} {\d z_j}(e^{-\beta\cdot H_n(z_1,\ldots,z_n)})\,dA_n(z_1,\ldots,z_n)\\
&=\beta\frac 1 {Z_n^\beta}\int_{\C^n}\psi(z_j)\cdot \frac {\d H_n}{\d z_j}\cdot e^{-\beta\cdot H_n}\, dA_n\\
&=\beta\cdot\Expect_n^\beta[\d_j H_n(z_1,\ldots,z_n)\cdot \psi(z_j)].
\end{align*}
(In the last expression, $\d_j$ is short for $\d/\d z_j$.)

Now observe that for each $j$
$$\d_j H_n=n [\d Q](z_j)-\sum_{k\ne j}\frac 1 {z_j-z_k}.$$

Hence a summation in $j$ gives
$$\frac 1 n\sum_{j=1}^n\Expect_n^\beta[\d \psi(z_j)]=\beta\cdot\Expect_n^\beta\left[ \sum_{j=1}^n \psi(z_j)(\d Q(z_j)-\frac 1 n\sum_{k\ne j}\frac 1 {z_j-z_k})\right].$$

We have shown that
\begin{equation}\label{ward}\Expect_n^\beta[W_n^+[\psi]]=0\end{equation}
where $W_n^+[\psi]$ is the ``Ward's tensor''
$$W_n^+[\psi]=\sum_{j=1}^n\d\psi(z_j)-\beta n\sum_{j=1}^n \psi(z_j)\cdot \d Q(z_j)+\frac \beta 2 \sum_{j\ne k}
\frac {\psi(z_j)-\psi(z_k)}{z_j-z_k}.$$

The identity \eqref{ward} is what is called ``Ward's identity'' in papers such as \cite{AKM}. To deduce the infinitesimal version in Proposition \ref{wdid}
we proceed as follows.

By the definition of 1-point function and an integration by parts we have
$$\Expect_n^\beta[\sum_{j=1}^n\d\psi(z_j)]=\int_\C \d \psi(z)R_{n,1}^\beta(z)\, dA(z)=-\int_\C \psi\cdot \d R_{n,1}^\beta\, dA.$$
Also
$$\Expect_n^\beta[\sum_{j=1}^n \psi(z_j)\d Q(z_j)]=\int_\C \psi\cdot\d Q\cdot R_{n,1}^\beta\, dA,$$
and
\begin{align*}\Expect_n^\beta \left[\frac 1 2 \sum_{j\ne k}\frac {\psi(z_j)-\psi(z_k)}{z_j-z_k}\right]&=
\frac 1 2 \int_{\C^2}\frac {\psi(z)-\psi(w)}{z-w}\,R_{n,2}^\beta(z,w)\, dA_2(z,w)\\
&=\int_\C \psi(z)\cdot R_{n,1}^\beta(z)\, dA(z)\int_\C\frac 1 {z-w}\frac {R_{n,2}^\beta(z,w)}{R_{n,1}^\beta(z)}\, dA(w).
\end{align*}

Since the identity \eqref{ward} holds for every test-function $\psi$, we obtain the pointwise identity for all $z\in\Lambda$,
\begin{align}-\d R_{n,1}^\beta(z)-\beta n\d Q(z)R_{n,1}^\beta(z)+\beta R_{n,1}^\beta(z)\int_\C\frac 1 {z-w}\frac {R_{n,2}^\beta(z,w)}{R_{n,1}^\beta(z)}\, dA(w)=0.
\end{align}
(First we obtain the identity in the sense of distributions on $\Lambda$, then everywhere, since the functions involved are smooth.)

Dividing through by $\beta R_{n,1}^\beta$ we find
$$-\frac 1 \beta \d \log R_{n,1}^\beta(z)-n\d Q(z)+\int_\C\frac 1 {z-w}\frac {R_{n,2}^\beta(z,w)}{R_{n,1}^\beta(z)}\, dA(w)=0.$$

Recalling that
$$B_n^\beta(z,w)=R_{n,1}^\beta(w)-\frac {R_{n,2}^\beta(z,w)}{R_n^1(z)},$$
we obtain
\begin{align*}\label{last}-\frac 1 \beta \d \log R_{n,1}^\beta(z)-n\d Q(z)+\int_\C\frac {R_{n,1}^\beta(w)} {z-w}\, dA(w)-\mu_{n,z}^{\,\beta}(k_z)=0.\end{align*}
Taking $\dbar$-derivatives with respect to $z$ in the last identity, we finish the proof.
\end{proof}

\subsection{On various asymptotic relations} \label{genem}
Now set $\beta=1$.
We have the following theorem on the Cauchy transform $\mu_{n,z}(k_z)$ of the Berezin measure $\mu_{n,z}$. (The symbol $k_z$ denotes
the Cauchy kernel $k_z(w)=(z-w)^{-1}$, and $\omega_z$ denotes the harmonic measure of $U$ evaluated at a point $z\in U$.)

\begin{thm} \label{genthm}
With $\phi:U\to\D_e$ the normalized conformal map, we have the identity
\begin{equation}\label{caut}\lim_{n\to\infty}\mu_{n,z}(k_z)=\omega_z(k_z)=\frac \d {\d z}\log(|\phi(z)|^2-1)+H(z),\qquad (z\in U)\end{equation}
where $H(z)$ is a holomorphic function in $U$ with $H(z)=O(z^{-2})$ as $z\to\infty$.
\end{thm}

\begin{proof} The first equality in \eqref{caut} is immediate by Theorem \ref{berth1}. (The singularity of $k_z(w)$ at $w=z$ presents no trouble since $B_n(z,w)\le K_n(w,w)$ and $K_n(w,w)$ converges to
zero uniformly for $w$ outside of any given neighbourhood of $(\Int\Gamma)\cup\Gamma$, see for example \cite[Theorem 1]{A}.)

In order to prove the remaining equality, we assume for simplicity that $Q$ is $C^2$-smooth throughout $U\setminus\{\infty\}$; the extension to more general admissible potentials may be left to the reader.

We shall use Theorem \ref{kern1}, which implies that for $z\in U$,
$$R_n(z)=\frac {\sqrt{n}}{\sqrt{2\pi}}e^{n\re\calQ(z)}e^{-nQ(z)}e^{\re\calH(z)}|\phi'(z)||\phi(z)|^{2n}\frac 1 {|\phi(z)|^2-1}\cdot (1+O(n^{-\beta})),\qquad (n\to\infty)$$
where $0<\beta<\frac 14$. This implies that $R_n(z)$ is negligible for large $n$ whereas
\begin{align*}\Delta \log R_n(z)&=-n\Delta Q(z)-\Delta\log (|\phi(z)|^2-1)+O(n^{-\beta})\\
&=-n\Delta Q(z)+\frac {|\phi'(z)|^2}{(|\phi(z)|^2-1)^2}+O(n^{-\beta}).
\end{align*}

It follows that
$$R_n(z)-n\Delta Q(z)-\Delta\log R_n(z)=-\frac {|\phi'(z)|^2}{(|\phi(z)|^2-1)^2}+O(n^{-\beta}),$$
and hence by Proposition \ref{wdid}, we have with locally uniform convergence
\begin{equation}\label{diffl}\lim_{n\to\infty}\frac \d {\d\bar{z}}(\mu_{n,z}(k_z))=-\frac{|\phi'(z)|^2}{(|\phi(z)|^2-1)^2}=\Delta_z\log(|\phi(z)|^2-1),\qquad(z\in U).\end{equation}
Hence the function $H(z)=\lim\limits_{n\to\infty}\mu_{n,z}(k_z)-\frac \d {\d z}\log(|\phi(z)|^2-1)$ is holomorphic in $U$, and since (for each $n$)
$$\lim_{z\to\infty}z\mu_{n,z}(k_z)=\mu_{n,z}(1)=1=\lim_{z\to\infty}\frac {z\phi'(z)\overline{\phi(z)}}{|\phi(z)|^2-1},$$
we see that $H(z)=O(z^{-2})$ as $z\to\infty$.
\end{proof}

\begin{ex} Suppose that $Q(z)$ is radially symmetric. Then $U=\D_e(r)$ is an exterior disc, and the conformal map $\phi:U\to \D_e$ is just $\phi(z)=z/r$;
the harmonic measure is $dP_z(\theta)\, d\theta$ where $P_z(\theta)=\frac 1 {2\pi}\frac {|z|^2-r^2}{|z-re^{i\theta}|^2}$ for $|z|>r$.
A standard computation gives that $\omega_z(k_z)=\frac {\bar{z}}{|z|^2-r^2}=\frac \d{\d z}\log(|\phi(z)|^2-1)$. Hence $H$ in \eqref{caut} vanishes identically if $Q$ is radially
symmetric. On the other hand, for non-symmetric potentials, $H(z)$ seems typically to be nontrivial.
\end{ex}


For the Ginibre ensemble we have the following result.

\begin{thm} \label{Ginthm} When $Q(z)=|z|^2$, we have the asymptotic expansion
\begin{equation}\label{buckla}\mu_{n,z}(k_z)=\frac {\bar{z}} {|z|^2-1}-\frac 1 n\frac {\bar{z}(|z|^2+1)}{(|z|^2-1)^3}+O(n^{-2}),\qquad (z\in\D_e).
\end{equation}
\end{thm}

Before proving the theorem, we remark that we can at this point easily prove a differentiated form of \eqref{buckla}.
Namely, by Theorem \ref{tue}, we know that
$$R_n(z)=\frac {\sqrt{n}}{\sqrt{2\pi}}\frac {|z|^{2n}}{|z|^2-1}e^{n-n|z|^2}\cdot (1+\frac 1 n \rho_1(|z|^2)+\cdots),\qquad (|z|>1).$$
where $\rho_1(\zeta)=-\frac 1 {12}-\frac \zeta {(\zeta-1)^2}$.
Passing to logarithms and differentiating we see that, for $|z|>1$,
\begin{align*}\Delta \log R_n(z)&=-n-\Delta\log(|z|^2-1)+\Delta \log\left[1-\frac 1 n \rho_1(|z|^2)+\cdots\right]\\
&=-n+\frac 1 {(|z|^2-1)^2}+\frac 1 n(\rho_1'(|z|^2)+|z|^2\rho_1''(|z|^2))+\cdots.
\end{align*}
Since $R_n(z)$ is negligible for $|z|>1$ while $\Delta Q=1$, we find (after some computation) that the right hand side in Proposition \ref{wdid} equals to
\begin{equation}\label{sasplit}\begin{split}R_n(z)-n-\Delta\log R_n(z)&=-\frac 1 {(|z|^2-1)^2}-\frac 1 n(\rho_1'(|z|^2)+|z|^2\rho_1''(|z|^2))+O(n^{-2})\cr
&=-\frac 1 {(|z|^2-1)^2}+\frac 1 n\frac {|z|^4+4|z|^2+1}{(|z|^2-1)^4}+\cdots.\cr
\end{split}
\end{equation}
One checks readily that the right hand side in \eqref{sasplit} equals to the $\d/\d\bar{z}$-derivative of \eqref{buckla}.

\begin{proof}[Proof of Theorem \ref{Ginthm} (Sketch)]
For fixed $z\in\D_e$ we write
$k_z(w)=(z-w)^{-1}$
for the Cauchy-kernel and $d\omega_z(\theta)=P_z(\theta)\, d\theta$ for the harmonic measure of $U=\D_e$ evaluated at $z$. Here of course
$P_z(\theta)=\frac 1 {2\pi}\frac {|z|^2-1}{|z-e^{i\theta}|^2}$ is the exterior Poisson kernel.

By Theorem \ref{genthm} we already know that
$$\mu_{n,z}(k_z)=\frac {\bar{z}}{|z|^2-1}+o(1),\qquad (n\to\infty).$$



To find the $O(1/n)$-term in \eqref{buckla}, we fix $z\in\D_e$ and use the approximation provided by Theorem \ref{tue},
$$B_n(z,w)=\sqrt{\frac n {2\pi}}e^{n-n|w|^2}|w|^{2n}\frac {|z|^2-1}{|z\bar{w}-1|^2}\cdot\left[\frac {|1+n^{-1}\rho_1(z\bar{w})|^2}{1+n^{-1}\rho_1(|z|^2)}+O(n^{-2})\right]$$
and work in the coordinate system $(\theta,t)$ where $w=e^{i\theta}\cdot (1+\frac t {2\sqrt{n}})$. A lengthy but straightforward computation based on Taylor's formula, residues,
and the elementary identity $\omega_w(k_z)=\frac {\bar{w}}{z\bar{w}-1}$ produces the asymptotic in \eqref{buckla};
we omit details.
\end{proof}

Beyond the Ginibre ensemble, it is not clear from our above results that there is a similar large $n$-expansion. However, a qualitative result of Hedenmalm and Wennman
comes to the rescue.

\begin{thm} Let $Q$ be a potential satisfying the assumptions in Subsection \ref{backgr}. There is then an asymptotic expansion
\begin{equation}\label{try}R_n(z)=\frac {\sqrt{n}}{\sqrt{2\pi}}e^{n\re\calQ(z)-nQ(z)+\re\calH(z)}\frac {|\phi'(z)||\phi(z)|^{2n}} {|\phi(z)|^2-1}\cdot (1+\frac 1 n \rho_1(z,\bar{z})+\frac 1 {n^2}\rho_2(z,\bar{z})+\cdots)\end{equation}
where $z\in U$ and
$\rho_1,\rho_2,\ldots$ are some unknown correction terms which are subject to the completely integrable system given by Proposition \ref{wdid}.
\end{thm}

\begin{proof} Fix $z$ in the exterior component $U$
and pick $w$ near $\Gamma$ or in $U$. It follows from \cite[Theorem 1.4.1]{HW2} that
there is an asymptotic expansion
\begin{equation}\label{hwek}\frac {K_n(w,z)}{\sqrt{K_n(z,z)}}=n^{\frac 1 4}F_n(w,z)(a_0(w,z)+\frac 1 n a_1(w,z)+\frac 1 {n^2}a_2(w,z)+\cdots),\end{equation}
where $F_n(w,z)$ and the leading term $a_0(w,z)$ are explicitly given in \cite{HW2}.

It follows from \eqref{hwek} that
\begin{equation}\label{bnex}B_n(z,w)=\sqrt{n}|F_n(w,z)|^2(b_0(z,w)+\frac 1 n b_1(z,w)+\frac 1 {n^2}b_2(z,w)+\cdots),\end{equation}
where $|F_n(w,z)|^2$ and $b_0(z,w)=|a_0(w,z)|^2$ are again certain explicit functions, which of course must match up with the expressions found in Theorem \ref{berth1}
for $z\in U$ and $w$ near the boundary $\d U$.

In the expansion \eqref{hwek}, the points $z$ and $w$ play highly asymmetric roles. Nevertheless,  the form of the expansion \eqref{bnex}, specialized to the diagonal case
when $w=z$ belongs to $U$ shows that the form of \eqref{try} must hold for appropriate correction terms $\rho_j(z,\bar{z})$, which are proportional to $b_j(z,z)/b_0(z,z)$. \end{proof}

By inserting the ansatz \eqref{try} in the loop equation, we get a feed-back relation for the
correction terms $\rho_j(z,\bar{z})$. A deeper analysis of this structure is beyond the scope of our present investigation.

\subsection{Back to $\beta$-ensembles} We finish with a few words about $\beta$-ensembles.
In the case when $z\in U$, it is known due to the localization theorem in
\cite{A} that $R_n^\beta(z)\to 0$ quickly as $n\to\infty$. Hence the right hand side in Ward's identity (Proposition \ref{wdid}) is
$$R_n^\beta(z)-n\Delta Q(z)-\frac 1 \beta\Delta \log R_n^\beta(z)=-n\Delta Q(z)-\frac 1 \beta\Delta \log R_n^\beta(z)+o(1).$$
(We assume here that $Q$ is smooth at $z$.)

The left hand side in Ward's identity is not known, but it seems plausible that we should have $\frac \d {\d\bar{z}}(\mu_{n,z}^\beta(k_z))=O(1)$ when $z\in U$.
Assuming that this is the case, and comparing $O(n)$-terms in Ward's identity we find ``heuristically'' the approximation
\begin{equation}\label{dots}\Delta\log R_n^\beta(z)=-n\beta\Delta Q(z)+\cdots,\end{equation}
where the dots represent terms of lower order in $n$.

When $z\in U$ is close to the boundary $\d U$,
\eqref{dots} is consistent with predictions found in \cite{CFTW}, and also with the
localization theorem in \cite{A}.

The papers \cite{AR,CFTW,CSA} and the references there provide more information about the  problem of finding asymptotics for the 1-point function $R_n^\beta$ when $\beta>1$.

\subsection{A glance at disconnected droplets}\label{arch}
We now briefly touch on the case of disconnected droplets, where the condition \ref{1st} in the definition of an admissible potential is replaced by
\begin{enumerate}
\item[(1')] \label{1sp} $Q$ is $C^2$-smooth on $S$ and real analytic in a neighbourhood of the outer boundary $\Gamma=\d U$.
\end{enumerate}

We start by noting that if we replace assumption \ref{1st} by (1') in our definition of
admissible potential, then the existence of a local Schwarz function $\calS$ at each point $p\in\Gamma$ can be established precisely as in the proof of Lemma \ref{schw}.
It follows by Sakai's main result in \cite{Sa} that $\C\setminus U$ has finitely many components $K_1,\ldots,K_d$, and that the normalized Riemann maps
$\chi_l:\D_e\to \hat{\C}\setminus K_l$ can be continued analytically across $\T$ for $l=1,\ldots,d$. We put $\Gamma^l=\chi_l(\T)=\d K_l$ and assume that $\chi_l'\ne 0$ on $\T$ for $l=1,\ldots,n$.
Then each $\Gamma^l$ is an analytic, non-singular Jordan curve.

For a basic model case we consider the disconnected lemniscate droplet, defined by the potential
\begin{equation}\label{qds}Q(z)=\frac 1 d|z^d-d|^2,\end{equation}
where $d\ge 2$ is an integer.




By \cite[Lemma 1]{BM}, the $\tau$-droplet (i.e. droplet in potential $Q/\tau$) is
$$S_\tau=\{z\in\C\,;\, |z^d-d|^2\le \tau\}.$$
See Figure \ref{FigLem0}.

We write $S=S_1$ and note that the
the equilibrium measure is
$d\sigma(z)=d|z|^{2(d-1)}\1_S(z)\, dA(z).$
\begin{figure}[ht]
\begin{center}
\includegraphics[width=0.3\textwidth]{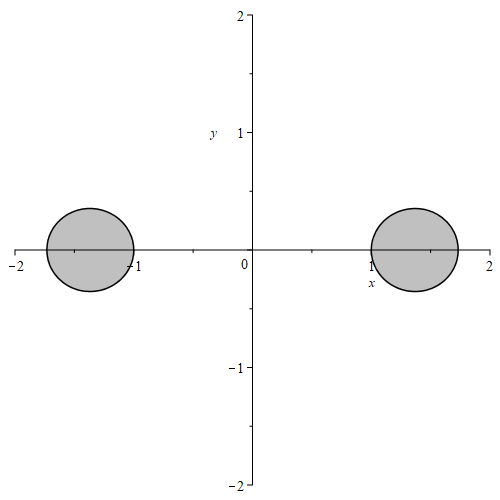}\hfil
\includegraphics[width=0.3\textwidth]{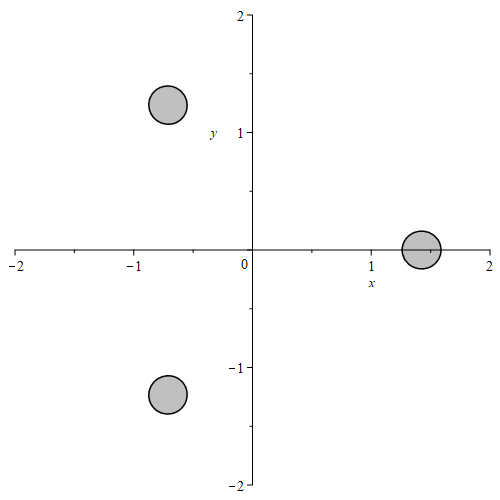}
\end{center}
\caption{Droplets (for $\tau=1$) of $Q(z)=\frac 1 d |z^d-d|^2$ for $d=2$ and for $d=3$.}
\label{FigLem0}
\end{figure}

The following problem presents itself: to find the asymptotics of $K_n(z,w)$ when $z$ and $w$ belong to different boundary components of $\d S$. Since $U$ is not simply connected, there is no longer a Riemann map $\phi$, and as a consequence our technique using quasipolynomial approximation will not work, at least not without substantial changes. Fortunately, in the special case of the potential \eqref{qds}, approximate orthogonal polynomials can be found using Riemann-Hilbert techniques following the works \cite{BBLM,BGM,BLY,LY}. This was used to generate Figure \ref{Figco}. 

The recent work \cite{C} studies other types of ensembles with disconnected droplets, with ``hard edges'', in which the droplet consists of several concentric annuli. Among other things it is shown that a Jacobi theta function emerges when studying certain associated gap-probabilities. (More generally, theta functions are known to emerge in various more or less related contexts, see \cite{BFS} and the  references there.) The present setting of ``soft edge'' ensembles with disconnected droplets is the topic of our forthcoming work \cite{ACC}.

\begin{figure}[ht]
\begin{center}
\includegraphics[width=0.4\textwidth]{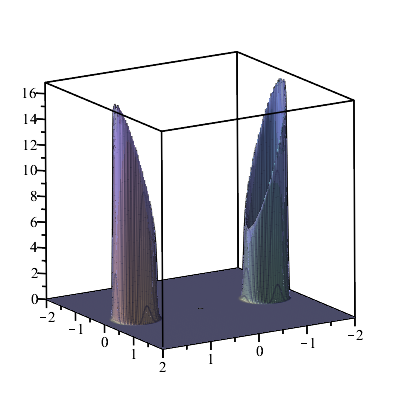}
\includegraphics[width=0.4\textwidth]{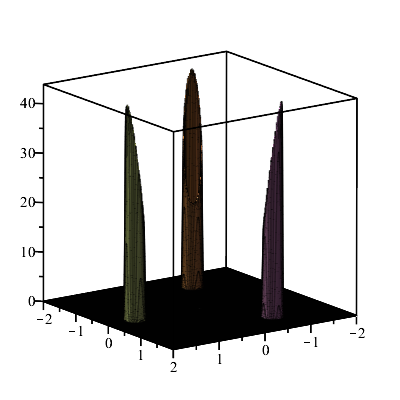}
\end{center}
\caption{Graphs of approximate orthogonal polynomials $|W_{n,n}^\sharp(z)|^2$ for $n=100$ and $d=2,3$, obtained using Riemann-Hilbert techniques.}
\label{Figco}
\end{figure}

Other types of open problems enter in the case 
when the boundary of the droplet has one or several singular points (cusps, double points, or lemniscate-type singularities which are found at boundary points where $\Delta Q=0$).
More background on singular points can be found e.g.~in \cite{AKMW} and the references there.


\begin{thebibliography}{999}
\bibitem{Ahl} Ahlfors, L. V., \textit{Complex Analysis}, Third Edition, McGraw Hill 1979.
\bibitem{ABD} Akemann, G., Baik, J., Di Francesco, P. (Eds.), \textit{The Oxford Handbook of Random Matrix Theory}, Oxford 2011.
\bibitem{ACV} Akemann, G., Cikovic, M., Venker, M., \textit{Universality at weak and strong non-Hermiticity beyond the
elliptic Ginibre ensemble}, Comm. Math. Phys. \textbf{362} (2018), 1111–1141.
\bibitem{ADM} Akemann, G., Duits, M., Molag, L., \textit{The Elliptic Ginibre Ensemble: A Unifying Approach to Local and Global Statistics
for Higher Dimensions}, arxiv preprint 2203.00287.
\bibitem{A} Ameur, Y., \textit{A localization theorem for the planar Coulomb gas in an external field}. Electron. J. Probab. \textbf{26} (2021), article no. 46.
\bibitem{A2} Ameur, Y., \textit{Near-boundary asymptotics of correlation kernels},
J. Geom. Anal. \textbf{23} (2013), 73--95.
\bibitem{AB} Ameur, Y., Byun, S.-S., \textit{Almost-Hermitian random matrices and bandlimited point processes}, arxiv: 2101.03832.
\bibitem{ACC} Ameur, Y., Charlier, C., Cronvall, J., \textit{The two-dimensional Coulomb gas: fluctuations through a spectral gap}, To Appear.
\bibitem{AHM1} Ameur, Y., Hedenmalm, H., Makarov, N., \textit{Berezin transform in polynomial Bergman spaces}, Comm. Pure Appl. Math. \textbf{63} (2010),  1533-1584.
\bibitem{AHM2} Ameur, Y., Hedenmalm, H., Makarov, N., \textit{Fluctuations of eigenvalues of random normal matrices}, Duke J. Math. \textbf{159} (2011), 1533-1584.
\bibitem{AM} Ameur, Y., Hedenmalm, H., Makarov, N., \textit{Ward identities and random normal matrices}, Ann. Probab. \textbf{43} (2015), 1157--1201.
\bibitem{AKM} Ameur, Y., Kang, N.-G., Makarov, N., \textit{Rescaling Ward identities in the random normal matrix model},
   Constr. Approx. \textbf{50} (2019), 63--127.
\bibitem{AKMW} Ameur, Y., Kang, N.-G., Makarov, N., Wennman, A., \textit{Scaling limits of random normal matrix processes at singular boundary points}, J. Funct. Anal. \textbf{278} (2020), 108340.
\bibitem{AKS} Ameur, Y., Kang, N.-G., Seo, S.-M., \textit{On boundary confinements for the Coulomb gas}, Anal. Math. Phys. \textbf{10}, paper no. 68 (2020).
\bibitem{AR} Ameur, Y., Romero, J.-L., \textit{The planar low temperature Coulomb gas: separation and
equidistribution}, Rev. Mat. Iberoam. (2022) DOI 10.4171/RMI/1340
\bibitem{Ar} Aronszajn, N., \textit{Theory of reproducing kernels}, Trans. Amer. Math. Soc. \textbf{68} (1950), 337-404.
\bibitem{B} Bai, Z. D., \textit{Circular law}, Ann. Probab. \textbf{25} (1997), 494-529.
\bibitem{BBLM} Balogh, F., Bertola, M., Lee, S.-Y., Mclaughlin, K. D., T-R, \textit{Strong asymptotics of the orthogonal polynomials with respect to a measure supported on the plane},
Comm. Pure Appl. Math. \textbf{68} (2015), 112-172.
\bibitem{BGM} Balogh, F., Grava, T., Merzi, D., \textit{Orthogonal polynomials for a class of measures with discrete rotational symmetries in the complex plane}, Constr. Approx. \textbf{46} (2017), 109-169.
\bibitem{BMe} Balogh, F., Merzi, D., \textit{Equilibrium Measures for a Class of Potentials
with Discrete Rotational Symmetries}, Constr. Approx. \textbf{42} (2015), 399-424.
\bibitem{Bar} Barker, W.H. II, \textit{Kernel functions on domains with hyperelliptic double}, Trans. Amer. Math. Soc. \textbf{231} (1977), 339-347.
\bibitem{BBNY2} Bauerschmidt, R., Bourgade, P., Nikula, M., Yau, H.-T., \textit{The two-dimensional Coulomb plasma: quasi-free approximation and central limit theorem}, Adv. Theor. Math. Phys. \textbf{23}, 841-1002, (2019).
\bibitem{Bell} Bell, S.R., \textit{The Cauchy Transform, Potential Theory and Conformal Mapping},
Chapman \& Hall 2016.
\bibitem{BERG} Bertola, M., Elias Rebelo, J. G., Grava, T., \textit{Painlev\'{e} IV Critical Asymptotics
for Orthogonal Polynomials in the Complex Plane}, SIGMA \textbf{14} (2018).
\bibitem{BFS} B\'{e}termin, L., Faulhuber, M., Steinerberger, S., \textit{A variational principle for Gaussian lattice sums}, arxiv 2110.06008.
\bibitem{BM} Bleher, P., Mallison, R. Jr., \textit{Zero sections of exponential sums}, Int. Math. Res. Not. IMRN Art. ID 38937 (2006),
49 pp.
\bibitem{BG} Boyer, R., Goh, W., \textit{On the zero attractor of the Euler polynomials}, Adv. in Appl. Math. \textbf{38} (2007), 97-132.
\bibitem{BGNW} Butez, R., Garc\'{i}a-Zelada, D., Nishry, A., Wennman, A., \textit{Universality for outliers in weakly confined
Coulomb-type systems}, arxiv 2104.03959.
\bibitem{BLY} Byun, S.-S., Lee, S.-Y., Yang, M., \textit{Lemniscate ensembles with spectral singularity}, arxiv: 2107.0722.
\bibitem{CFTW} Can, T., Forrester, P.J., T\'{e}llez, G., Wiegmann, P., \textit{Singular behavior at the edge of Laughlin states.} Phys.
Rev. B \textbf{89}, 235137 (2014).
\bibitem{CSA} Cardoso, G., St\'{e}phan, J.-M., Abanov, A., \textit{The boundary density profile of a Coulomb droplet. Freezing at the edge}, J. Phys. A.: Math. Theor. \textbf{54}(1), (2021), 015002.
\bibitem{C} Charlier, C., \textit{Large gap asymptotics on annuli
in the random normal matrix model
}, Arxiv 2110.06908.
\bibitem{DS} Dea\~{n}o, A., Simm, N.J., \textit{Characteristic polynomials of complex random matrices and Painlev\'{e} transcendents}, International Mathematics Research Notices IMRN (2020).
\bibitem{DR} Dontchev, A. L., Rockafellar, R. T., \textit{Implicit Functions and Solution Mappings}, Springer 2009.
\bibitem{DRR} Dubail, J., Read, N., Rezayi, E. H.,
\textit{Edge-state inner products and real-space entanglement spectrum of trial quantum Hall states}, Phys. Rev. B \textbf{86}, 245310 (2012).
\bibitem{D} Duren, P., \textit{Theory of $H^p$-spaces}, Dover 2000.
\bibitem{ESV} Edrei, A., Saff, E.B., Varga, R.S.: \textit{Zeros of Sections of Power Series}, Volume 1002 of Lecture Notes
in Mathematics. Springer, Berlin (1983)
\bibitem{ES} Estienne, B., St\'{e}phan, J.-M., \textit{Entanglement spectroscopy of chiral edge modes in the Quantum Hall effect}, Phys. Rev. B \textbf{101}, 115136 (2020).
\bibitem{F} Forrester, P.J., \textit{A review of exact results for fluctuation formulas in random matrix theory}, arxiv 2204.03303.

\bibitem{Fo} Forrester P.J., \textit{Log-gases and Random Matrices} (LMS-34), Princeton University Press, Princeton 2010.

\bibitem{FH} Forrester, P.J., Honner, G., \textit{Exact statistical properties of the zeros of complex random polynomials}, J. Phys. A. \textbf{41}, 375003 (1999).
\bibitem{FJ} Forrester, P.J., Jancovici, B., \textit{Two-dimensional one-component plasma
in a quadrupolar field}, International Journal of Modern Physics A \textbf{11}, no. 5 (1996).
\bibitem{Gar} Garabedian, P.R., \textit{Schwarz's lemma and the Szeg\H{o} kernel function}, Trans. Amer. Math. Soc. \textbf{67} (1949), 1-35.
\bibitem{GM} Garnett, J. B., Marshall, D. E., \textit{Harmonic measure}, Cambridge 2005.
\bibitem{G} Ginibre, J., \textit{Statistical ensembles of complex, quaternion, and real matrices}, J. Math. Phys. \textbf{6} (1965), 440-449.
\bibitem{GOC} Gr\"{o}chenig, K., Ortega-Cerd\`{a}, J., \textit{Marcinkiewicz-Zygmund inequalities for polynomials in Fock space}, arxiv 2019.11852 (2021).
\bibitem{GPSS} Gustafsson, B., Putinar, M., Saff, E. B., Stylianopolous, N., \textit{Bergman polynomials on an archipelago: estimates, zeros and shape reconstruction}, Adv. Math. \textbf{222} (2009),
1405-1460.
\bibitem{GTV} Gustafsson, B., Teodorescu, R., Vasil'ev, A., \textit{Classical and stochastic Laplacian growth}, Birkh\"{a}user 2014.
\bibitem{He} Hedenmalm, H., \textit{Soft Riemann-Hilbert problems and planar orthogonal polynomials.} arXiv
 2108.05270.
\bibitem{HH} Haimi, A., Hedenmalm, H., \textit{The polyanalytic Ginibre ensembles}, J. Stat. Phys. \textbf{153} (2013), 10-47.
\bibitem{HS} Hedenmalm, H., Shimorin, S., \textit{Hele-Shaw flow on hyperbolic surfaces}, J. Math. Pures et Appl. \textbf{81} (2002),
187-222.
\bibitem{HW0} Hedenmalm, H., Wennman, A., \textit{A real variable calculus for planar orthogonal polynomials}, Arxiv 2205.15054.
\bibitem{HW2} Hedenmalm, H., Wennman, A., \textit{Off-spectral analysis of Bergman kernels}, Comm. Math. Phys. \textbf{373} (2020), 1049-1083.
\bibitem{HW} Hedenmalm, H., Wennman, A., \textit{Planar orthogonal polynomials and boundary universality in the random normal matrix model}. Acta Math. \textbf{227} (2021), 309-406. 
\bibitem{HW3} Hedenmalm, H., Wennman, A., \textit{Riemann-Hilbert hierarchies for hard edge orthogonal polynomials}, Preprint, Arxiv 2008.02682
\bibitem{HKPV} Hough, J. Ben, Krishnapur, M., Peres, Y., Vir\'{a}g, B., \textit{Zeros of Gaussian analytic functions and determinantal point processes}, University Lecture Series \textbf{51}, AMS 2009.
\bibitem{H} H\"{o}rmander, L., \textit{Notions of convexity}, Birkh\"{a}user 1994.
\bibitem{IT}  Its, A., Takhtajan, L., \textit{Normal matrix models, $\dbar$-problem, and orthogonal polynomials in the complex plane}, arXiv:0708.3867 (2007).
\bibitem{La20} Lambert, G., \textit{Maximum of the characteristic polynomial of the Ginibre ensemble}, Commun. Math. Phys. \textbf{378} (2020), 943--985.
\bibitem{LSe} Lebl\'{e}, T., Serfaty, S., \textit{Fluctuations of two-dimensional Coulomb gases}, Geom. Funct. Anal.
\textbf{28} (2018), 443-508.
\bibitem{LM} Lee, S.-Y., Makarov, N., \textit{Topology of quadrature domains}, J. Amer. Math. Soc. \textbf{29} (2016), 333-369.
\bibitem{LR} Lee, S.-Y., Riser, R., \textit{Fine asymptotic behaviour of random normal matrices: ellipse case}, J. Math. Phys. \textbf{57} (2016), 023302.
\bibitem{LY} Lee, S.-Y., Yang, M., \textit{Discontinuity in the Asymptotic Behavior of Planar
Orthogonal Polynomials Under a Perturbation
of the Gaussian Weight},
Commun. Math. Phys. \textbf{355}, 303-338 (2017).
\bibitem{M} Mehta, M. L., \textit{Random matrices}, Third Edition, Academic Press 2004.
\bibitem{Ne} Nehari, Z., \textit{Conformal mapping} Dover 1975.
\bibitem{ND} Nemes, G., Daalhuis, A.B.O., \textit{Asymptotics for the incomplete gamma function}, Mathematics of Computation \textbf{88} (2018), DOI 10.1090/mcom/3391.
\bibitem{OLBC} Olver, F. W., Lozier, D. W., Boisvert, R. F., Clark, C. W. (Editors), \textit{NIST Handbook of Mathematical Functions}, Cambridge University Press, Cambridge, 2010.
\bibitem{RV} Rider, B., Vir\'{a}g, B., \textit{The noise in the circular law and the Gaussian free field}, Int. Math. Res. Not. IMRN 2007, no. 2, Art. ID rnm006, 33 pp.
\bibitem{ST} Saff, E. B., Totik, V., \textit{Logarithmic potentials with
external fields}, Springer 1997.
\bibitem{Sa} Sakai, M., \textit{Regularity of a boundary having a Schwarz function}, Acta Math. \textbf{166} (1991), 263--297.
\bibitem{S} Shapiro, H., \textit{Unbounded quadrature domains}, in ``Complex Analysis I'', Springer Lecture Notes in Math. \textbf{1275} (1987).
\bibitem{Sz} Szeg\H{o}, G., \textit{\"{U}ber eine eigenschaft der exponentialreihe}, Sitzungsber. Berlin Math. Gessellschaftwiss. \textbf{23} (1924), 50-64.
\bibitem{TV} Tao, T., Vu, V., \textit{Random matrices: Universality of local spectral statistics of non-Hermitian matrices}, Ann. Probab., Vol. 43, no. 2, (2015), 782-874.
\bibitem{T} Temme, M., \textit{Computational aspects of incomplete gamma functions with large complex parameters.} In
R. V. M. Zahar (Ed.), Approximation and Computation.
A Festschrift in Honor of Walter Gautschi., Volume \textbf{119} of
International Series of Numerical Mathematics, pp. 551-
562. Boston, MA: Birkh\"{a}user Boston.
\bibitem{Tr} Tricomi, F. G., \textit{Asymptotische eigenschaften der unvollst\"{a}ndigen gammafunktion}, Math. Z. \textbf{53} (1950), 136-148.
\bibitem{V} Vargas, A. R., \textit{The Saff-Varga Width Conjecture and Entire Functions with Simple Exponential Growth}, Constr. Approx. \textbf{49} (2019), 307-383.
\bibitem{Z} Zabrodin, A., \textit{Random matrices and Laplacian growth}, In The Oxford handbook of random matrix theory, Oxford (2011), 802-823.
\bibitem{ZW} Zabrodin, A., Wiegmann, P., \textit{Large $N$ expansion for the 2D Dyson gas}, J. Phys.
A: Math. Gen. \textbf{39} (2006), 8933-8964.
\end{thebibliography}
\end{document}